\documentclass{amsart}

\usepackage{amssymb,amscd,amsthm,verbatim}
\usepackage{eucal}
\usepackage[dvips]{graphicx}
\usepackage{multirow}
\usepackage{fancyhdr}
\usepackage{mathrsfs}
\usepackage{color}

\newtheorem{proposition}{Proposition}[section]
\newtheorem{theorem}[proposition]{Theorem}

\newtheorem{lemma}[proposition]{Lemma}

\newtheorem{remark}[proposition]{Remark}

\newtheorem{thm}[proposition]{Theorem}

\newtheorem{lem}[proposition]{Lemma}

\newtheorem{assumption}[proposition]{Assumption}

\newcommand{\NN}{N,n}
\newcommand{\la}{\left \langle}
\newcommand{\ra}{\right\rangle}

\newcommand\myeq{\mathrel{\overset{\makebox[0pt]{\mbox{\normalfont\tiny\sffamily def}}}{=}}}

\begin{document}

\title[ Network effects in default clustering for large systems]{ Network effects in default clustering for large systems}

\author{Konstantinos Spiliopoulos and Jia Yang}
\address{Department of Mathematics and Statistics \\
Boston University \\
Boston, MA 02215}
\email[Konstantinos Spiliopoulos]{kspiliop@math.bu.edu}
\email[Jia Yang]{jiayang@bu.edu}
\thanks{The present research was partially supported by the National Science Foundation (DMS 1412529 and DMS 1550918). We would like to thank Kay Giesecke and Paolo Guasoni for discussions on this project. }
\date{\today}

\begin{abstract}
We consider a large collection of dynamically interacting components defined on a weighted directed graph determining the impact of default of one component to another one. We prove a law of large numbers for the empirical measure capturing the evolution of the different components in the pool and from this we extract important information for quantities such as the loss rate in the overall pool as well as the mean impact on a given component from system wide defaults. A singular value decomposition of the adjacency matrix of the graph  allows to coarse-grain the system by focusing on the highest eigenvalues which also correspond to the components with the highest contagion impact on the pool. Numerical simulations demonstrate the theoretical findings.
\end{abstract}

\maketitle
\section{Introduction}

The financial crisis of 2007-2009 made clear to the mathematical finance community that connectedness and  network effects in financial systems need to be better understood and modelled. Risk can propagate through the system and network topology can affect its propagation.

Exogenous risks acting as initial shocks, such  as devaluation of mortgage-backed securities, changes in interest rates or commodity prices cannot fully explain crisis events, but can lead to contagion effects, \cite{Meinerding2012,azizpour-giesecke-schwenkler}. In particular, shocks can lead to spiral events within the system and the topology and connectedness of the system can then affect how these spiral events unfold and propagate.  This can then lead to systemic risk events, see for example \cite{Pedersen2009}, which has been by now widely accepted to be a dynamic event, \cite{azizpour-giesecke-schwenkler,Brunnermeier2012}.

In the past ten years researchers have tried to understand and model such behavior in different ways. A significant body of literature has emerged that is aiming at understanding and modeling complex financial systems. Before describing the main contributions of this paper, let us first briefly describe the three main different lines of research that have emerged in the study of systemic risk. Firstly, there is the network models for clustering and contagion that follow the earlier work of \cite{AllenGale2000,EisenbergNoe2001}, see also \cite{GlassermanYoung2016} for a review. Secondly, there is the dynamic mean field type of models literature, see for example \cite{BoCapponi2013,BoCapponi2015,CapponiSunYao2019,daipra-tolotti,FouqueIchiba2013,giesecke-weber,hambly2011,HamblySojmark2018,PapanicolaouSystemicRisk}. Thirdly, there is the reduced form credit and portfolio risk literature that is using intensity models of correlated default, \cite{CMZ,GSS,GSSS,Spiliopoulos2015,SpiliopoulosSiriganoGiesecke2013,SpiliopoulosSowers2013}. Despite this significant progress, many questions are still wide open.

 Our work falls in the last category, i.e., in the reduced form credit risk literature. Motivated by the empirical work of \cite{azizpour-giesecke-schwenkler} and following \cite{GSS,GSSS}, the intensity to default process for each individual name in the pool is characterized by three terms: an idiosyncratic term, which is specific to each name, a contagion term, which is responsible for clustering of defaults, and an exogenous risk term common to all names in the pool. When considering a large system, we will often refer to this as pool of names where names are the system's components.  As it has been established in \cite{azizpour-giesecke-schwenkler,GSS, GSSS}, see also \cite{Spiliopoulos2015} for a review, these terms give important insights on how risk propagates and on how defaults cluster. Due to the interconnectedness of the system, the failure of a single component increases the likelihood of failure of other components in the system.  Uncertainty becomes an issue which then leads market participants to fear even more losses in asset prices disproportional  to the magnitude of the crisis. Reduce-form point process models of correlated default are many times used to assess portfolio credit risk and are based on counting processes. We use dynamic portfolio credit risk models to understand  large financial systems asymptotics and default clustering.

Our contribution in this paper is twofold. Firstly, we consider network effects, a feature missing from the earlier work of \cite{CMZ,GSS} and its follow ups.
To be more precise, we specify the interaction of names  by a weighted, directed graph $G(\Gamma,\mathcal{E}, \omega)$ where $\Gamma$ is the set of vertices (i.e., names), $\mathcal{E}$ is the set of (directed) edges and $\omega \,:\, \mathcal{E} \to (0,\infty)$ is a function assigning weights to edges
(as a convention we could define $\omega(i,j) = 0$ whenever $(i,j) \notin \mathcal{E}$). An edge $(i,j) \in \mathcal{E}$ implies a directed interaction, the impact that the default of name $i$ has on name $j$. The weight $\omega(i,j)$ measures the strength
of the interaction. For example, $\omega(i,j)$ could represent the loss of name $j$ at the default of counterparty $i$ (the loss is usually the positive part of the mark-to-market value of the contract at default). As we shall see, the weight $\omega(i,j)$ also represents the magnitude of the increase in the default intensity of name $j$ due to counterparty losses at the default of name $i$. Let $\Delta$ be the matrix with elements $\omega(i,j)$ for $i,j=1,\cdots, N$. As it turns out, a singular value decomposition (SVD) of $\Delta$ allows us to quantify contagion effects. In addition, the SVD  allows us to quantify the levels of interaction (this is the number of non-zero eigenvalues of $\Delta$ and it will become precise in Section \ref{S:Model}) that we need in order to  effectively coarse-grain the heterogeneous system. It also allows us to reduce the dimensionality of the system via appropriate low-rank approximations. In this paper, we theoretically analyze the limit of the empirical measure of surviving names as $N\rightarrow\infty$ and we also showcase the different cases by numerical studies. We demonstrate numerically that if there is sufficient spectral gap in the eigenvalues of $\Delta$ from the SVD, then the probability distribution of stochastic processes of interest is very well approximated by appropriate low rank approximations. This becomes practically useful, since without the low-rank approximation, as we will see, the computation of the quantities of interest can become prohibitively expensive.

In this paper, we assume that we are given an adjacency matrix $\Delta$ with sufficiently regular behavior (see Sections \ref{S:Model}-\ref{S:Assumptions} for details). Then, our goals are to study the typical behavior of the loss rate both in the overall pool and within names of the same type. In addition, we study the mean impact to default on a given name from system wide defaults as the number of components $N\rightarrow\infty$. We allow the pool to be heterogeneous with stochastic intensity that evolves dynamically in time and with different weights $\omega(i,j)$ for different $i,j$. In addition, the  loss rate (either overall in the pool or for names of specific types) and the mean default impact  on a given name from system wide defaults are dynamic quantities and their computations can be numerically cumbersome. We show numerically that low-rank approximations motivated through the SVD can be very effective in accurately reducing the dimension of the system and thus making their evaluation numerically tractable.

Therefore, the procedure developed in this paper allows to quantify the effect of the given adjacency matrix $\Delta$ on dynamic quantities that are of interest, such as distribution  of the loss rate in the pool, distribution of the loss rate within names of specific types, mean effect on given names from system wide defaults, etc.  Note that evaluation of quantities such as loss rate of the whole pool and loss rate within names of specific types offers additional insights into the possibility of many names defaulting within short periods of time from each other (i.e. of default clustering). Indeed, an increase of the  mean of the loss rate of the pool at a given time signals higher likelihood  of many defaults. Then, studying the loss rates within names of specific types indicates which types of names are more likely to default. Naturally, names of types with larger mean loss rate will be more likely to default, revealing the structure of the cascade event. In addition, we find,  via the SVD, that the mean loss rate in the pool is positively correlated with a specific coefficient, later on called the contagion coefficient. In particular, the contagion coefficient is a function of the corresponding singular values and of the orthogonal vector coefficients capturing the exposure of the network to contagion. We demonstrate these findings in our numerical studies of Section \ref{S:Numerics}, where we demonstrate how these issues can be quantified.

Secondly, we consider general stochastic intensity-to-default processes where the drift coefficient of the idiosyncratic component is only required to satisfy appropriate dissipative properties instead of requiring it to be affine. We prove well-posedness of the related stochastic intensity models and rigorously characterize the limit of the empirical survival distribution of the names in the pool as their number grows to infinity.

In the recent related work \cite{CapponiSunYao2019}, that falls in the literature of dynamic mean field models, the authors consider a model of interbank lending (not a reduced form credit risk model that we consider here) which  accounts for network topology and propagation of systemic risk and perform asymptotic analysis as the number of banks $N\rightarrow\infty$ as well. In the present paper we also look at the limit as $N\rightarrow\infty$ and account for network topology, but we focus on the impact of the network matrix through spectral analysis on the evolution of default intensities and on specific statistics of interest such as loss rate in the pool and mean impact on specific types of network components.

At this point, we want to mention that even though our primary motivation comes from interacting particle systems in financial mathematics, our results are broader applicable. In a given system with many different components, not all components are equally connected to other components or equally affected by the default of other components. The failure of one component due to external forcing giving rise to failure of other components of a given system is of broader interest.

The rest of the paper is organized as follows. In Section \ref{S:Model} we describe our model in detail. In Section \ref{S:Assumptions} we lay down our assumptions that are assumed to hold throughout the paper. Section \ref{S:MainResults} contains the main results of this paper. The proof of the main theorem is in the subsequent sections. In particular, tightness and characterization of the limit points of the empirical measure is discussed in 
Section \ref{S:LimitCharacterization} followed by uniqueness of the limiting point in Section \ref{S:Uniqueness}. Section \ref{S:Numerics} contains our simulation studies and numerical results on low-rank approximations. Technical results and their proofs have been gathered in Appendix \ref{A:Appendix}. Section \ref{S:Conclusions} is about our conclusions and outlook for future work.
\section{Model description}\label{S:Model}

The model considered in this paper models the evolution of a system consisting of $N$ names which are subject to default risk. The model for the  default risk takes into account three terms:  an idiosyncratic risk (specific to a given name), a systematic risk (common to all names) and a term modeling default contagion and spiral events. The last term takes into account the network topology.

Fix a probability space $(\Omega,\mathcal{F},\mathbb{P})$ where all random variables are defined. Let $\{W^n\}_{n\in\mathbb{N}}$ be a collection of i.i.d. standard Brownian motions which are used to model the idiosyncratic risk for each component of the pool. Let $V$ be a standard Brownian motion independent from  the $W^n$'s,  driving the randomness of  the systematic risk factor process $X$. Let $\mathcal{V}_t=\sigma(V_s,0\leq s \leq t)\vee \mathcal{N}$, where $\mathcal{N}$ is the set of null sets. Let $\{\mathfrak{e}_n\}_{n\in\mathbb{N}}$ be a collection of independent standard exponential random variables.

For $N\in \mathbb{N}$ and $n\in \{1,2,\ldots,N\}$, denote by $\tau^{\NN}$  the stopping time at which the $n$-th component of the system fails. The failure time $\tau^{\NN}$ has stochastic intensity process $\lambda^{\NN}$ to be described below. The default time $\tau^{N,n}$ is
\[
\tau^{N,n} = \mbox{inf} \Big\{t\geq0: \int_{0}^t \lambda_s^{N,n} ds \geq \mathfrak{e}_n\Big\}.
\]

We can also write
\[
\chi_{\{\tau^{N,n} \leq t\}}=\chi_{[\mathfrak{e}_n,\infty)}\Big(\int_{0}^t\lambda_s^{N,n} ds\Big),
\]
where $\chi_B$ is the indicator function for a set $B$.

Recall the network structure of the system, which is described by a directed graph $G(\Gamma,\mathcal{E},\omega)$ where $\Gamma$ is the set of components in the system, $\mathcal{E}$ is the set of directed edges and $\omega:\mathcal{E}\to (0,\infty)$ is the function assigning weights to edges. $\omega(i,j)$ represents the default impact the $i$-th name has on the $j$-th firm.

Then, the total loss experienced by name $j$ due to system wide defaults by time $t$ is
\begin{align}
  \sum_{i =1}^N \omega(i,j) \chi_{\{\tau^{N,i}\leq t\}},
 \end{align}
 and, as we shall see, it also represents the total increase in the default intensity of the $n$-th name in the pool. Let $\Delta$ be the adjacency matrix of $G$, i.e. the $(i,j)$-th
entry of $\Delta$ is given by $\omega(i,j)$ for $(i,j) \in \mathcal{E}$ and
$0$ if $(i,j) \notin \mathcal{E}$.

 Then, the classical  singular value decomposition (SVD in short) yields
\begin{align} \label{SVD}
  \Delta = \sum_{j = 1}^d \xi^{2}_j \, \ell_j u_j^\top
\end{align}
where $\{\ell_1, \dots, \ell_d\}$ are orthonormal vectors
(spanning columns of $\Delta$), $\{u_1, \dots, u_d\}$ are orthonormal vectors
(spanning rows of $\Delta$) and $\xi^{2}_1 > \xi^{2}_2 > \dots >
\xi^{2}_d > 0$ are real numbers known as the singular values.

Here, $d\leq N$ is called the rank of $\Delta$.
In a sense $d$ represents the complexity of the system. The larger $d$ is, the more complex the structure of the interaction becomes.

Let $\ell_{i,j}$ be the
$i$-th entry of $\ell_j$ in (\ref{SVD}) and similarly let $u_{i,j}$ be the $i$-th entry of the vector $u_j$.
The mean default impact on the $n$-th name from system wide defaults up to time $t$, can be written as
\begin{align}
Q_t^{N,n,\Delta}&=\frac{1}{N}\sum_{i=1}^N\omega(i,n)\chi_{\{\tau^{N,i} \leq t\}}=\sum_{j=1}^d \xi_j^2 u_{n,j} \frac{1}{N} \sum_{i=1}^N \ell_{i,j} \chi_{\{\tau^{N,i} \leq t\}}=\beta_{n}^{C,\Delta} \cdot L_t^{N,\Delta},\label{Eq:MeanImpact0}
\end{align}
where $\beta_{n}^{C,\Delta}=(\xi_1^2 u_{n,1}, \xi_2^2 u_{n,2}, \ldots, \xi_d^2 u_{n,d})^T$ and the  vector-valued process $L^{N,\Delta}_t=(L^{N,1}_t,L^{N,2}_t,\dots,L^{N,d}_t)^T$ has elements
\[
L_t^{N,j}=\frac{1}{N}\sum_{i=1}^N \ell_{i,j} \chi_{\{\tau^{N,i} \leq t\}}.
\]

The $j$-th entry of $L_t^{N,\Delta}$ can be loosely interpreted as the stochastic loss rate of the $j$-th level of interaction of the network.

The element $\ell_{i,j}$ of the vector $\ell_{j}$ can be interpreted as the contribution of the $i$-th bank on the $j$-th level of interaction, for $j=1,\cdots, r$. Analogously, the element $\beta_{n,i}^{C,\Delta}$ of the vector $\beta_{n}^{C,\Delta}$ can be interpreted as the exposure of the $n$-th bank on the $i$-th level of interaction for $i=1,\cdots, r$.

Notice that $Q_t^{N,n,\Delta}$ can be interpreted as the mean increase over the $n$-th bank's intensity to default due to the default of other banks by time $t$.

We are interested in the behavior of quantities like $Q_t^{N,n,\Delta}$ when the system is large, i.e. when $N\rightarrow\infty$. As we elaborate in more detail in Remark \ref{R:NetworkRemark}, in large systems it is reasonable to rely on a low rank approximation. In addition, for purposes of computational feasibility one would like to approximate $\Delta$ by an appropriate low rank approximation.

One popular way to do so, is to use a classical result from matrix algebra  stating that
if $0<r<d$ is a positive integer, then the minimal value of the $L^{2}$ distance $\|D-B\|_{2}$ (the standard Frobenius norm) over all matrices $B$ with rank less or equal to $r$ is achieved at
\[
 A= \sum_{j = 1}^r \xi^{2}_j \, \ell_j u_j^\top
\]
with $\xi^{2}_j$ in decreasing order. In addition, we actually have
\begin{align}
 \|\Delta-A\|_{2}&=\sum_{i=r+1}^{d}\xi^{2}_{i}.\label{Eq:SVDreduction}
\end{align}

Such a reduction is especially meaningful if the rank of $\Delta$, $d$, is large but there are only a few dominant eigenvalues. In such a situation one typically would like to take advantage of this. This is the practical perspective that we take here. In fact, given a large matrix $\Delta$ one would first investigate the possibility of a good low rank approximation, then choose a certain low rank approximation that is comfortable with and work with that.  As we shall see in Section \ref{S:Numerics}, such an approximation in combination with the coars-graining achieved by Theorem \ref{T:MainTheorem}  makes the problem computationally more tractable.

At this point, let us also mention that while the elements of the original matrix $\Delta$, i.e. $\omega(i,j)$, are nonnegative, it is likely that an arbitrarily chosen low-rank approximation $A$ to $\Delta$ could have some of its elements to be negative. Therefore, some financial meaning could be lost sometimes depending on the chosen low-rank approximation. However, one does not expect this to be the case if the spectral gap in the eigenvalues is sufficiently large and one chooses a low rank approximation consistent with the spectral gap (i.e. one that corresponds to (\ref{Eq:SVDreduction}) with small right hand side). The numerical examples in Section \ref{S:Numerics} demonstrate that in this case the value of the statistics of interest (financial indicators of interest), up to negligible approximation errors, are not affected by such a good low-rank approximation.

The previous discussion then motivates us to replace $\Delta$ by $A$ and to subsequently define the quantity
\begin{align}
Q_t^{N,n,A}&=\beta_{n}^{C,A} \cdot L_t^{N,A},\label{Eq:MeanImpact}
\end{align}
where $\beta_{n}^{C,A}=(\xi_1^2 u_{n,1}, \xi_2^2 u_{n,2}, \ldots, \xi_r^2 u_{n,r})^T$ and the  vector-valued process $L^{N,A}_t=(L^{N,1}_t,L^{N,2}_t,\dots,L^{N,r}_t)^T$. For simplicity of notation we use the same notation for the components $L^{N,i}_{t}$ for both A and $\Delta$, even though we always work with $L^{N,A}$, so there should be no confusion.

Now that we have discussed the matrix $\Delta$ defining the network structure, let us be more specific in regards to the dynamics. An intensity is driven by an idiosyncratic risk represented by a Brownian motion $W^n$, a systematic risk represented by the process $X$, and spillover risk represented by the process $Q_t^{N,n,A}=\beta_n^{C,A} \cdot L^{N,A}_t$ (defined via $A$, the low-rank approximation to $\Delta$).
In particular, we consider the following interacting system
\begin{eqnarray}
d\lambda^{N,n}_t&=& b(\lambda^{N,n}_t,a_n)dt + \sigma_n\cdot(\lambda^{N,n}_t)^\rho dW^n_t +\beta^{C,A}_{n}\cdot dL^{N,A}_t + \beta^S_{n} \lambda^{N,n}_t dX_t\nonumber\\
\lambda^{N,n}_0&=&\lambda_{0,N,n}\nonumber\\
dX_t&=&b_0(X_t)dt + \sigma_0(X_t)dV_t\nonumber\\
X_0&=&x_0\nonumber\\
L_t^{N,j}&=&\frac{1}{N} \sum_{n=1}^N \ell_{n,j} \chi_{\{\tau^{N,n} \leq t\}}. \quad \quad j=1,2,\dots,r\label{Eq:MainModel}
\end{eqnarray}

Notice that (\ref{Eq:MainModel}) has been defined in terms of $A$ and not in terms of the original $\Delta$. This represents what one would do in practice, in order to simplify the system as it will become clearer in Sections \ref{S:Assumptions} and \ref{S:MainResults}.

In addition, we allow for a heterogeneous pool, which means that the intensity dynamics of different names can be different. In the model $\sigma_n \in \mathbb{R}_{+}$,$a_n \in \mathbb{R}^k$ for some $k>0$, $\beta_{n}^S \in \mathbb{R}$ are constants and $1/2\leq \rho < 1$. Let us set $\mathcal{P}=\mathbb{R}_{+}\times\mathbb{R}^{k+2r+1}$ and $\hat{\mathcal{P}}=\mathcal{P}\times\mathbb{R}_{+}$. For all $n\in \{1,2,\dots, N\}$, we capture these different dynamics by defining the ``types"
\begin{equation}
p^{n} = (\sigma_{n},a_n,\beta^C_{n,1},\cdots,\beta^C_{n,r},\beta^S_n, \ell_{n,1},\cdots,\ell_{n,r})\in \mathcal{P}\label{Eq:ParameterVector1}
\end{equation}
and
\begin{equation}
 \hat{p}^{n} = (p^{n},\lambda_{0,N,n})\in \hat{\mathcal{P}}.\label{Eq:ParameterVector2}
\end{equation}
Furthermore, we let  $\hat{p}^{n}_{t} = (p^{n},\lambda_{t}^{\NN})\in \hat{\mathcal{P}}$.

From now on we suppress the superindex $A$, and we simply write $Q_t^{N,n},\beta_{n}^{C}, L^{N}_t$ in place of $Q_t^{N,n,A},\beta_{n}^{C,A}, L^{N,A}_t$. It will always be clear from context which matrix is being used.

As just mentioned $Q_t^{N,n}=\beta_{n}^C \cdot L_t^N$ represents  the (approximate, due to the potential low-rank approximation) mean impact on the $n$-th name from system wide defaults up to time $t$. The vector $\beta_{n}^{C}=(\beta_{n,1}^{C},\cdots, \beta_{n,r}^{C})$ with $\beta_{n,i}^{C}=\xi_i^2 u_{n,i}$ will be interpreted as a contagion coefficient vector. Higher values of $\beta_{n,i}^C$ imply higher impact on the default intensity of the $n$-th name. This is natural to expect as the $n$-th column of the matrix $A$ represents the impact from defaults when claims of the $n-$th institution towards all other institutions are present.  Other network performance indicators of interest are $D_t^N = \frac{1}{N} \sum_{n=1}^N \chi_{\{\tau^{N,n} \leq t\}}$ and $D_t^{N}(p_B) =\frac{1}{N_B} \sum_{n=1}^N \chi_{\{\tau^{N,n} \leq t\}} \chi_{\{p^{N,n}=p_B\}}$  with $N_{B}=\sum_{n=1}^{N}\chi_{\{p^{N,n}=p_B\}}$ (we use the notation $\{p^{N,n}=p_B\}$ to distinguish names of type $B$), the overall loss rate in the pool and the loss rate for names of the same type, say type $B$, respectively. When $N$ is large, numerical approximation of the distribution of these quantities becomes possible through the approximation theorem (Theorem \ref{T:MainTheorem}) of this paper. As we shall see in Section \ref{S:Numerics}, names of types with large contagion coefficients will tend to have larger mean losses.

In addition, $d$ for $\Delta$ or $r$ for its low-rank approximation $A$ reveals a hierarchical structure of $d$ or $r$ levels respectively. For example, a rank one $(r=1)$  approximation of the matrix $\Delta$ will have a  more homogeneous structure than a rank two $(r=2)$ approximation of the matrix $\Delta$. In particular,  names that are of the same type in a rank one approximation of $\Delta$ (in terms of the dynamic evolution of their intensity process from (\ref{Eq:MainModel})), may be of different type in a rank two approximation (and thus have different intensity to default process in terms of (\ref{Eq:MainModel})). Said otherwise, a network system corresponding to a matrix $\Delta$ with a large number of non-zero eigenvalues $r$ will have a finer structure than a system with a smaller number of  $r$.   One can interpret  $r$ as the number of levels of interaction in the system. We will discuss this again in Sections \ref{S:Assumptions} and \ref{S:Numerics}.

 Our paper extends significantly the result of \cite{GSSS}. Firstly, the drift term $b(\lambda,a)$ only needs to have certain dissipative properties with respect to $\lambda$.  Secondly, we now have a network structure described through the adjacency matrix $\Delta$. As we shall see, the analysis of this model is not only more challenging, but it also requires new arguments and ideas.  While the main arguments and overall proof strategy is based on the  methods of \cite{GSS,GSSS}, the new mathematical arguments that are needed, are presented in the Appendix \ref{A:Appendix}. The introduction of the network structure through the adjacency matrix $\Delta$, allows for a far richer set of questions to be asked.

\section{Notation and Assumptions}\label{S:Assumptions}
In this section, we go over our assumptions that are assumed to hold throughout the paper.

We start with Assumptions \ref{A:Assumption2}, \ref{A:Assumption1b} and \ref{A:Assumption1} that are related to the importance of having sufficiently regular behavior of the adjacency matrix $\Delta$, or more specifically of its low-rank approximation $A$, and of the vector of parameters  $p^{n}$ and $\hat{p}^{n}$ defined via (\ref{Eq:ParameterVector1}) and (\ref{Eq:ParameterVector2}) respectively. In addition to the rest of the assumptions,  Assumptions \ref{A:Assumption2}, \ref{A:Assumption1b} and \ref{A:Assumption1} guarantee  well defined limits later on as well as computational feasibility of the limit equation.

\begin{assumption}\label{A:Assumption2}
	Assume that there is a constant $K_{\ref{bdd}}>0$ such that all the coefficients $\sigma_n$, $a_n$, $||\beta_n^C||$, $|\beta_n^S|$ and $|\ell_{n,j}|$ $j=1,2,\ldots,d$ and $n=1,2,\cdots$ are bounded by $K_{\ref{bdd}}$ and there exists a $\bar{\sigma}>0$ that $\inf_{n} \sigma_n^2 \geq \bar{\sigma}^2>0$.\label{bdd}
\end{assumption}

Assumptions \ref{A:Assumption1b} and \ref{A:Assumption1} that follow are phrased in terms of $A$ because the model (\ref{Eq:MainModel}) is based on $A$. Clearly, if they already hold for the ordinal matrix $\Delta$, then $\Delta$ can be used directly in place of $A$ in (\ref{Eq:MainModel}).

 In practice one is typically given a large matrix $\Delta$, chooses a good low-rank approximation to $\Delta$ and works with that specific approximation. In other words, for all practical purposes, one would like  to be able to work with low-rank approximations  $A$.   In fact, for theoretical reasons, we will assume a little bit more as Assumption \ref{A:Assumption1b} specifies.
\begin{assumption}\label{A:Assumption1b}
We assume that as $N$ grows, the rank $r$ of the matrix $A$ that is used in the model (\ref{Eq:MainModel}) stays bounded.
\end{assumption}

Next, let us define
\[
\pi^N=\frac{1}{N}\sum_{n=1}^N \delta_{p^{n}},\text{ and }\Lambda_0^N=\frac{1}{N}\sum_{n=1}^N \delta_{\lambda_{0,N,n}}.
\]

The measures $\pi^N$ and $\Lambda_0^N$ belong to  the space of  Borel probability measures on $\mathcal{P}$ and $\mathbb{R}$ respectively. These spaces will be denoted by $\mathfrak{P}(\mathcal{P})$ and $\mathfrak{P}(\mathbb{R})$ respectively.
\begin{assumption}\label{A:Assumption1}
Assume that the limits
	$$\pi=\lim_{N\to \infty}\pi^N$$
	$$\Lambda=\lim_{N\to \infty}\Lambda_0^N$$
	exist on $\mathfrak{P}(\mathcal{P})$ and $\mathfrak{P}(\mathbb{R})$ respectively.
\end{assumption}

Undoubtedly Assumptions \ref{A:Assumption2}, \ref{A:Assumption1b} and \ref{A:Assumption1} imply certain behavior of the network of institutions. The following Remarks \ref{R:NetworkRemark} and \ref{R:NetworkRemark2} are related.
\begin{remark}\label{R:NetworkRemark}
Assumption \ref{A:Assumption2} on the boundedness of $||\beta_n^C||$ and $|\ell_{n,j}|$ for $j=1,2,\ldots,d$ and all $n\in\mathbb{N}$ allows us to prove tightness of the measure valued process keeping track of the defaults (see Section \ref{S:MainResults}) but it also implies that the original matrix $\Delta$ can be very well approximated by setting equal to zero singular values lower than a given threshold, see \cite{Chatterjee2015}. In particular, \cite{Chatterjee2015} shows that for given $\epsilon>0$, the $\epsilon$-rank of $\Delta$ (i.e. the smallest possible rank of matrices whose distance from $\Delta$ in
terms of the maximum absolute entry norm is less than $\epsilon$) is at most of order $\sqrt{N}$. This result is then strengthened in \cite{UdelTownsend2019} to order of $\log N$, if in addition each element of the matrix $\Delta$ can be generated by applying a piecewise analytic function to potentially high dimensional but bounded latent variables.

These results imply that sufficiently large data sets tend to have low rank structure even if there may be no underlying physical reason, see \cite{UdelTownsend2019}.  These suggest that when $N$ is large one can reasonably expect that the matrix $\Delta$ is well approximated by a low-rank matrix $A$. This is the regime of interest in this paper. Low rank approximations are not new in the financial literature, see for example \cite{Jondeau}.
  Low rank structure is evident in block-models networks and low-rank approximations can be used to identify core-periphery structures (a well known financial network of interest) see \cite{Cucuringu}.

 The empirical results of \cite{CraigVonPeter2014} demonstrate  that  the core-periphery structure is a financial network of interest.  It is found empirically in \cite{CraigVonPeter2014} for the German interbank network that interbank markets are tiered which means that most banks do not lend to each other directly but through intermediaries. This phenomenon can be captured by a core-periphery model. The network observed in \cite{CraigVonPeter2014} is sparse, directed and valued.

In this paper, we are interested in studying the limit behavior of dynamic quantities such as  $Q_{t}^{n,N,\Delta}$ and loss rate in  the pool or within names of given type as $N\rightarrow\infty$ and in order to be able to do so, both mathematically and numerically, we need to assume that we can work with a matrix $\Delta$ (or an appropriate low-rank approximation $A$) such that its rank can be taken, or approximately considered to be bounded as $N\rightarrow\infty$. Assumption \ref{A:Assumption1b} makes this restriction precise, in which case the theoretical results of Section \ref{S:MainResults} hold. In addition, Assumption \ref{A:Assumption1b} also holds in the numerical examples, including the core-periphery one, that we numerically study in Section \ref{S:Numerics}. In the numerical experiments presented in Section \ref{S:Numerics}, it will be clear which matrix is being used to define $Q_{t}^{n,N}$ and consequently the model (\ref{Eq:MainModel}). The conclusions section \ref{S:Conclusions} discusses the possibility of treating the case where the rank increases with $N$ as well, but we do not elaborate more on this in this work.

Assumption \ref{A:Assumption1} on $p^{n}$ implies that the empirical distribution of the spanning columns and rows of the adjacency matrix have a well defined limit in distribution. For example, this assumption holds if there is only a finite number of non-zero entries in the vectors $\ell_{j},u_{j}$ for each $j$ with specific frequencies. This will be the case for example in all of the numerical studies of Section \ref{S:Numerics}. In practice given a specific large $N$, one would use Theorem \ref{T:MainTheorem} to approximate the probability distribution of  quantities of interest, but of course use the empirical distribution $\pi^{N}$ and $\Lambda_{0}^{N}$ as approximations to $\pi$ and $\Lambda_{0}$ respectively.
\end{remark}

\begin{remark}\label{R:NetworkRemark2}
For completeness, let us present a simple example of a core-periphery model that has bounded rank as $N$ grows to infinity, i.e., it satisfies Assumption \ref{A:Assumption1b}. The model presented here, also satisfies Assumption \ref{A:Assumption1}. Consider a base model
\begin{equation*}
\Delta_{N_{0}\times N_{0}} =
\begin{pmatrix}
C & CP  \\
PC & O
\end{pmatrix}
\end{equation*}
where $C$ is a $N_{c}\times N_{c}$ matrix (representing the base model for the core), $CP$ is  a $N_{c}\times N_{p}$ matrix, $PC$ is a $N_{p}\times N_{c}$ matrix (representing the base model for the interactions between core and periphery) and $N_{c}+N_{p}=N_{0}$. Then, for $k\in\mathbb{N}$, let $N=k\times N_{0}$ and consider the $N\times N$ network matrix $\Delta$
\begin{equation*}
\Delta=\Delta_{N\times N} =
\begin{pmatrix}
C & \cdots & C & CP &\cdots & CP \\
\vdots & \ddots & \vdots & \vdots & \ddots & \vdots \\
C & \cdots & C & CP &\cdots & CP \\
PC & \cdots & PC & O &\cdots & O \\
\vdots & \ddots & \vdots & \vdots & \ddots & \vdots \\
PC & \cdots & PC & O &\cdots & O
\end{pmatrix}
\end{equation*}
$\Delta=\Delta_{N\times N}$ is constructed via $\Delta_{N_{0}\times N_{0}}$ by extending the $N_{c}$ core banks in the base model to $k\times N_{c}$ banks and the $N_{p}$ periphery banks in the base model to $k\times N_{p}$ banks. The rank of the matrix $\Delta_{N_{0}\times N_{0}}$ is bounded form above by $ N_{0}$. In addition, a simple computation shows that for every $k\in\mathbb{N}$ and therefore for every $N\in\mathbb{N}$ the rank of $\Delta=\Delta_{N\times N}$ is also bounded from above by  $ N_{0}$.

For such a model we assume that each one of the $k$-copies of the original $N_{0}$ institutions has a corresponding intensity to default process $\lambda_{t}^{\NN}$ defined with the same values (or i.i.d copies if randomly chosen) for the defining parameters $(a_{n},\sigma_{n},\beta^{S}_{n},\lambda_{0,\NN})$ with the corresponding original institution in the base model.
\end{remark}

For the drift coefficient function $b(\lambda,a)$ we assume the following growth and regularity conditions.
\begin{assumption}\label{assumptionb}
The mapping  $\lambda\mapsto b(\lambda,\cdot)$ is  locally Lipschitz and there exists finite constants $d>1$, $q>1$, $K>0$ and positive bounded functions $\gamma$ and $k$ with $\gamma(a)>0$ and $k(a)>0$ such that $$\lambda b(\lambda,a)<-\gamma(a)|\lambda|^d, \quad \text{ for } |\lambda|\geq K$$ $$|b(\lambda,a)|\leq k(a) (1+|\lambda|^q),$$  and $$b(0,a)>0.$$
	Furthermore we assume that for any $\lambda\in \mathbb{R}_{+}$,  $a \mapsto b(\lambda,a)$ is a continuous function.
\end{assumption}

A remark in regards to Assumption \ref{assumptionb} follows.
\begin{remark}
If we take $a=(\bar{\alpha},\bar{\lambda})\in\mathbb{R}^{2}_{+}$ and $b(\lambda,a)=-\bar{a}(\lambda-\bar{\lambda})$, then the idiosyncratic part of the intensity process becomes the classical CEV model. Notice that in this case
$b(\lambda,a)=-\partial_{\lambda}V(\lambda,a)$ with $V(\lambda,a)=\frac{\bar{a}}{2}(\lambda-\bar{\lambda})^{2}$ and the function $V$ has a single minimum point at $\lambda=\bar{\lambda}$. In turn this mean reversion of $\lambda$  implies that the impact of a default fades away with time and the intensity will tend to revert back to the level $\lambda=\bar{\lambda}$.

Assumption \ref{assumptionb} relaxes the affine structure to a requirement about appropriate dissipativity of the drift coefficient $b(\lambda,a)$. This enlarges the class of drifts $b(\lambda,a)$ that one can consider. For example, one could consider situations where $b(\lambda,a)=-\partial_{\lambda}V(\lambda,a)$ with $V(\lambda,a)$ being a bistable potential. Such situations could correspond to situations where the creditworthiness of certain names might have two equilibria, corresponding to two different parts of the business cycle.

The goal of this paper is to explore (Section \ref{S:Numerics}) the potential effects of the network structure and of low-rank approximations on the distribution of dynamically evolving stochastic processes of interest. The aforementioned numerical exploration is based on the rigorous mean field limit of the empirical survival distribution of the names in the pool (Theorem \ref{T:MainTheorem}). Theorem \ref{T:MainTheorem} proves that appropriate dissipative conditions on the drift coefficient $b(\lambda,a)$ are enough to guarantee  well defined intensity-to-default processes and subsequently well defined mean field limits of the empirical survival distribution. See also Section \ref{S:Conclusions} for a more elaborate related discussion and potential future directions.
\end{remark}

The rest of the assumptions are related to the exogenous risk process $X$.
\begin{assumption}\label{s0}
	Assume that the function $\sigma_0(\cdot)$ is bounded, that is there exists a constant $K_{\ref{s0}}$ such that  $|\sigma_0(x)|<K_{\ref{s0}}$. For $b_0$ assume $\sup_{t<\infty}\mathbb{E}|b_0(X_t)|^{4p}<\infty$ for some $p\geq 1$.
\end{assumption}

Let us define
\[
\Gamma_t=-\beta^S \int_0^t b_0(X_s) ds.
\]
\begin{assumption}\label{Xt}
	Assume that for some $p\geq1$, $\sup_{t<\infty}\mathbb{E}[X_t^{2p}]$ and $\sup_{t<\infty}\mathbb{E}[e^{4p|\Gamma_t|}]$ are bounded.
\end{assumption}

The last Assumption \ref{unbdd_b} makes sure that we can extend some technical lemmas from bounded drifts $b_0(x)$ to potentially unbounded ones.

\begin{assumption}\label{unbdd_b}
Assume there is a function $u(x)$ such that $\sigma_0(x)u(x)=-b_0(x)$ and for any $T>0$ we have
	\[
\mathbb{E}\left[e^{1/2 \int_0^T |u(X_s)|^2 ds}\right] <\infty,
\]
and that for any $T$ there is a $p>1$ such that
	$$\mathbb{E} \left[\left|e^{-\int_0^T u(X_s) dV_s - 1/2 \int_0^T |u(X_s)|^2 ds} \right|^p \right] < \infty.
$$
\end{assumption}

An example where Assumptions \ref{s0}, \ref{Xt} and \ref{unbdd_b} hold is to take $b_{0}(x)=-\gamma x$ and $\sigma_{0}(x)=1$, which is the mean reverting example that is studied in \cite{GSS}.

\section{Well-posedness of the model and main results}\label{S:MainResults}

In this section we prove that the model is well-possed and we present our main results. Let us begin with well-posedness of the model, Lemma \ref{sol}.
\begin{lem}\label{sol}
	Let $\xi$ be  a vector of processes having $r$ components, predictable, right-continuous, monotone and bounded with $\xi_{0}=0$. Let Assumptions \ref{A:Assumption2}-\ref{unbdd_b} hold. There exists a unique nonnegative solution $\lambda$ of the following SDE:
	\begin{align}
	d\lambda_t&=b(\lambda_t,a)dt + \sigma(\lambda_t\vee 0)^{\rho} dW_t + \beta^C \cdot d \xi_t+ \beta^S \lambda_t dX_t\nonumber\\
	\lambda_0&=\lambda_o\nonumber\\
	dX_t&=b_0(X_t)dt+\sigma_0(X_t)dV_t.\nonumber
	\end{align}
\end{lem}

Lemma \ref{lambda} is about an essential a-priori bound that will be used in many places of the subsequent proofs.
\begin{lem} \label{lambda}
	Let $\lambda^{N,n}$ be the unique solution to (\ref{Eq:MainModel}), guaranteed under the assumptions of Lemma \ref{sol}. Let  $p\geq1$ be such that Assumptions \ref{s0} and \ref{Xt} hold. Then, for such $p \geq 1$ and for every $T \geq 0$,
	$$K_{\ref{lambda}} \myeq \sup_{0\leq t \leq T, N \in \mathbb{N}} \frac{1}{N} \sum_{n=1}^N \mathbb{E}[|\lambda_t^{N,n}|^p]$$
	is finite.
\end{lem}

Proofs of Lemmas \ref{sol} and \ref{lambda} are in Appendix \ref{A:Appendix}. Let us denote the survival indicator process for a given name in the pool by
	$$M_t^{N,n}=\chi_{\{\tau^{N,n}>t\}}$$
and define the empirical  distribution of the $\hat{p}^n$'s corresponding to the names that have survived up to time $t$ as follows:
	$$\mu_t^{N}=\frac{1}{N} \sum_{n=1}^N \delta_{\hat{p}_t^{n}} M_t^{N,n}.$$
Notice that $\mu_t^{N}$ captures the entire dynamics of the model (including the effect of the heterogeneities and network topology).

In order to study the convergence of $\mu^N$, we need to set up the appropriate topological framework. That is, let $E$ be the collection of sub-probability measures on $\hat{\mathcal{P}}$,  i.e., $E$ consists
of those Borel measures $\nu$ on $\hat{\mathcal{P}}$ such that $\nu(\hat{\mathcal{P}})\le 1$. Then fix a point $\star$ which is not in $\hat{\mathcal{P}}$ and let $\hat{\mathcal{P}}^{+}=\hat{\mathcal{P}} \cup \{\star\}$ (the so-called  one-point compactification of $\hat{\mathcal{P}}$). Open sets are those which are open subsets of $\hat{\mathcal{P}}$ (endowed with the original topology) or complements in $\hat{\mathcal{P}}^+$ of closed subsets
of $\hat{\mathcal{P}}$ (again, in the original topology of $\hat{\mathcal{P}}$).

Define a bijection $\zeta$ from $E$ to the Borel probability measures on $\hat{\mathcal{P}}^{+}$ as
\[
(\zeta \nu) (Z) = \nu (Z\cap \mathcal{P}) + (1- \nu (\mathcal{P}))\delta_{\star}(Z),
\]	
for any $Z \in \mathscr{B} (\hat{\mathcal{P}}^{+})$. Then we can make $E$ a Polish space.

 We define the Skorokhod topology on $\mathfrak{P}(\hat{\mathcal{P}}^+)$, and define a corresponding metric on $E$ by requiring $\zeta$ to be an isometry.  Then, the space $E$ will be  Polish.

Thus, $\mu^N$ is an element of $D_E[0,\infty)$, i.e., is a map from $[0,\infty)$ into $E$ which
is right-continuous and has left-hand limits.  The space $D_E[0,\infty)$ will be endowed with the Skorohod metric, which we  denote by $d_E$, see \cite{EithierKurtz}.

Next, for each $f \in C^{\infty}(\hat{\mathcal{P}})$ define
\[
<f,\mu>_E=\int_{\hat{p} \in \hat{\mathcal{P}} } f(\hat{p}) \mu(d\hat{p}).
\]

In addition, define the generators
\begin{align}
\mathcal{L}_1 f(\hat{p})&=\frac{1}{2} \sigma^2 \lambda^{2\rho} \frac{\partial^2 f}{\partial \lambda^2} (\hat{p}) + b(\lambda,a) \frac{\partial f}{\partial \lambda}(\hat{p})-\lambda f(\hat{p})\nonumber\\
\mathcal{L}_2^x f(\hat{p})&=\frac{1}{2} (\beta^S)^2 \lambda^2 {\sigma_0}^2(x) \frac{\partial^2 f}{\partial \lambda^2} (\hat{p}) + \beta^S \lambda b_0(x) \frac{\partial f}{\partial \lambda}(\hat{p})\nonumber\\
\mathcal{L}_3^x f(\hat{p})&= \beta^S \lambda \sigma_0(x) \frac{\partial f}{\partial \lambda}(\hat{p})\nonumber\\
\mathcal{L}_4 f &= \beta^C \frac{\partial f}{\partial \lambda} (\hat{p})\nonumber\\
\iota(\hat{p})&=\lambda\nonumber\\
\nu(\hat{p})&=\ell\label{Eq:Operators}
\end{align}
	with $\beta^{C},\ell$ are vector valued, of the form $\beta^{C}=(\beta^{C}_1, \beta^{C}_2,\cdots,\beta^{C}_r)$ and  $\ell=(l_1, l_2,\cdots,l_r)$ respectively. We write $\nu_j(\hat{p})=l_j$ for $j=1, 2, \dots,r$.  After presenting Theorem \ref{T:MainTheorem} we shall elaborate on the meaning of the operators defined in (\ref{Eq:Operators}).

Recall that $\mathcal{V}_t=\sigma(V_s,0\leq s \leq t)\vee \mathcal{N}$, where $\mathcal{N}$ is the set of null sets. Introduce the notation
$\mathbb{E}_{\mathcal{V}_{t}}[\cdot]=\mathbb{E}[\cdot | \mathcal{V}_{t}]$.

Now, we are in position to state the main result of the paper.
\begin{theorem}\label{T:MainTheorem}
Let Assumptions \ref{A:Assumption2}-\ref{unbdd_b} hold. We have that $\mu^{N}_{\cdot}$ converges in distribution to the measure valued process $\bar{\mu}_{\cdot}$ with values in
$D_E[0,T]$. The evolution of
$\bar{\mu}_{\cdot}$ is given by the measure evolution equation
\begin{align}
d\la f,\bar \mu_t\ra_E &=  \left\{ \langle\mathcal{L}_1 f, \bar{\mu}_t\rangle_E + \langle\mathcal{L}_2^{X_t} f, \bar{\mu}_t\rangle_E +  \langle\iota\nu, \bar{\mu}_t\rangle_E \cdot \langle \mathcal{L}_4 f, \bar{\mu}_t\rangle_E \right\}dt\nonumber\\
&\qquad+\langle\mathcal{L}_3^{X_t} f, \bar{\mu}_t\rangle_E dV_t,\quad
\forall f\in C^\infty(\hat{\mathcal{P}}) \text{ a.s.}\label{Eq:LLN}
\end{align}
In addition, if  $(Q_i(t),\lambda_t^{*}(\hat{p}),i=1, \ldots,r)$, where $\hat{p}=(p,\lambda_{0})$, is the unique pair satisfying
	$$Q_i(t)=\int_{\hat{p} \in \hat{\mathcal{P}}} l_i \mathbb{E}_{\mathcal{V}_t} \left[\lambda_t^*(\hat{p})\exp\left[-\int_{0}^t \lambda_s^*(\hat{p})ds\right]\right]\pi(dp)\Lambda_0(d\lambda_{0}).$$
	$$\lambda_t^*(\hat{p})=\lambda_0+\int_{0}^t b(\lambda_s^*(\hat{p}),a)ds + \sigma\cdot(\lambda_s^*(\hat{p}))^{\rho} dW^*_s + \int_{0}^t \sum_{i=1}^r \beta_i^C Q_i(s) ds + \beta^S \int_{0}^t \lambda_s^*(\hat{p}) dX_s.$$
then for any $A \in \mathfrak{B}(\mathcal{P})$ and $B \in \mathfrak{B}({\mathbb{R}_{+}})$, $\bar{\mu}$ is given by
\begin{align}
\bar{\mu}_t(A \times B)&= \int_{\hat{p}\in \mathcal{P}} \chi_A(p) \mathbb{E}_{\mathcal{V}_t} \left[ \chi _B(\lambda_t^*(\hat{p}))\exp\left[-\int_{0}^t\lambda_s^*(\hat{p})ds \right] \right] \pi(dp) \Lambda_0(d\lambda_{0}).\label{Eq:LLN_Characterization}
\end{align}
\end{theorem}

\begin{proof}[Proof of Theorem \ref{T:MainTheorem}]
The ingredients of the proof are in Sections \ref{S:LimitCharacterization} and \ref{S:Uniqueness}. In Section \ref{S:LimitCharacterization} we state that the family $\{\mu^N\}_{N\in \mathbb{N}}$ is relatively compact (as a $D_E[0,\infty)$-valued random variable). Therefore $\{X,\mu^N\}_{N\in \mathbb{N}}$ is also relatively compact.   Hence, it (or a subsequence) converges in distribution to a stochastic process $(X,\bar{\mu})$. By the Skorokhod representation theorem, one can find a probability space and realizations, still denoted for convenience, $\{X,\mu^N\}_{N\in \mathbb{N}}$ and $(X,\bar{\mu})$ such that $(X,\mu^N)$ converges almost surely to $(X,\bar{\mu})$. By the calculations in Section \ref{S:LimitCharacterization} we obtain that $(X,\bar{\mu})$ will satisfy (\ref{Eq:LLN}). The results of Section \ref{S:Uniqueness} show that $\bar{\mu}$ is actually unique and  given by (\ref{Eq:LLN_Characterization}). The pair $(Q_i(t),\lambda_t^{*}(\hat{p}),i=1, \ldots,r)$ exists and is unique via Lemma \ref{uniqueQ}. These results complete the proof of the theorem.
\end{proof}

 The operator $\mathcal{L}_1 f$ represents the idiosyncratic risk of the default intensity and notice that a killing term $-\lambda f$ is also included due to the defaults. The operators $\mathcal{L}_2^{x} f$ and $\mathcal{L}_3^{x} f$ represent the effect of the exogenous risk $x=X_t$. The most intriguing term, perhaps, is the nonlinear term of the equation $\langle\iota\nu, \bar{\mu}_t\rangle_E \cdot \langle \mathcal{L}_4 f, \bar{\mu_t}\rangle_E$, which is the term responsible for the contagion effects and possible default clusters. In particular, as we shall also see in the numerical experiments of Section \ref{S:Numerics}, larger values of the contagion vector parameter (element-wise) $\beta^C$ lead to larger mean losses in the overall pool as well as in individual levels of interaction. In addition, the mean impact on given names from system wide defaults is larger when the associated contagion parameter $\beta^C$ is larger. The limiting term $\langle\iota\nu, \bar{\mu}_t\rangle_E \cdot \langle \mathcal{L}_4 f, \bar{\mu}_t\rangle_E$ is a sum of $r$ components, which shows the need to have $r$ bounded for the limiting procedure to go through.  Potential weakening of this is discussed in the Conclusions Section \ref{S:Conclusions}.

We end this section with a discussion on Theorem \ref{T:MainTheorem}.
\begin{remark}\label{R:UseOfTheorem}
Theorem \ref{T:MainTheorem} will be used in Section \ref{S:Numerics} to approximate dynamic quantities of interest such as $Q_{t}^{n,N}$, the overall loss rate in the pool $D_t^N = \frac{1}{N} \sum_{n=1}^N \chi_{\{\tau^{N,n} \leq t\}}$ or the loss within collections of names of the same type as $N\rightarrow\infty$.  To be more specific, let the $n$-th name be of type $p_{i}$. Then, Theorem \ref{T:MainTheorem} implies that as $N\rightarrow\infty$ one can approximate quantities like $Q_{t}^{n,N}$ and $D_t^N$  by the corresponding limiting objects $Q_{t}(p_{i})$ and $D_{t}$. Theorem \ref{T:MainTheorem} allows for a more efficient numerical computation because it replaces a system of SDE's, model (\ref{Eq:MainModel}), by a single limiting equation, (\ref{Eq:LLN}),  that can be efficiently computed (see Section \ref{S:Numerics} for details). The limit object (\ref{Eq:LLN}) is the weak formulation of a non-local PDE for the density of the measure $\bar{\mu}_{t}$, say $v(t,\hat{p})$. Due to its non-local form, it involves computation of integral terms coming from  the product term $\langle\iota\nu, \bar{\mu_t}\rangle_E \cdot \langle \mathcal{L}_4 f, \bar{\mu_t}\rangle_E$. In turn, the term $\langle\iota\nu, \bar{\mu_t}\rangle_E$ is an integral over the whole parameter space $\hat{\mathcal{P}}$ that also includes the vectors $\ell_{j}$, for $j=1,\cdots,r$, arising from the SVD. In order to compute the latter with an exogenously given adjacency matrix $A$ that has a large, but finite, dimension $N$, and its SVD, we approach the computation of the integral term $\langle\iota\nu, \bar{\mu_t}\rangle_E$ as a finite sum based on the empirical distribution of $\{\ell_{n,j}\}_{n\in\{1,\cdots, N\}}$, for each $j=1,\cdots, r$.  In Section \ref{S:Numerics} we make this precise on specific examples of interest and  collect the main findings of our numerical studies.
\end{remark}

\section{Numerical studies and simulation results}\label{S:Numerics}
In this section we demonstrate numerically the theoretical results of the paper. Before presenting the numerical studies, we first describe the numerical method that we follow and we also comment on general aspects and issues that are common in all examples.

One of the quantities that we are interested in is  the overall loss rate in the pool, defined by
\[
D_t^N = \frac{1}{N} \sum_{n=1}^N \chi_{\{\tau^{N,n} \leq t\}} = \frac{1}{N} \sum_{n=1}^N (1-M_t^{N,n}) = 1-\mu_t^N(\hat{\mathcal{P}}).
\]

Related to this quantity is also the loss rate for names of the same type, say type $B$, denoted by $p_B$:
\begin{align}
	D_t^{N}(p_B) &=\frac{1}{N_B} \sum_{n=1}^N \chi_{\{\tau^{N,n} \leq t\}} \chi_{\{p^{N,n}=p_B\}} = \frac{1}{N_B} \sum_{n=1}^N (1-M_t^{N,n})\chi_{\{p^{N,n}=p_B\}}\nonumber\\
	&=1- {\frac{N}{N_B}}\mu_t^N(\{\hat{p}:p=p_B\}),\nonumber
	\end{align}
where $N_B$ is the total number of names of type $B$ in the pool.

We are also interested in the mean impact on name $n\in\{1,\cdots, N\}$ from system wide defaults by time $t$, which is $Q_t^{N,n}$ defined by
$$Q_t^{N,n}=\beta_n^C \cdot L_t^N$$

with the contagion coefficient vector being
$$\beta_n^C = (\xi_1^2 u_{n1}, \xi_2^2 u_{n2}, \dots, \xi_r^2 u_{nr})$$
and the $r$-dimensional vector $L_t^{N}=(L_t^{N,1},L_t^{N,2},\dots,L_t^{N,r})$.

Recall that $Q_t^{N,n}$ can be interpreted as the mean increase of the $n$-th default intensity due to defaults of other banks by time $t$.

In order to be able to compute $Q_t^{N,n}$ we need to be  able to  compute $L_t^{N,j}$, which is associated to the   $j$-th level of interaction of the network
\begin{align}
L_t^{N,j} &= \frac{1}{N} \sum_{n=1}^N \ell_{n,j} \chi_{\{\tau^{N,n} \leq t\}}= \frac{1}{N}\sum_{n=1}^N \ell_{n,j} - \frac{1}{N}\sum_{n=1}^N \ell_{n,j} M_t^{N,n}\nonumber\\
&= \frac{1}{N}\sum_{n=1}^N \ell_{n,j} - \left<\nu_j, \mu_t^N\right>,\nonumber
\end{align}
where we recall that  $\nu_j(\hat{p})=l_j$.

The asymptotic result from Theorem \ref{T:MainTheorem} is used to evaluate network performance indicators such as $D_t^N$, $D_t^{N}(p_B)$ and $Q_t^{N,n}$. For large $N$,  quantities like $\mu_t^N(\hat{\mathcal{P}})$, $\mu_t^N(\{\hat{p}:p=p_B\})$ and $\left<\nu_j, \mu_t^N\right>$ are approximated by $\bar{\mu}_t (\hat{\mathcal{P}})$, $\bar{\mu}_t(\{\hat{p}:p=p_B\})$ and $\left<\nu_j, \bar{\mu}_t\right>$ respectively; made possible via Theorem \ref{T:MainTheorem}. In order to be able to numerically compute the latter quantities, we first  write $\bar{\mu}_{t}(d\hat{p})= v(t,\hat{p})d\hat{p}$ with $\hat{p}=(p,\lambda)$. A formal integration by parts on the stochastic evolution equation that $\bar{\mu}_{t}(d\hat{p})$ satisfies gives that, in distributional sense, the density satisfies
\begin{align}
&dv(t,\hat{p}) =\left\{ \mathcal{L}^*_1 v(t,\hat{p}) + \mathcal{L}^{*,X_t}_2 v(t,\hat{p}) + \sum_{j=1}^r \left(\int_{\hat{p}' \in \hat{\mathcal{P}} } \nu_j(\hat{p}')\iota (\hat{p}')v(t,\hat{p}') d\hat{p}'\right) ({\mathcal{L}}^{*,j}_4) v(t,\hat{p})\right\} dt\nonumber\\
&\qquad \qquad+ \mathcal{L}^{*,X_t}_3 v(t,\hat{p}) dV_t\nonumber\\
&v(0,\hat{p})=\left(\pi\times\Lambda_{0}\right)(\hat{p})\nonumber\\
&v(t,p,\lambda=0)=\lim_{\lambda \to \infty} v(t,p,\lambda)=0,\nonumber
\end{align}

where the adjoint operators are given by:
\begin{align}
\mathcal{L}^*_1 v(t,\hat{p})&= \frac{\partial^2 }{\partial \lambda^2} (\frac{1}{2} \sigma^2 \lambda^{2\rho}  v(t,\hat{p}) ) -\frac{\partial }{\partial \lambda}( b(\lambda,a)  v(t,\hat{p}) )-\lambda v(t,\hat{p}),\nonumber\\
\mathcal{L}^{*,x}_2  v(t,\hat{p})&= \frac{\partial^2 }{\partial \lambda^2} (\frac{1}{2} (\beta^S)^2 \lambda^2 {\sigma_0}^2(x)v(t,\hat{p})) -  \frac{\partial }{\partial \lambda}(\beta^S \lambda b_0(x)v(t,\hat{p})),\nonumber\\
\mathcal{L}_3^{*,x} v(t,\hat{p})&=  -\frac{\partial }{\partial \lambda}(\beta^S \lambda \sigma_0(x)v(t,\hat{p})),\nonumber\\
\mathcal{L}^{*,j}_4 v(t,\hat{p})& = -\beta^C_j \frac{\partial v(t,\hat{p})}{\partial \lambda}, j=1,2,\dots,r,\nonumber\\
\iota(\hat{p})&=\lambda,\nonumber\\
\nu(\hat{p})&=(\nu_1(\hat{p}),\dots,\nu_r(\hat{p}))=\ell=(l_1,\dots,l_r).\nonumber
\end{align}

Now, motivated by Theorem \ref{T:MainTheorem} we approximate,
\begin{align}
D_t^N \approx D_t &= 1-\bar{\mu}_t(\hat{\mathcal{P}}) = 1- \int_{\hat{p} \in \hat{\mathcal{P}}} \bar{\mu}_t (d\hat{p})
 = 1-\int_{p \in \mathcal{P}} \int_{\lambda=0}^{\infty} v(t,p,\lambda) d\lambda \ \ \pi(dp)\nonumber
\end{align}

\begin{align}
D_t^{N}(p_B) \approx D_t(p_B) &= 1- \kappa_B \bar{\mu}_t (\{\hat{p}:p=p_B\})  = 1 - \kappa_B \int_{\hat{p}: p=p_B} \bar{\mu}_t (d\hat{p})\nonumber\\
&= 1-  \int_{\lambda=0}^{\infty} v(t,p_B,\lambda) d\lambda \nonumber
\end{align}
where $\kappa_B = \lim_{N \to \infty} \frac{N}{N_B} =[\pi(\{p_B\})]^{-1}$ if the limit exists.

\begin{align}
L_t^{N,j} \approx  \frac{1}{N} \sum_{n=1}^N \ell_{n,j} - \left<\nu_j, \bar{\mu}_t\right> &= \frac{1}{N} \sum_{n=1}^N \ell_{n,j} - \int_{\hat{p} \in \hat{\mathcal{P}}} \nu_j(\hat{p}) \bar{\mu}_t (d\hat{p}) \nonumber\\
& = \frac{1}{N} \sum_{n=1}^N \ell_{n,j} - \int_{p \in \mathcal{P}} l_j \int_{\lambda=0}^{\infty}v(t,p,\lambda) d\lambda \ \ \pi(dp)\nonumber
\end{align}

Hence, it is  enough to be able to compute
$u_0(t,p)= \int_0^\infty  v(t,\hat{p}) d\lambda$. In order to do so, we first define the  $k$-th moment to be, see also \cite{GSSS},
\begin{equation}
u_k(t,p)= \int_0^\infty \lambda^k v(t,\hat{p}) d\lambda.\label{Eq:MomentCond}
\end{equation}

The moment $u_k(t,p)$ can be calculated from the evolution function of $d v(t,\hat{p})$, by multiplying it with $\lambda^k$ and integrating by parts over $[0,\infty)$. As it will become clearer in the examples that follow, $u_k(t,p)$ will satisfy a system of equations. However, this system is not a closed system in that for any $k\in\mathbb{N}$, $u_{k}$ depends on $u_{k+1}$. To resolve this, we follow the method of truncation and in particular for a large enough $K$, we set $u_{K+1}=u_{K}$ and then we solve backwards. As we shall see later on (see also \cite{GSSS} for related results)  this truncation is a sufficiently good and computationally efficient approximation of $u_0(t,p)= \int_0^\infty  v(t,\hat{p}) d\lambda$ and, in addition that $K$ can typically be taken to be small. Our numerical studies showed that choosing $K=20$ is more than sufficient to guarantee good approximation properties,  at least for the numerical examples studied in this paper. In addition, as it is demonstrated numerically in the Appendix A of \cite{SpiliopoulosSiriganoGiesecke2013}, in the simpler case without the network structure, such a mean field type of approximation is advantageous from numerical point of view as opposed to direct simulation of the finite $N$ system.

Now, if the number of levels of interaction $d$ is large or if the pool has a large degree of heterogeneity, then the number of equations $u_k(t,p)$ in the system can be prohibitively large. To resolve this and make the computation numerically feasible one can result in appropriate low-rank approximations as dictated by the SVD. The SVD facilitates the decomposition of the network interaction into
$r$ mean-field type levels of interaction. 

This singles out the contribution of the most important level of interaction. To support this claim further
note that the orthonormality of the vectors $\{u_j, j=1,\cdots,r\}$ and the definition $\beta^{C}_{n,j}=\xi^{2}_{j}u_{n,j}$ gives that for every $j=1,\cdots,r$
\[
\|\beta^{C}_{\cdot,j}\|_{2}= \xi^{2}_{j}\|u_{\cdot,j}\|_{2}=\xi^{2}_{j},
\]
 which immediately gives a ranking of $\|\beta^{C}_{\cdot,j}\|_{2}$ based on the eigenvalues $\xi_{j}$.

We will see the power of the low-rank approximation in the examples that follow. In particular,  if there is enough of spectral gap in the eigenvalues given by the SVD, then the limiting  loss rate $D_{t}$ as well as the limiting mean impact on a given name $n$, $Q^{n}_{t}$, are very well approximated by only considering the levels of interaction associated to the first few large eigenvalues and ignoring the rest.

Before presenting the numerical studies, let us collect here their main findings  and state some useful observations  (see Figures \ref{F:LimitingLoss1rank}-\ref{F:MeanImpactCorePeri1rankApprox_2lambdabar}):
\begin{itemize}
\item{A rank one approximation to $\Delta$ is a coarser approximation to the network structure than a rank two approximation in terms of the description of the intensity-to-default dynamics of the model (\ref{Eq:MainModel}). Similarly a rank two approximation is a coarser approximation to the network structure than a rank three approximation, and so on and so forth,  leading to a hierarchical structure.  This is a simple consequence of the SVD.}
\item{The ranking of the eigenvalues of $\Delta$, $\xi_{j}$, $j=1,\cdots, r,\cdots,
d$ gives a clear ranking of the importance of the different levels of interaction in explaining the heterogeneity of the pool. }
\item{The ranking  of the corresponding contagion parameter, $\beta^{C}_{n,j}$, gives a clear ranking of the mean impact on names belonging to the same level of interaction from system wide defaults.
    }
\item{Given that the other parameters of the model are the same, names of a type with larger value for $\beta^{C}_{n,j}$, the contagion coefficient, will have larger mean default rate than names of types with smaller value for $\beta^{C}_{n,j}$. This means that if the overall loss rate in the pool is large, signaling the existence of contagion clustering, names of types with large values for $\beta^{C}_{n,j}$ will be more prone to default if the rest of the parameters in the model description are the same.}
\item{Larger values of $\beta^{C}_{n,j}$ imply larger mean impact on the $n^{th}$ name from system wide defaults and, as we see in the subsequent sub sections, we are able to quantify this precisely, see for example the Figures in the example presented in Subsection \ref{Ex:Core-Periphery1}.}
\item{The  level of mean reversion $\bar{\lambda}$ also has an important effect on the losses experienced by names of the same type, see Example \ref{Ex:Core-Periphery2}.  Names with smaller mean reversion rate $\bar{\lambda}$ will be less likely to default.}
\item{In complicated networks with many different levels of interaction or high degree of heterogeneity, the numerical computation of quantities like $D_{t}$ or $Q_{t}$ can be prohibitively large. The singular value decomposition together with the limiting result Theorem \ref{T:MainTheorem} allow us to reduce the dimension of the system making such computations feasible, while maintaining accuracy, via low-rank approximations and large $N$ approximations. }
\end{itemize}

The effect of the exogenous risk component $X_t$ is quantified via the parameter $\beta^{S}$. As in \cite{GSSS} larger values of $\beta^S$ naturally lead to larger losses, due to an increase in the default intensities. Given that this effect here is analogous to what was observed in \cite{GSSS} and because in this paper our focus is on studying network effects through the contagion term, we do not study the effect of $\beta^S$ further here.

In all the numerical examples that follow, we consider for simplicity a specific form of function $b(\lambda,\alpha)=-\bar{\alpha}(\lambda-\bar{\lambda})$ and $\rho=1/2$, and take the systematic risk process to be a CIR process
$dX_t=\kappa(\theta-X_t)dt + \epsilon \sqrt{  X_t} dV_t$. For the numerical purposes of this paper, we have restricted attention to the aforementioned choices as we want to be able to compare and draw intuition from the existing literature which is largely based on the affine model (see also the Conclusions Section \ref{S:Conclusions} for related future directions).

We consider below four different numerical studies. The first example has one level of interaction, i.e. $d=1$ in the SVD, and the second example has two levels of interaction, i.e. $d=2$ in the SVD. The third and fourth  examples are motivated by the well documented core-periphery network structure for financial models, see for example \cite{CraigVonPeter2014,HojmanSzeidl2008}. In the third example all the names have the same mean-reversion coefficient. In example four, we choose different mean reversion coefficient for the core and for the periphery institutions. Notice that names of different types may belong to the same level of interaction. Namely each level of interaction does not need to be homogenous. This becomes clear in the specific examples below.  The matrix $\Delta$ for the core-periphery examples is chosen to reflect the empirical evidence \cite{CraigVonPeter2014,FrickeLux2015} that periphery banks are smaller and less active than core banks.

\subsection{One level of interaction case}

In this example, we consider a situation where the adjacency matrix $\Delta$ has only one positive eigenvalue. This corresponds to having one level of interaction, $d=r=1$, but of course the pool can still be heterogenous.

Let us start by fixing some values for the parameters  $\kappa=4$,  $\theta=0.5$, $\epsilon=0.5 $, $X_0=0.2$,
$\sigma=0.9$, $\bar{\alpha}=4$, $\bar{\lambda}=0.2$, $\lambda_0=0.2$ and $\beta^S=2$. Also, let us consider a pool of $N=1000$ names.

In addition, assume that 50\% of the $\beta^C_{n,1}$'s are taking the value $\beta_{1}^{C,1}=1.2361$ and the rest 50\% of the $\beta^C_{n,1}$'s are taking the value $\beta_1^{C,2}=0.6362$, while all $\ell_{n,1}$'s take value $l_1^1=0.0316$.  To describe this more effectively, we slightly abuse notation and consider discrete random variables $\tilde{\beta}_1^C$ and $\tilde{\ell}_1$ defined by
 \[
 \mathbb{P}(\tilde{\beta}_1^C=\beta_1^{C,1})=0.5, \quad \mathbb{P}(\tilde{\beta}_1^C=\beta_1^{C,2})=0.5 \text{ and } \mathbb{P}(\tilde{\ell}_1=l_1^1)=1.
 \]

The corresponding adjacency matrix $\Delta$ has a singular value decomposition with only one nonnegative eigenvalue 10. The first column of the left matrix takes one value 0.0316. The first column of the right matrix takes two values 0.12361 and 0.06362 with same frequencies. Notice that we indeed have $\beta_{1}^{C,1}=0.12361 \cdot 10=1.2361$ and $\beta_1^{C,2}=0.06362 \cdot 10=0.6362$, as expected.

Hence, we have a heterogeneous pool with two different types, where however both of them belong to the same level of interaction.

In this case, the moments, as defined by (\ref{Eq:MomentCond}) satisfy the following pair of coupled equations
\begin{align}
&d u_k(t,p_1) =\left\{u_k(t,p_1)(-\bar{\alpha} k + \beta^S \kappa(\theta-X_t) k + 0.5 (\beta^S)^2 \epsilon^{2} X_t k(k-1)) - u_{k+1}(t,p_1)\right\} dt \nonumber\\
 &\quad+ \left\{u_{k-1}(t,p_1) (0.5 \sigma^2 k(k-1) + \bar{\alpha} \bar{\lambda} k + l_1^1 \beta_1^{C,1} k (1/2u_1(t,p_1)+1/2u_1(t,p_2)))\right\}dt\nonumber\\
 &\quad+ \beta^S \epsilon\sqrt{ X_t} k u_k(t, p_1) dV_t\nonumber
\end{align}
\begin{align}
&d u_k(t,p_2) =\left\{u_k(t,p_2)(-\bar{\alpha} k + \beta^S \kappa(\theta-X_t) k + 0.5 (\beta^S)^2 \epsilon^{2} X_t k(k-1)) - u_{k+1}(t,p_2)\right\} dt \nonumber\\
&\quad + \left\{u_{k-1}(t,p_2) (0.5 \sigma^2 k(k-1) + \bar{\alpha} \bar{\lambda} k + l_1^1 \beta_1^{C,2} k (1/2u_1(t,p_1)+1/2u_1(t,p_2)))\right\}dt\nonumber\\
&\quad + \beta^S \epsilon \sqrt{ X_t} k u_k(t, p_2) dV_t\nonumber
\end{align}
with
$u_k(0,p)= \int_0^\infty \lambda^k (\pi\times \Lambda_0)(\hat{p})d\lambda$.

Then, we have that the overall loss rate, for large $N$, is
\[
D_t^N \approx D_t= 1- (1/2 u_0(t,p_1) + 1/2 u_0(t,p_2))
\]

The loss rate for type $p_i$,  $i=1 \text{ or } 2$ is
\[
D_t^{N}(p_i) \approx D_t(p_i) = 1- u_0(t,p_i),  \  \    i=1,2
\]

The mean impact, from the system wide defaults by time t, on name $n$, which comes from type $p_i$, $i=1 \text{ or } 2$ is
\[
Q_t^{N,n} \approx Q_t(p_i) = \beta_1^{C,i} L_t, \  \    i=1,2
\]

where
\[
L_t=l_1^1-l_1^1\left(1/2u_0(t,p_1)+1/2u_0(t,p_2)\right).
\]

Now notice that the system which the moments satisfy is a non-closed system, since the equation for the $k$-th moment depends on the $(k+1)$ moment. In order to solve this we truncate the system at a certain level $K$, by setting $u_{K}(t,p)=u_{K+1}(t,p)$ and solve backwards. This will then give us $u_1(t,p)$ and $u_0(t,p)$ for any time $t$. Here we choose the time endpoint to be $T=1$. We do the numerical iteration with time step being 0.01. We run 50,000 Carlo trials and plot the overall limiting loss $D_t$ at different truncation levels $K=5, 10, 20, 50$ in Figure \ref{F:LimitingLoss1rankKvary}. It is clear from  Figure \ref{F:LimitingLoss1rankKvary} that the results are visually indistinguishable for all those different truncation levels, meaning that the truncation mechanism is reliable even for a low level of truncation.
\begin{figure}[h!]
	\centering
	\includegraphics[scale=0.6]{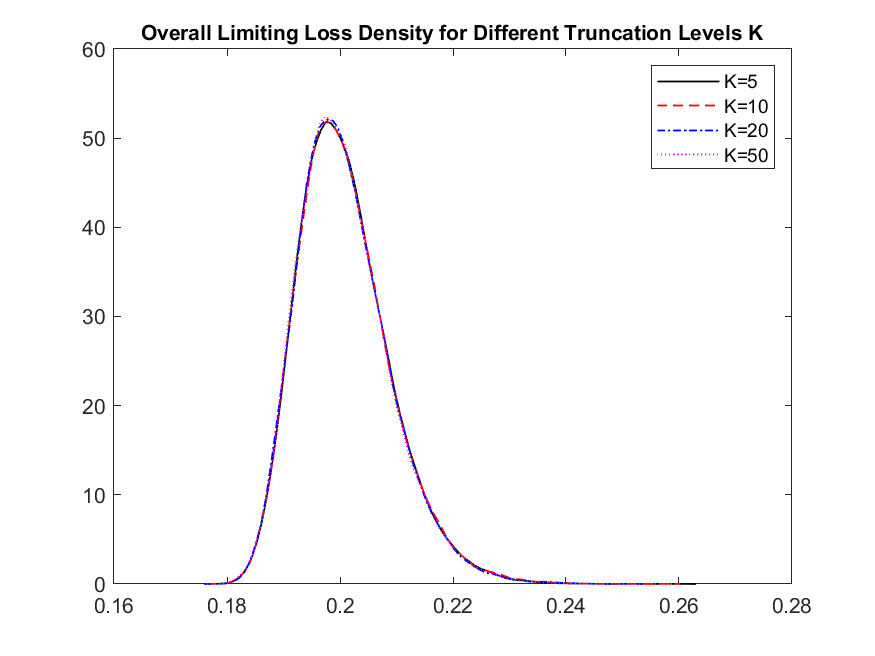}
	\small \caption{Density for overall limiting loss $D_T$ at different truncation levels $K=5, 10, 20, 50$.}\label{F:LimitingLoss1rankKvary}
\end{figure}

In the following experiments we will use $K=20$ with the same number of Monte Carlo trials. We plot the overall limiting loss $D_t$ and limiting loss for Type $p_i$, $D_t(p_i)$, $i=1,2$ in Figure \ref{F:LimitingLoss1rank} left plot. We also plot the empirical mean of overall limiting loss rate $D_T$ and the empirical mean of limiting loss rate for two types $D_T(p_{i})$, $i=1,2$, up to time $T=1$ in Figure \ref{F:LimitingLoss1rank} right plot.
\begin{figure}[h!]
	\centering
	\includegraphics[scale=0.6]{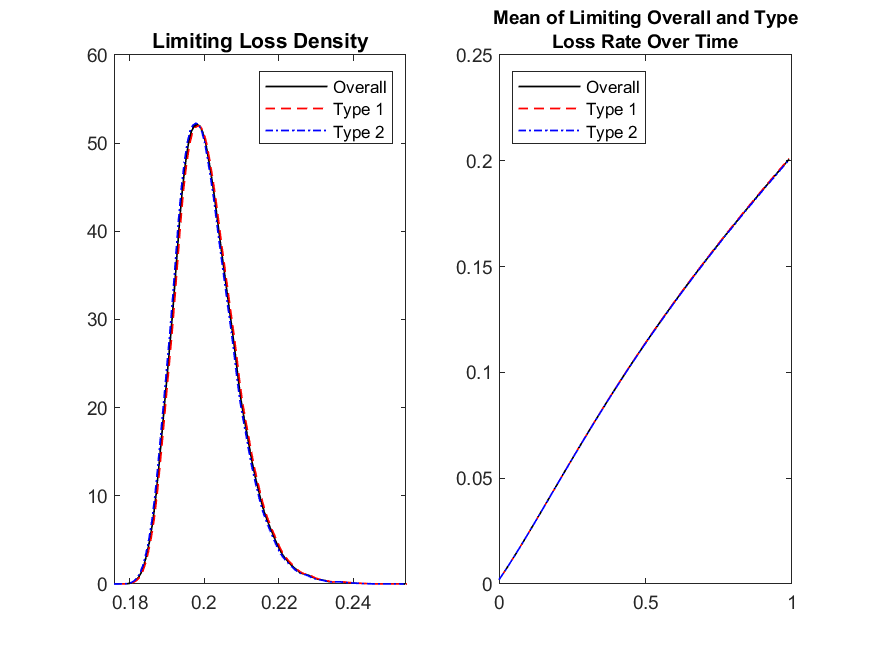}
	\small \caption{Left: Density for overall limiting loss $D_T$ and Type $i$ limiting loss for $D_T(p_i)$, $i=1,2$ at $T=1$; Right: Empirical mean of overall limiting loss $D_T$ and empirical mean of limiting loss for types $D_T(p_i)$, $i=1,2$ up to time $T=1$.}\label{F:LimitingLoss1rank}
\end{figure}

In Figure \ref{F:MeanImpact1rank}, we plot the mean impact on a name $n$, i.e., $Q_t(p^n)$, from system wide default as a function of time $t$ for the two different types of names. Here the name $n$, can be one of two types, type $1$ or type $2$, as indicated by the parameters $\beta^{C,1}_{1},\beta^{C,2}_{1}$. It is instructive to notice from the plots that $Q_t(p_1)\geq Q_{t}(p_2)$, which is to be expected due to the relation $\beta^{C,1}_{1}>\beta^{C,2}_{1}$ of the contagion coefficients.
\begin{figure}[h!]
	\centering
	\includegraphics[scale=0.6]{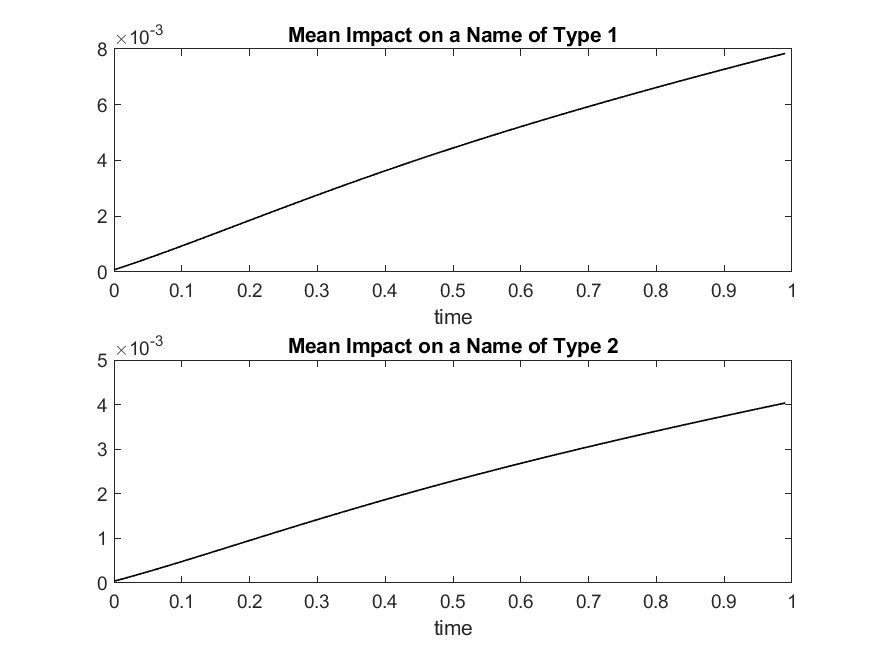}
	\small \caption{Mean impact on names of type 1 $Q_t(p_1)$ and type 2 $Q_t(p_2)$ from system wide default up to time $T=1$.} \label{F:MeanImpact1rank}
\end{figure}

\subsection{Two levels of interaction case}

In this example now we consider the case where $\Delta$ has two positive eigenvalues. This corresponds to having a heterogeneous pool with two levels of interaction, $d=2$. In this example, we will also test numerically the effect of the low-rank approximation on the limiting loss and on the mean impact on given names by system wide defaults.

Let us choose the following values for the parameters
 $\kappa=4$, $\theta=0.5$, $\epsilon=0.5$, $X_0=0.2$,
$\sigma=0.9$, $\bar{\alpha}=4$, $\bar{\lambda}=0.2$, $\lambda_0=0.2$ and $\beta^S=2$. Also, let us consider a pool of $N=1000$ names.

Furthermore, we assume that 50\% of the $\beta^{C}_{n,1}$'s (first level of interaction) are taking the value $\beta_{1}^{C,1}=0.2050$ and the rest 50\% of the $\beta^{C}_{n,1}$'s are taking $\beta_{1}^{C,2}=0.3980$. All the $l_{n,1}$'s take the value $l_{1}^1=0.0316$.

In addition, $2/3$ of the $\beta^{C}_{n,2}$'s (second level of interaction) are taking the value $\beta_{2}^{C,1}=0.0009$ and the rest 1/3 of the $\beta^{C}_{n,2}$'s are taking the value $\beta_{2}^{C,2}=0.0022$. Finally, 50\% of the $l_{n,2}$'s are taking the value $l_{2}^1=0.0043$ whereas the rest 50\% of the $l_{n,2}$'s are taking the value $l_{2}^2=-0.0022$.

As with the previous example, we slightly abuse notation and define discrete random variables $\tilde{\beta}_1^C$, $\tilde{\beta}_2^C$, $\tilde{\ell}_1$ and $\tilde{\ell}_2$ such that
 \begin{align}
 \mathbb{P}(\tilde{\beta}_1^C &= \beta_1^{C,1})=1/2, \quad \mathbb{P}(\tilde{\beta}_1^C = \beta_1^{C,2})=1/2,\nonumber\\
 \mathbb{P}(\tilde{\ell}_1 &= l_1^1)=1,\nonumber\\
  \mathbb{P}(\tilde{\beta}_2^C &= \beta_2^{C,1})=2/3,\quad \mathbb{P}(\tilde{\beta}_2^C = \beta_2^{C,2})=1/3\nonumber\\
  \mathbb{P}(\tilde{\ell}_2 &= l_2^1)=1/2,\quad  \mathbb{P}(\tilde{\ell}_2 = l_2^2)=1/2.\nonumber
  \end{align}

We assume that the random variables $\tilde{\beta}_1^C$, $\tilde{\beta}_2^C$, $\tilde{\ell}_1$, $\tilde{\ell}_2$ are independent.

For the corresponding adjacency matrix $\Delta$, the SVD has two nonnegative eigenvalues 10 and 1. The first column of the right matrix takes two values 0.0205 and 0.0398 with same frequencies. This indeed corresponds to the two values $\beta^{C,1}_{1}=0.0205 \cdot 10= 0.2050$ and $\beta^{C,2}_{1}=0.0398 \cdot 10 = 0.3980$.
The second column of the right matrix takes two values 0.0009 and 0.0022 with ratio of frequencies being 2:1. This indeed corresponds to the two values $\beta^{C,1}_{2}=0.0009 \cdot 1 = 0.0009$ and $\beta^{C,2}_{2}=0.0022 \cdot 1= 0.0022$.
The first column of the left matrix takes only one value 0.0316. The second column of the left matrix takes two values 0.0043 and -0.0022 with equal frequencies.

Let us now denote by $u_{k}(t;k_1,k_2,k_3)$ to be the $k$th moment at time t with $k_1,k_2,k_3\in\{1,2\}$ being the choice index for $\tilde{\beta}_1^C$, $\tilde{\beta}_2^C$ and $\tilde{\ell}_2$ respectively. For example, $k_1=1, k_2=1, k_3=2$ corresponds to the choice $\tilde{\beta}_1^C=\beta_{1}^{C,1}$, $\tilde{\beta}_2^C=\beta_{2}^{C,1}$ and $\tilde{\ell}_{2}=l_{2}^2$. Then there will totally be $2^3=8$ equations in the coupled system. However because of the special structure we end up with only 4 different equations. In particular, for $k_{1},k_{2},k_{3}\in\{1,2\}$ we have
\begin{align}
	&d u_{k}(t;k_1,k_2,k_3) =\left\{u_{k}(t;k_1,k_2,k_3)(-\bar{\alpha} k + \beta^S \kappa(\theta-X_t) k + 0.5 (\beta^S)^2 \epsilon^{2} X_t k(k-1)) \right. \nonumber\\
	&\quad \left.- u_{k+1}(t;k_1,k_2,k_3)\right\} dt+ u_{k-1}(t;k_1,k_2,k_3) \left\{ (0.5 \sigma^2 k(k-1) + \bar{\alpha} \bar{\lambda} k) + G_{k}(t;k_1,k_2) \right\}dt\nonumber\\
	&\quad+ \beta^S \epsilon\sqrt{ X_t} k u_{k}(t;k_1,k_2,k_3) dV_t\nonumber
\end{align}

Notice that $u_{k}(t;k_1,k_2,1)=u_{k}(t;k_1,k_2,2)$ for $k_{1},k_{2}=1,2$. We supplement $u_{k}(t;k_1,k_2,k_{3})$ with initial conditions
together with $u_{k}(0;k_1,k_2,k_3)= \int_0^\infty \lambda^k (\pi\times\Lambda_0)(\hat{p}) d\lambda$ and we define
\begin{align}
G_{k}(t;k_1,k_2) &= k l_1^1 \beta_1^{C,k_1} \sum_{i_1,i_2,i_3} u_1(t;i_1,i_2,i_3) \mathbb{P}(\tilde{\beta}_1^C=\beta_1^{C,i_1}, \tilde{\beta}_2^C=\beta_2^{C,i_2}, \tilde{\ell}_2=l_2^{i_3})\nonumber\\
&+ k \beta_2^{C,k_2} \sum_{i_1,i_2,i_3} l_2^{i_3}  u_1(t;i_1,i_2,k_3) \mathbb{P}(\tilde{\beta}_1^C=\beta_1^{C,i_1}, \tilde{\beta}_2^C=\beta_2^{C,i_2}, \tilde{\ell}_2=l_2^{i_3}),\nonumber
\end{align}
where $k_1,k_2=1,2$. Then we have that the overall loss rate is
\[
D_t^N \approx D_t = 1- \sum_{k_1,k_2,k_3} u_0(t;k_1,k_2,k_3) \mathbb{P}(\tilde{\beta}_1^C=\beta_1^{C,k_1}, \tilde{\beta}_2^C=\beta_2^{C,k_2},
\tilde{\ell}_2=l_2^{k_3}).
\]

The loss rate for type $(k_1,k_2,k_3)$, where $k_1,k_2,k_3=1,2$ essentially changes only with $k_1$ and $k_2$ and takes the form,
\[
D_t^{N}(k_1,k_2,1) \approx D_t(k_1,k_2,1) = 1-u_0(t;k_1,k_2,1),\quad D_t(k_1,k_2,1)=D_t(k_1,k_2,2).
\]

The mean impact on name $n$ from system wide defaults up to time $t$ is determined only via the choices for $k_1$ and $k_2$ through $\tilde{\beta}_1^C$ and $\tilde{\beta}_2^C$ respectively. In particular, we have
\[
Q_t^{N,n} \approx Q_t(k_1,k_2) = \beta_1^{C,k_1} L_t^1 + \beta_2^{C,k_2} L_t^2,
\]
where  for the $j-$th level of interaction, $j=1,2$, we have
$$L_t^1= l_1^1 - l_1^1 \sum_{k_1,k_2,k_3} u_0(t;k_1,k_2,k_3)\mathbb{P}(\tilde{\beta}_1^C=\beta_1^{C,k_1}, \tilde{\beta}_2^C=\beta_2^{C,k_2}, \tilde{\ell}_2=l_2^{k_3})$$
\begin{align}
L_t^2= &\sum_{k_3} l_2^{k_3} \mathbb{P}(\tilde{\ell}_2=l_2^{k_3}) \nonumber\\
& + \sum_{k_1,k_2,k_3} l_2^{k_3}  u_0(t;k_1,k_2,k_3)\mathbb{P}(\tilde{\beta}_1^C=\beta_1^{C,k_1}, \tilde{\beta}_2^C=\beta_2^{C,k_2}, \tilde{\ell}_2=l_2^{k_3})\nonumber
\end{align}

Due to the assumed independence, all the joint probabilities can be written as the product of marginals, for example,  $\mathbb{P}(\tilde{\beta}_1^C=\beta_1^{C,k_1}, \tilde{\beta}_2^C=\beta_2^{C,k_2}, \tilde{\ell}_2=l_2^{k_3})=\mathbb{P}(\tilde{\beta}_1^C=\beta_1^{C,k_1})\mathbb{P}(\tilde{\beta}_2^C=\beta_2^{C,k_2})\mathbb{P}(\tilde{\ell}_2=l_2^{k_3})$.

As with the previous example, we choose the time endpoint to be $T=1$. We do the numerical iteration with time step being 0.01. We run 50,000 Monte Carlo trials. In Figure \ref{F:LimitingLoss2rankKvary}, we show the densities for the overall limiting loss rate in the pool for different truncation levels $K=5, 10, 20, 50$. Again, the results are visually indistinguishable for all those different truncation levels, meaning that the truncation mechanism is reliable even for a low level of truncation.
\begin{figure}[h!]
	\centering
	\includegraphics[scale=0.6]{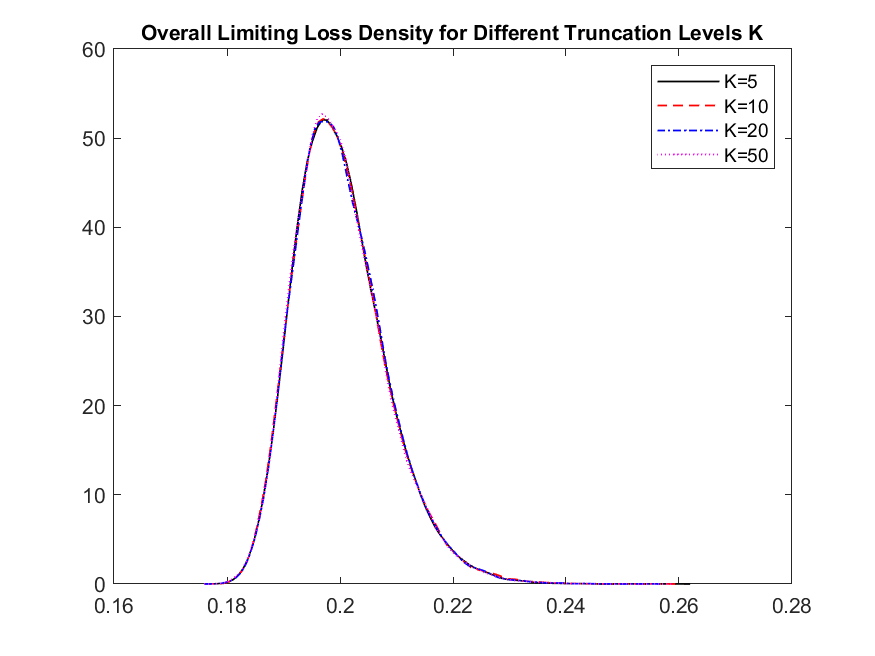}
	\small \caption{Density for overall limiting loss $D_T$ at different truncation levels $K=5, 10, 20, 50$.}\label{F:LimitingLoss2rankKvary}
\end{figure}

In the following experiments we still choose the truncation level $K=20$ and plot overall limiting loss rate $D_t$ and the limiting loss rate for different types $D_t(k_1,k_2)$, $k_1,k_1=1,2$ in the left plot of Figure \ref{F:LimitingLoss2rank}. We also plot the empirical mean of the overall limiting loss rate and the empirical mean of the loss rate $D_T$ for different types over time $D_T(k_1,k_2)$, $k_1,k_1=1,2$ in the right plot of Figure \ref{F:LimitingLoss2rank}.
\begin{figure}[h!]
	\centering
	\includegraphics[scale=0.6]{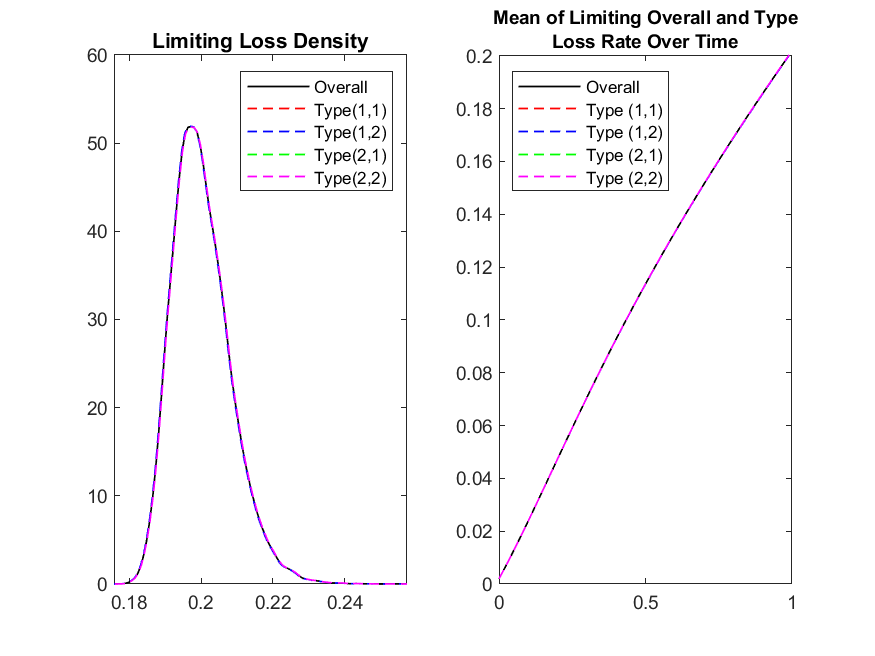}
	\small \caption{Left: Density for overall limiting loss $D_T$ and limiting loss for types $D_T(k_1,k_2)$ at $T=1$; Right: Empirical mean of overall limiting loss $D_T$ and empirical mean of limiting loss for types $D_T(k_1,k_2)$ up to time $T=1$.}\label{F:LimitingLoss2rank}
\end{figure}

In Figure \ref{F:MeanImpact2rank} we plot the mean impact on a name from type $(k_1,k_2)$, $k_1,k_2=1,2$ due to system wide defaults up to  time $T=1$ .
\begin{figure}[h!]
	\centering
	\includegraphics[scale=0.6]{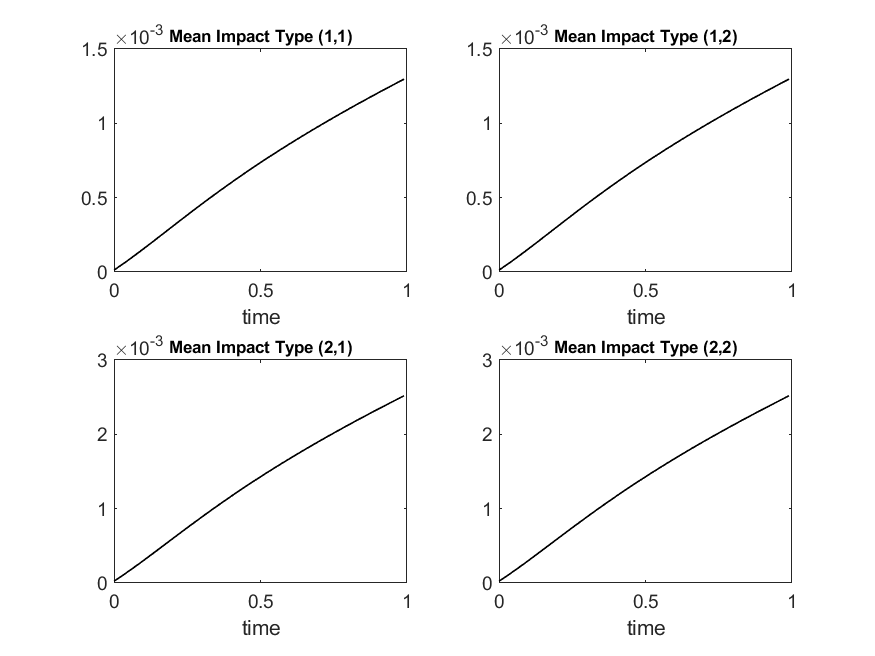}
	\small \caption{Mean impact on names of type $(k_1,k_2)$, $Q_t(k_1,k_2)$, from system wide default by up to time $T=1$.}\label{F:MeanImpact2rank}
\end{figure}

As we discussed in the beginning of this section,
the SVD facilitates the decomposition of the network interaction into
$r$ mean-field type levels of interaction.

We test the effect of the low-rank approximation by only keeping the first level of interaction. This singles out the contribution of the most important level of interaction.

 In other words, we replace  $\Delta$ by
$$A=\Delta_{\textrm{aprrox}}=\xi_1^2 \ell_1 u_1^T$$
which reduces the problem to a one level of interaction problem.
Comparing the overall limiting loss that we get from the two level of interaction case  $D_t$ and its first level of interaction approximation $D_{\textrm{approx},t}$, see left plot of Figure \ref{F:LimitingLossOverall2rank}, we get that the distribution of the limiting loss processes are practically indistinguishable. Similar conclusion can be made from the right plot of Figure \ref{F:LimitingLossOverall2rank}, where we plot the empirical mean of overall limiting loss rate over time in the two level of interaction example $D_T$ and it first-level of interaction approximation $D_{\textrm{approx},T}$. These in turn imply that the second level of interaction can be neglected for the purposes of these computations.
\begin{figure}[h!]
	\centering
	\includegraphics[scale=0.6]{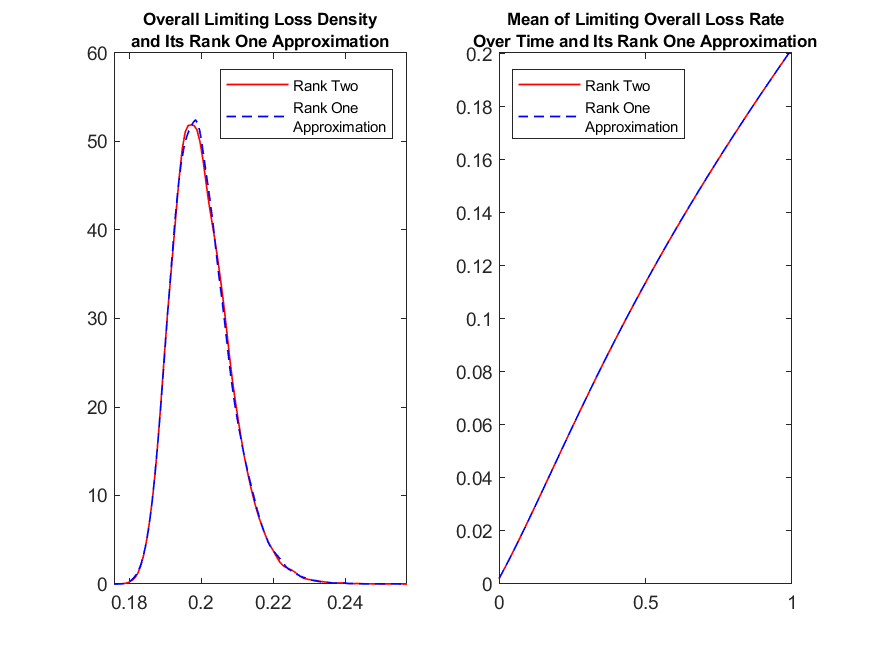}
	\small \caption{Left: Density for overall limiting loss $D_T$ and overall limiting loss $D_{\textrm{approx},T}$ from its rank one approximation at $T=1$; Right: Empirical mean of overall limiting loss $D_T$ and empirical mean of limiting loss in its rank one approximation $D_{\textrm{approx},T}$ up to time $T=1$.}\label{F:LimitingLossOverall2rank}
\end{figure}

Lastly, we investigate how the mean impact on a name from system wide defaults for the two level of interaction case and its one level of interaction approximated version compare. In Figure \ref{F:MeanImpact2rank}, we see that the mean impact on given names  depends mainly on $\tilde{\beta}_1^{C}$, and not so much on $\tilde{\beta}_2^{C}$. This will be further verified in the one level of interaction approximation case, where we calculate the approximated mean impacts on these two types by using the information only from first entries of $\beta^C$ and $L_t^N$, i.e., by using only the information from the first level of interaction, shown in Figure \ref{F:MeanImpact2rankApprox}.
\begin{figure}[h!]
	\centering
	\includegraphics[ scale=0.6]{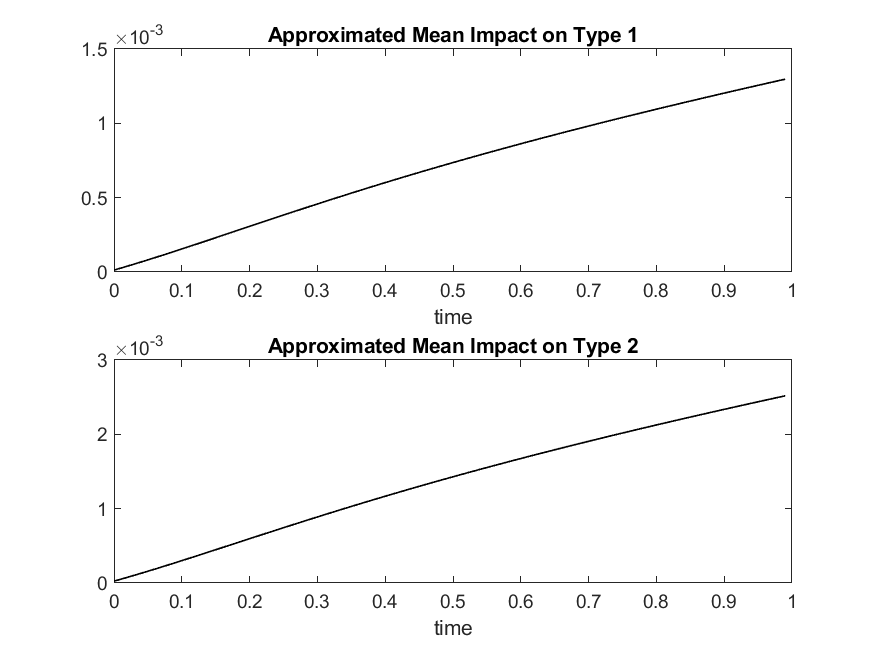}
	\small \caption{Approximated mean impact on names of different types from system wide default by time t in the coarse-grained case.}\label{F:MeanImpact2rankApprox}
\end{figure}
$$Q_{\text{approx},t}^{N,n}(p_i) \approx \beta_{1}^{C,i} L_{\text{approx},t}^{1}, \text{ for } i=1,2.$$

Comparing Figures \ref{F:MeanImpact2rank} and \ref{F:MeanImpact2rankApprox} we see that the first level of interaction, which has the largest eigenvalue,  indeed captures the behavior on the mean impact on a given name of type defined by $\tilde{\beta}_1^C$. In addition, notice that the mean default impact on names of type 2 is larger than the mean default impact on names of type 1 for all $t\in[0,1]$. This is to be expected due to the relation $\beta_{1}^{C,2}>\beta_{1}^{C,1}$.

\subsection{Core-Periphery example one: homogeneous mean-reverting coefficient}\label{Ex:Core-Periphery1}

A reasonably realistic model for financial related applications is the core-periphery case, see for example \cite{CraigVonPeter2014,HojmanSzeidl2008}. In a core-periphery model, one has a few names that constitute the core of the network and considerably depend on each other, in a sense forming the most influential part of the network, and the periphery which is composed by the rest of the names in the pool which depend less on each other.  Core institutions borrow from, and lend to, at least one institution in the periphery.

Motivated by this structure, let us consider the case of  $N=1000$ names and an appropriate adjacency matrix $\Delta$. For illustration purposes a $10\times 10$ block of $\Delta$ is given by:
$$\Delta_{10\times 10}=
\begin{bmatrix}
0  &  10  &  1  &  10  &  10  &  1  &  10  &  1  &  1  &  10 \\
10  &  0  &  1  &  1  &  10  &  10  &  10  &  1  &  10  &  1 \\
1  &  1  &  0  &  1  &  1  &  1  &  1  &  1  &  1  &  1 \\
5  &  1  &  1  &  0  &  1  &  1  &  1  &  1  &  1  &  1 \\
5  &  5  &  1  &  1  &  0  &  1  &  1  &  1  &  1  &  1 \\
1  &  5  &  1  &  1  &  1  &  0  &  1  &  1  &  1  &  1 \\
5  &  1  &  1  &  1  &  1  &  1  &  0  &  1  &  1  &  1 \\
1  &  1  &  1  &  1  &  1  &  1  &  1  &  0  &  1  &  1 \\
1  &  5  &  1  &  1  &  1  &  1  &  1  &  1  &  0  &  1 \\
1  &  1  &  1  &  1  &  1  &  1  &  1  &  1  &  1  &  0 \\
\end{bmatrix}
$$

  The SVD for such a  matrix gives 5 eigenvalues 1029, 143, 137.8, 59.9 and 58.5 significantly larger than the rest, with the first one being dominantly big. Therefore, motivated by the low rank approximation, we can use the first few levels of interaction to approximate the behavior of the network.

\subsubsection{One level of interaction approximation for core-periphery}
Let us choose the first eigenvalue to do the low rank approximation. Similarly to what was done for the previous examples, we define discrete random variables $\tilde{\beta}_1^C$ and $\tilde{\ell}_1$ taking values from the SVD with corresponding relative frequencies. It turns out that the SVD composition yields six different values for $\tilde{\beta}_1^C$ and three different values for $\tilde{\ell}_1$. We record the values  in Tables \ref{T:betaC1} and \ref{T:l1} respectively.
\begin{table}[h!]
	\begin{center}
		\begin{tabular}{c|c|c|c|c|c|c|c}
			$\tilde{\beta}_1^C$ & $\beta^{C,1}_{1}$ & $\beta^{C,2}_{1}$ & $\beta^{C,3}_{1}$ & $\beta^{C,4}_{1}$ & $\beta^{C,5}_{1}$ & $\beta^{C,6}_{1}$ & \\
			\hline
            value  & 31.0514 & 32.4883 & 32.5136 & 33.9505 & 73.6927 & 74.4088  \\
		\end{tabular}
			\caption{Possible values for $\tilde{\beta}_1^C$.}
		    \label{T:betaC1}
	\end{center}
\end{table}
\begin{table}[h!]
	\begin{center}
		\begin{tabular}{c|c|c|c|c}
			$\tilde{\ell}_1$ & $l_{1}^1$ & $l_{1}^2$ & $l_{1}^3$ & \\
			\hline
			value & 0.0308 & 0.1597 & 0.1625  \\
		\end{tabular}
		\caption{Possible values for $\tilde{\ell}_1$.}
		\label{T:l1}
	\end{center}
\end{table}

Let us choose the following values for the parameters
 $\kappa=4$, $\theta=0.5$, $\epsilon=0.5$, $X_0=0.2$,
$\sigma=0.9$, $\bar{\alpha}=4$, $\bar{\lambda}=0.2$, $\lambda_0=0.2$ and $\beta^S=2$.

Let us denote by $u_{k}(t;k_1,k_2)$ to be the $k$-th moment at time t with $k_1 \in\{1,2,\dots,6\}$ and $k_2 \in \{1,2,3\}$ being the  index choice for $\tilde{\beta}_1^C$, and $\tilde{\ell}_1$ respectively. For example, $k_1=1, k_2=2$ corresponds to the choice $\tilde{\beta}_1^C=\beta_{1}^{C,1}$ and $\tilde{\ell}_1=l_{1}^2$. The empirical joint distribution of $\tilde{\beta}_1^C$ and $\tilde{\ell}_1$ is summarized as follows.
\begin{table}[h!]
	\begin{center}
		\begin{tabular}{c|c|c|c}
			$k_1$ & $k_2$ & probability & \\
			\hline
			6 & 3 & 0.001 \\
			5 & 2 & 0.001 \\
			4 & 1 & 0.227 \\
			3 & 1 & 0.238 \\
			2 & 1 & 0.228 \\
			1 & 1 & 0.305 \\
		\end{tabular}
		\caption{Joint distribution for $\tilde{\beta}_1^C$ and $\tilde{\ell}_1$.}
		\label{T:jointrank1}
	\end{center}
\end{table}

In general there would have been  in total  $6 \times 3=18$ equations in the coupled system. However, because of the special structure we end up with only 6 different equations.  Based on the available combinations of $k_1,k_2$ as indicated in Table \ref{T:jointrank1} we have
\begin{align}
	&d u_{k}(t;k_1,k_2) =\left\{u_{k}(t;k_1,k_2)(-\bar{\alpha} k + \beta^S \kappa(\theta-X_t) k + 0.5 (\beta^S)^2 \epsilon^{2} X_t k(k-1)) \right. \nonumber\\
	&\quad \left.- u_{k+1}(t;k_1,k_2)\right\} dt+ u_{k-1}(t;k_1,k_2) \left\{ (0.5 \sigma^2 k(k-1) + \bar{\alpha} \bar{\lambda} k) + G_{k}(t;k_1) \right\}dt\nonumber\\
	&\quad+ \beta^S \epsilon\sqrt{ X_t} k u_{k}(t;k_1,k_2) dV_t,\nonumber
\end{align}
together with $u_{k}(0;k_1,k_2)= \int_0^\infty \lambda^k (\pi\times\Lambda_0)(\hat{p}) d\lambda$
and where we define
$$G_k(t;k_1)=\left( \sum_{i_1,i_2} l_{1}^{i_2}  u_1(t;i_1,i_2) \mathbb{P}(\tilde{\beta}_1^C=\beta_1^{C,i_1},\tilde{\ell}_1=l_1^{i_2}) \right) k \beta_1^{C,k_1}.$$

In particular $u_k(t;k_1,k_2)$ is only affected by the index $k_1$ through $G_k(t;k_1)$. The overall loss rate in the one-level of interaction approximation is
$$
D_{1\text{approx},t}^N \approx D_{1\text{approx},t} = 1- \sum_{k_1,k_2} u_0(t;k_1,k_2) \mathbb{P}(\tilde{\beta}_1^C=\beta_1^{C,k_1},\tilde{\ell}_1=l_1^{k_2}).
$$

The loss rate for type $(k_1,k_2)$ where $k_1=1,2,\dots,6$ and $k_2=1,2,3$ in the one-level of interaction approximation are actually falling into 6 distinct categories indexed by $k_1$, the choice of $\tilde{\beta}_1^C$.
\[
D_{1\text{approx},t}^{N}(k_1,k_2) \approx D_{1\text{approx},t}(k_1,k_2) = 1- u_0(t;k_1,k_2).
\]

The mean impact, from system wide defaults up to time t, on name $n$,   turns out to be characterized only by the first index $k_1$
\[
Q_{1\text{approx},t}^{N,n}(k_1,k_2) \approx Q_{1\text{approx},t}(k_1) = \beta_1^{C,k_1} L_{1\text{approx},t},
\]
for any $k_2=1,2,3$ with
\begin{align}
	L_{1\text{approx},t} = & \sum_{k_2} l_1^{k_2} \mathbb{P}(\tilde{\ell}_1=l_1^{k_2})\nonumber\\
	& - \sum_{k_1,k_2} l_1^{k_2}  u_0(t;k_1,k_2) \mathbb{P}(\tilde{\beta}_1^C=\beta_1^{C,k_1},\tilde{\ell}_1=l_1^{k_2})\nonumber
\end{align}

As with the previous two examples, we truncate at the level $K=20$, and choose the time endpoint to be $T=1$. We do the numerical iteration with time step being 0.01. We run 50,000 Monte Carlo trials and plot overall limiting loss rate $D_{1\text{approx},t}$ and the limiting loss rate for different types $D_{1\text{approx},t}^{k_1}$, $k_1=1,2,\dots,6$ in Figure \ref{F:LimitingLossCorePeriphery1rank}. Notice how the mean of the distribution shifts to the right as the value for $k_1$ increases, indicating an increase to the value that the random variable $\tilde{\beta}^{C}_{1}$ takes. We plot the mean of the loss rate over time for the whole pool and for individual types  in Figure \ref{F:LimitingLossOverTimeCorePeriphery1rank}. We observe that the plot indicates larger losses as the value for $k_1$ increases, signaling that names with large value for $\beta^{C}_{1}$ will be more likely to default and thus contribute more to a potential default clustering event.
\begin{figure}[h!]
	\centering
	\includegraphics[scale=0.6]{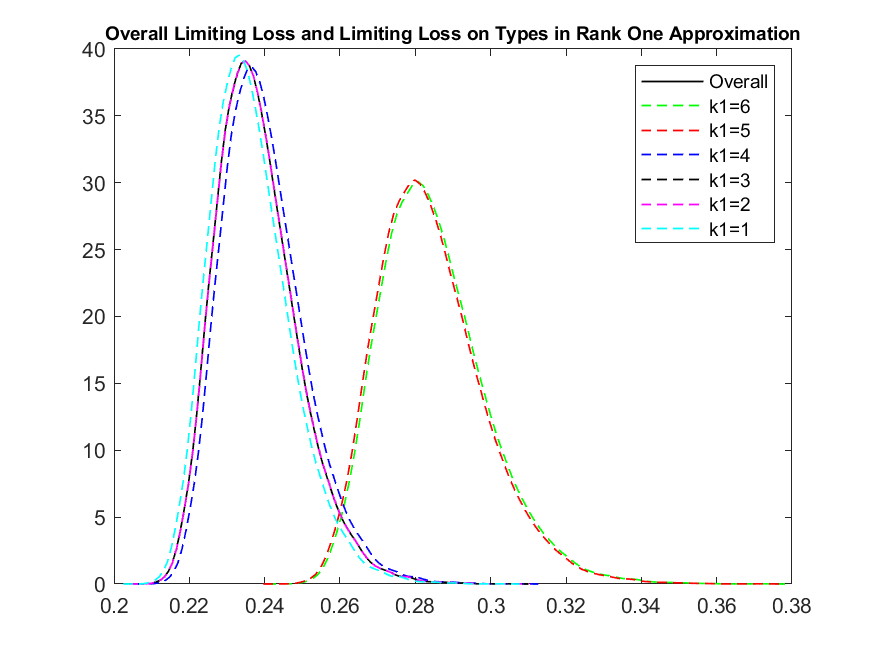}
	\small \caption{Density for overall limiting loss $D_{1\text{approx},T}$ and limiting loss for types $D_{1\text{approx},T}(k_1)$ at $T=1$.}\label{F:LimitingLossCorePeriphery1rank}
\end{figure}

\begin{figure}[h!]
	\centering
	\includegraphics[scale=0.6]{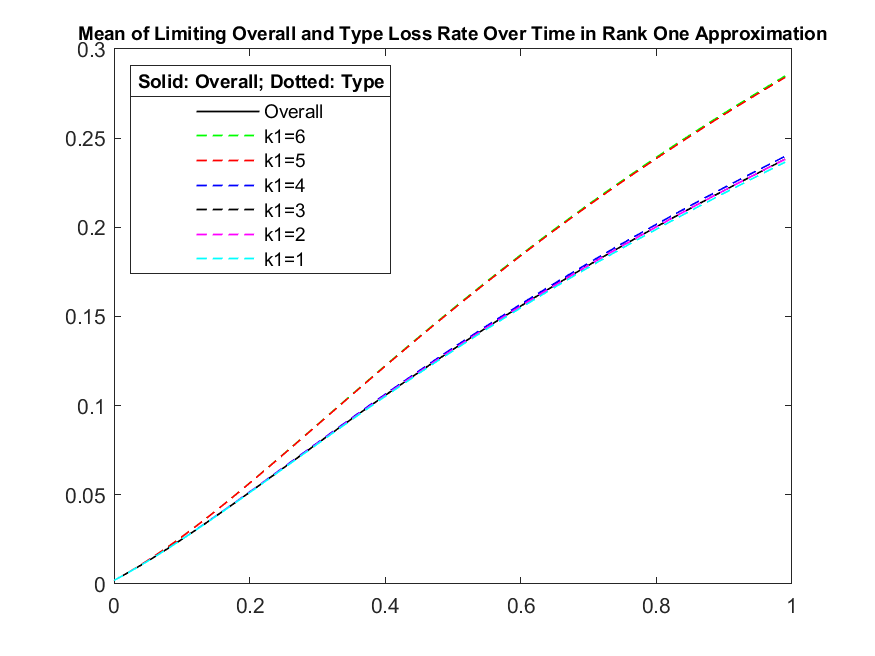}
	\small \caption{Empirical mean of overall limiting loss $D_{1\text{approx},T}$ and empirical mean of limiting loss for types $D_{1\text{approx},T}(k_1)$  up to time $T=1$}\label{F:LimitingLossOverTimeCorePeriphery1rank}
\end{figure}

In Figure \ref{F:MeanImpactCorePeri1rankApprox}, we plot the mean impact on a name from system wide defaults up to time $t$. There are totally 6 different categories indexed by $k_1$, the choice of $\tilde{\beta}_1^C$, as we discussed before.
\begin{figure}[h!]
	\centering
	\includegraphics[scale=0.6]{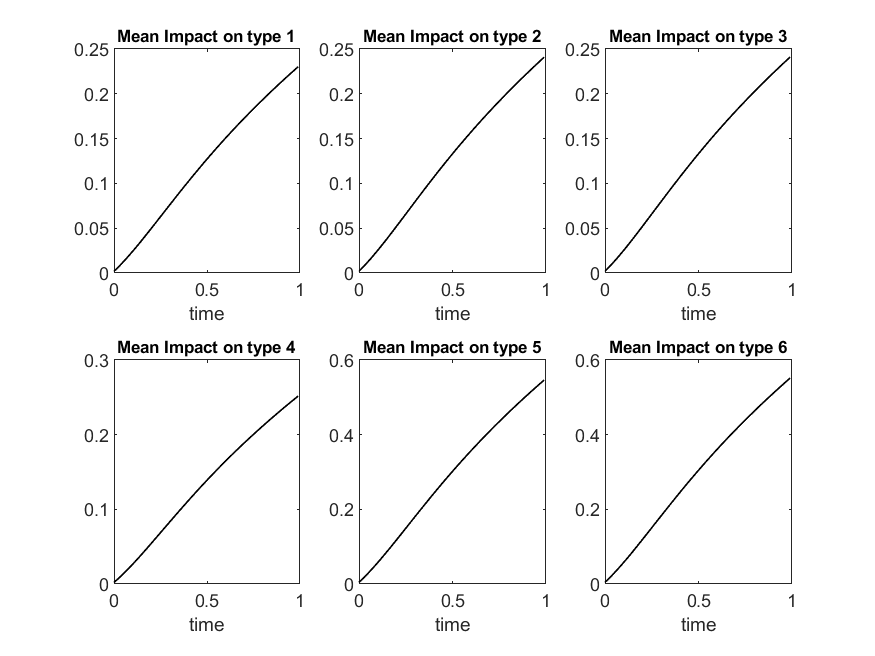}
	\small \caption{Mean impact on names of different types from system wide default by time t for the Core-Periphery case approximated by the first rank.}\label{F:MeanImpactCorePeri1rankApprox}
\end{figure}

\subsubsection{Two levels of interaction approximation for core-periphery}

Let us now investigate the core-periphery case by doing a low rank approximation based on the first two levels of interaction. From the SVD decomposition, the second largest eigenvalue is 143. Below, we summarize the empirical distributions of coefficients from the second columns of the matrices from the SVD decomposition. Table \ref{T:betaC2} is for coefficient $\tilde{\beta}_2^C$ and Table \ref{T:l2} is for coefficient $\tilde{\ell}_2$.
\begin{table}[h!]
	\begin{center}
		\begin{tabular}{c|c|c|c|c|c|c|c|c|c|c}
			$\tilde{\beta}_2^C$ & $\beta^{C,1}_{2}$ & $\beta^{C,2}_{2}$ & $\beta^{C,3}_{2}$ & $\beta^{C,4}_{2}$ & $\beta^{C,5}_{2}$ & $\beta^{C,6}_{2}$ & $\beta^{C,7}_{2}$ & $\beta^{C,8}_{2}$ & $\beta^{C,9}_{2}$ &\\
			\hline
            value & -12.7072 & -12.1454 & -5.7944 & 0.2753 & 0.2777 & 0.5080 & 0.5105 & 6.5777 & 6.5801 \\
		\end{tabular}
			\caption{Possible values for $\tilde{\beta}_2^C$.}
		    \label{T:betaC2}
	\end{center}
\end{table}

\begin{table}[h!]
	\begin{center}
		\begin{tabular}{c|c|c|c|c|c|c}
			$\tilde{\ell}_2$ & $l_{2}^1$ & $l_{2}^2$ & $l_{2}^3$ & $l_{2}^4$ & $l_{2}^5$ &\\
			\hline
			value & -0.0107 & -0.0081 & -0.0054 & 0.6674 & 0.7002\\
		\end{tabular}
		\caption{Possible values for $\tilde{\ell}_2$.}
		\label{T:l2}
	\end{center}
\end{table}

Let us now denote by $u_{k}(t;k_1,k_2,k_3,k_4)$ to be the $k$-th moment by time t with $k_1 \in\{1,2,\dots,6\}$, $k_2 \in\{1,2,\dots,9\}$, $k_3 \in \{1,2,3\} $, and $k_4\in\{1,\dots,5\}$ being the  index choice for $\tilde{\beta}_1^C$, $\tilde{\beta}_2^C$, $\tilde{\ell}_1$ and $\tilde{\ell}_2$ respectively. For example, $k_1=1, k_2=1, k_3=2, k_4=1$ corresponds to the choice $\tilde{\beta}_1^C=\beta_{1}^{C,1}$, $\tilde{\beta}_2^C=\beta_{2}^{C,1}$, $\tilde{\ell}_1=l_{1}^2$ and $\tilde{\ell}_2=l_{2}^1$.

The empirical joint distribution of $\tilde{\beta}_1^C$,  $\tilde{\beta}_2^C$, $\tilde{\ell}_1$ and $\tilde{\ell}_2$ is summarized as follows.
\begin{table}[h!]
	\begin{center}
		\begin{tabular}{c|c|c|c|c|c}
			$k_1$ & $k_2$ & $k_3$ & $k_4$ &probability & \\
			\hline
			6 & 1 & 3 & 5 & 0.001 \\
			5 & 2 & 2 & 4 & 0.001 \\
			4 & 9 & 1 & 2 & 0.089 \\
			4 & 9 & 1 & 1 & 0.120 \\
			4 & 8 & 1 & 3 & 0.018 \\
			3 & 7 & 1 & 2 & 0.171 \\
			3 & 6 & 1 & 3 & 0.067 \\
			2 & 5 & 1 & 2 & 0.172 \\
			2 & 4 & 1 & 3 & 0.056 \\
			1 & 3 & 1 & 3 & 0.305 \\
		\end{tabular}
		\caption{Joint distribution for $\tilde{\beta}_1^C$,  $\tilde{\beta}_2^C$, $\tilde{\ell}_1$ and $\tilde{\ell}_2$.}
		\label{T:jointrank2}
	\end{center}
\end{table}

In general there would have been in total $6\times9\times3\times5=810$ equations in the coupled system. However, because of the special structure we end up with only $10$ different equations. Based on the allowable choices for $k_1,k_2,k_3,k_4$  as indicated in Table \ref{T:jointrank2} we have
\begin{align}
	&d u_{k}(t;k_1,k_2,k_3,k_4) =\left\{u_{k}(t;k_1,k_2,k_3,k_4)(-\bar{\alpha} k + \beta^S \kappa(\theta-X_t) k + 0.5 (\beta^S)^2 \epsilon^{2} X_t k(k-1)) \right. \nonumber\\
	&\quad \left.- u_{k+1}(t;k_1,k_2,k_3,k_4)\right\} dt+ u_{k-1}(t;k_1,k_2,k_3,k_4) \left\{ (0.5 \sigma^2 k(k-1) + \bar{\alpha} \bar{\lambda} k) + G_{k}(t;k_1,k_2) \right\}dt\nonumber\\
	&\quad+ \beta^S \epsilon\sqrt{ X_t} k u_{k}(t;k_1,k_2,k_3,k_4) dV_t\nonumber
\end{align}
together with $u_{k}(0;k_1,k_2,k_3,k_4)= \int_0^\infty \lambda^k (\pi\times\Lambda_0)(\hat{p})d\lambda$ where we have defined
\begin{align}
	&G_k(t;k_1,k_2)\nonumber\\
	 = & k \beta_1^{C,k_1} \left[ \sum_{i_1,i_2,i_3,i_4} l_1^{i_3}  u_1(t;i_1,i_2,i_3,i_4) \mathbb{P}(\tilde{\beta}_1^C=\beta_1^{C,i_1},\tilde{\beta}_2^C=\beta_2^{C,i_2},\tilde{\ell}_1=l_1^{i_3},\tilde{\ell}_2=l_2^{i_4})\right]\nonumber\\
	 & + k \beta_2^{C,k_2} \left[ \sum_{i_1,i_2,i_3,i_4} l_2^{i_4}  u_1(t;i_1,i_2,i_3,i_4) \mathbb{P}(\tilde{\beta}_1^C=\beta_1^{C,i_1},\tilde{\beta}_2^C=\beta_2^{C,i_2},\tilde{\ell}_1=l_1^{i_3},\tilde{\ell}_2=l_2^{i_4})\right].\nonumber
\end{align}

In particular, $u_{k}(t;k_1,k_2,k_3,k_4)$ is only affected by the choices of $k_1,k_2$ through $G_k(t;k_1,k_2)$. The overall loss rate is
\begin{align}
D_{2\text{approx},t}^N \approx D_{2\text{approx},t} = &1- \sum_{k_1,k_2,k_3,k_4}  u_0(t;k_1,k_2,k_3,k_4)\nonumber\\
&\quad \quad \cdot \mathbb{P}(\tilde{\beta}_1^C=\beta_1^{C,k_1},\tilde{\beta}_2^C=\beta_2^{C,k_2},\tilde{\ell}_1=l_1^{k_3},\tilde{\ell}_2=l_2^{k_4})\nonumber
\end{align}
$$D_{2\text{approx},t}^{N}(k_1,k_2,k_3,k_4) \approx D_{2\text{approx},t}(k_1,k_2,k_3,k_4) = 1-u_0(k_1,k_2,k_3,k_4).$$

The mean impact on name $n$ from type $(k_1,k_2,k_3,k_4)$, where $k_1=1,2,\dots,6$, $k_2=1,2,\dots,9$, $k_3=1,2,3$ and $k_4=1,2,\dots,5$, is again determined by the choice $k_1$ and $k_2$ for $\tilde{\beta}_1^C$ and $\tilde{\beta}_2^C$ respectively
$$Q_{2\text{approx},t}^{N,n}(k_1,k_2,k_3,k_4) \approx Q_{2\text{approx},t} (k_1,k_2) = \beta_1^{C,k_1} L_{2\text{approx},t}^1 + \beta_2^{C,k_2} L_{2\text{approx},t}^2,$$
where  for the $j$-th level of interaction, $j=1,2$, in the two-level of interaction approximation we have
\begin{align}
L_{2\text{approx},t}^1 =& \sum_{k_3} l_1^{k_3} \mathbb{P}(\tilde{\ell}_1=l_1^{k_3})- \sum_{k_1,k_2,k_3,k_4} l_1^{k_3}  u_0(t;k_1,k_2,k_3,k_4)\nonumber\\
&   \quad \quad \cdot  \mathbb{P}(\tilde{\beta}_1^C=\beta_1^{C,k_1},\tilde{\beta}_2^C=\beta_2^{C,k_2},\tilde{\ell}_1=l_1^{k_3},\tilde{\ell}_2=l_2^{k_4})\nonumber
\end{align}

\begin{align}
	L_{2\text{approx},t}^2 =& \sum_{k_4} l_2^{k_4} \mathbb{P}(\tilde{\ell}_2=l_2^{k_4})- \sum_{k_1,k_2,k_3,k_4} l_2^{k_4}  u_0(t;k_1,k_2,k_3,k_4)\nonumber\\
	&  \quad \quad \cdot  \mathbb{P}(\tilde{\beta}_1^C=\beta_1^{C,k_1},\tilde{\beta}_2^C=\beta_2^{C,k_2},\tilde{\ell}_1=l_1^{k_3},\tilde{\ell}_2=l_2^{k_4}).\nonumber
\end{align}

As with the previous example, we truncate at the level $K=20$, and choose the time endpoint to be $T=1$. We do the numerical iteration with time step being 0.01. We run 50,000 Monte Carlo trials and plot the overall limiting loss $D_{2\text{approx},t}$ in the two level of interaction approximation. In the left plot of Figure \ref{F:LimitingLossCorePeriphery}, we see that the two approximations perform  similarly in estimating the overall loss rate. This can be also verified via the plot of the mean of overall loss rate over time for each one of the two approximations in the right plot of Figure \ref{F:LimitingLossCorePeriphery}.
\begin{figure}[h!]
	\centering
	\includegraphics[scale=0.6]{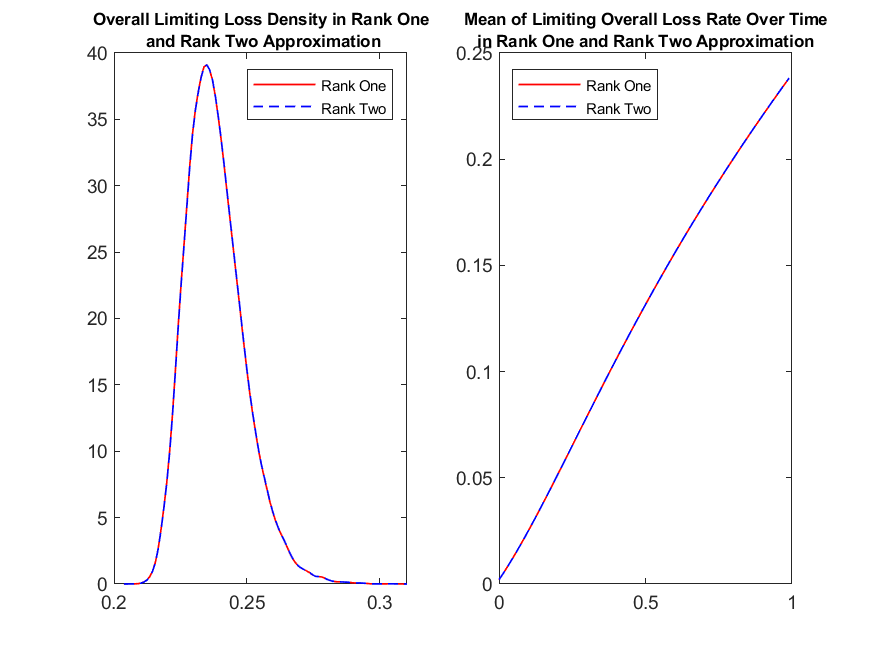}
	\small \caption{Left: Overall limiting loss from rank one approximation $D_{1\text{approx},T}$ and rank two approximation $D_{2\text{approx},T}$ at $T=1$; Right: Empirical mean of overall limiting loss for rank one approximation $D_{1\text{approx},T}$ and rank two approximation $D_{2\text{approx},T}$  up to time $T=1$.}\label{F:LimitingLossCorePeriphery}
\end{figure}

We can also investigate the mean impact on a name in the two-level of interaction approximation case. By Table \ref{T:jointrank2} we will have $10$ different types of mean impacts in the two-level of interaction approximation case. These are demonstrated in Figure \ref{F:MeanImpactCorePeriphery}.
\begin{figure}[h!]
	\centering
	\includegraphics[scale=0.6]{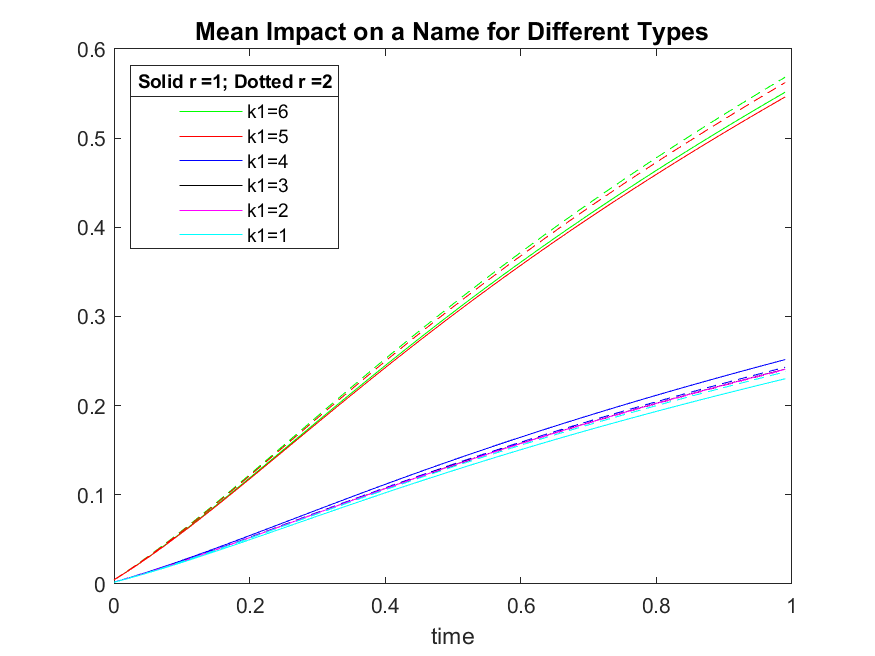}
	\small \caption{Mean impact on names of different types from system wide default by time t for the Core Periphery case approximated by rank one (solid line) and rank two (dotted line) with colors to distinguish the choice of $k_1$ in $\beta_1^{C,k_1}$ }\label{F:MeanImpactCorePeriphery}
\end{figure}

It is instructive to compare the low rank approximation based on just the first level of interaction with the low rank approximation based on the first two levels of interaction. The dotted lines are very well approximated by the solid line  in Figure \ref{F:MeanImpactCorePeriphery}. In fact, we computed numerically the percent error  of the mean impact on a name from the two different approximations, that is,
\[
\text{PE}_{t}(k_1,k_2)=|Q_{2\text{approx},t}(k_1,k_2)-Q_{1\text{approx},t}(k_1)|/Q_{2\text{approx},t}(k_1,k_2),
\]
and in all cases the percent error made by using the one-level of interaction approximation versus the two-level of interaction approximation was not greater than $1.7\%$ for all times $t\in[0,1]$. For comparison purposes we also mention that the computation of $D_{t}$ and $Q_{t}$ based on the two-level of interaction approximation took about two times larger than the their computation based on the one-level of interaction approximation, indicating savings in computational time while maintaining accuracy. Lastly, notice that the mean default impact on names of type $k_{1}=1,\cdots,6$  from system wide defaults is ordered according to the order of the corresponding contagion coefficients $\beta_{1}^{C,k_{1}}$ via Table \ref{T:betaC1}.

\subsection{Core-periphery example two: nonhomogeneous mean-reverting coefficients}\label{Ex:Core-Periphery2}
Now we investigate the core-periphery case with nonhomogeneous mean-reverting coefficients. We assume that the mean-reverting coefficient $\bar{\lambda}$ takes different values for names in the core and in the periphery component of the network: $\bar{\lambda}^1=\bar{\lambda}_{\text{core}}=0.02$ and $\bar{\lambda}^2=\bar{\lambda}_{\text{periphery}}=0.2$ and the rest of the coefficients as well as the network structure are the same from the one level approximation example of the previous Subsection \ref{Ex:Core-Periphery1}.

Notice that the choices $\bar{\lambda}_{\text{core}}=0.02$ and $\bar{\lambda}_{\text{periphery}}=0.2$ represent the anticipation that it is harder for a core institution to default than it is for a periphery institution. In the intensity model that we study, smaller mean-reverting parameter $\bar{\lambda}$ means smaller intensity to default process. In this example, we only investigate the rank one approximation. After all, as we showed in Subsection \ref{Ex:Core-Periphery1}, this approximation is sufficient to accurately capture the dynamical quantities we are interested in.

Let us denote by $u_{k}(t;k_1,k_2,k_3)$ to be the $k$-th moment by time t with $k_1 \in\{1,2,\dots,6\}$, $k_2 \in \{1,2,3\}$ and $k_3 \in \{1,2\}$ being the choice index for $\tilde{\beta}_1^C$, $\tilde{\ell}_1$ and $\bar{\lambda}$ respectively. For example, $k_1=1, k_2=2, k_3=1$ corresponds to the choice $\tilde{\beta}_1^C=\beta_{1}^{C,1}$, $\tilde{\ell}_1=l_{1}^2$ and $\tilde{\bar{\lambda}}=\bar{\lambda}^1=\bar{\lambda}_{\text{core}}=0.02$. The empirical joint distribution of $\tilde{\beta}_1^C$, $\tilde{\ell}_1$ and $\tilde{\bar{\lambda}}$ is summarized as follows.
\begin{table}[h!]
	\begin{center}
		\begin{tabular}{c|c|c|c|c}
			$k_1$ & $k_2$ & $k_3$ & probability & \\
			\hline
			6 & 3 & 1 & 0.001 \\
			5 & 2 & 1 & 0.001 \\
			4 & 1 & 2 & 0.227 \\
			3 & 1 & 2 & 0.238 \\
			2 & 1 & 2 & 0.228 \\
			1 & 1 & 2 & 0.305 \\
		\end{tabular}
		\caption{Joint distribution for $\tilde{\beta}_1^C$, $\tilde{\ell}_1$ and $\tilde{\bar{\lambda}}$.}
		\label{T:jointrank3}
	\end{center}
\end{table}

Because of the special structure of our system we end up with  6 different equations as indicated by Table \ref{T:jointrank3}:
\begin{align}
	&d u_{k}(t;k_1,k_2,k_3) =\left\{u_{k}(t;k_1,k_2,k_3)(-\bar{\alpha} k + \beta^S \kappa(\theta-X_t) k + 0.5 (\beta^S)^2 \epsilon^{2} X_t k(k-1)) \right. \nonumber\\
	&\quad \left.- u_{k+1}(t;k_1,k_2)\right\} dt+ u_{k-1}(t;k_1,k_2,k_3) \left\{ (0.5 \sigma^2 k(k-1) + \bar{\alpha} \bar{\lambda}_{k_3} k) + G_{k}(t;k_1) \right\}dt\nonumber\\
	&\quad+ \beta^S \epsilon\sqrt{ X_t} k u_{k}(t;k_1,k_2,k_3) dV_t,\nonumber
\end{align}
together with $u_{k}(0;k_1,k_2,k_3)= \int_0^\infty \lambda^k (\pi\times\Lambda_0)(\hat{p}) d\lambda$
and where we define
$$G_k(t;k_1)=\left( \sum_{i_1,i_2,i_3} l_{1}^{i_2}  u_1(t;i_1,i_2,i_3) \mathbb{P}(\tilde{\beta}_1^C=\beta_1^{C,i_1},\tilde{\ell}_1=l_1^{i_2},\tilde{\bar{\lambda}}=\bar{\lambda}^{i_3}) \right) k \beta_1^{C,k_1}.$$

In particular, $u_k(t;k_1,k_2,k_3)$ depends only on $k_1,k_3$ via $G_k(t;k_1)$ and $\bar{\lambda}_{k_3}$. The overall loss rate in the one-level of interaction approximation is
$$
D_{1\text{approx},t}^N \approx D_{1\text{approx},t} = 1-\sum_{k_1,k_2,k_3} u_0(t;k_1,k_2,k_3) \mathbb{P}(\tilde{\beta}_1^C=\beta_1^{C,k_1},\tilde{\ell}_1=l_1^{k_2},\tilde{\bar{\lambda}}=\bar{\lambda}^{k_3}).
$$

The loss rate for type $(k_1,k_2,k_3)$ where $k_1=1,2,\dots,6$, $k_2=1,2,3$ and $k_3=1,2$ in the one-level of interaction approximation are actually falling into 6 distinct categories indexed by $k_1$, the choice of $\tilde{\beta}_1^C$.
\[
D_{1\text{approx},t}^{N}(k_1,k_2,k_3) \approx D_{1\text{approx},t}(k_1,k_2,k_3) = 1- u_0(t;k_1,k_2,k_3).
\]

The mean impact on name $n$, from system wide defaults up to time t, associated to type $(k_1,k_2,k_3)$ as described in Table \ref{T:jointrank3},  turns out to be characterized by the first index $k_1$
\[
Q_{1\text{approx},t}^{N,n}(k_1,k_2,k_3) \approx Q_{1\text{approx},t}(k_1) = \beta_1^{C,k_1} L_{1\text{approx},t},
\]
for any $k_2=1,2,3,4$ and $k_3=1,2$ with
\begin{align}
	L_{1\text{approx},t} = & \sum_{k_2} l_1^{k_2} \mathbb{P}(\tilde{\ell}_1=l_1^{k_2})\nonumber\\
	& - \sum_{k_1,k_2,k_3} l_1^{k_2}  u_0(t;k_1,k_2,k_3) \mathbb{P}(\tilde{\beta}_1^C=\beta_1^{C,k_1},\tilde{\ell}_1=l_1^{k_2},\tilde{\bar{\lambda}}=\bar{\lambda}^{k_3}).\nonumber
\end{align}

As with the previous examples, we truncate at the level $K=20$, and choose the time endpoint to be $T=1$. We do the numerical iteration with time step being 0.01. We run 50,000 Monte Carlo trials and plot overall limiting loss rate $D_{1\text{approx},t}$ and the limiting loss rate for different types $D_{1\text{approx},t}^{k_1}$, $k_1=1,2,\dots,6$ in Figure \ref{F:LimitingLossCorePeriphery1rank_2lambdabar}. We also plot the mean of the loss rate over time for the whole pool and for individual types  in Figure \ref{F:LimitingLossOverTimeCorePeriphery1rank_2lambdabar}. We observe due to the smaller mean-reverting value, the names from the core component of the network are less likely to default than those in the periphery part of the network. This essentially confirms and quantifies what we expect to happen in this case. At this point it is indicative to compare Figure \ref{F:LimitingLossCorePeriphery1rank_2lambdabar} with Figure  \ref{F:LimitingLossCorePeriphery1rank}, as well as Figure \ref{F:LimitingLossOverTimeCorePeriphery1rank_2lambdabar} with Figure \ref{F:LimitingLossOverTimeCorePeriphery1rank}.
\begin{figure}[h!]
	\centering
	\includegraphics[scale=0.6]{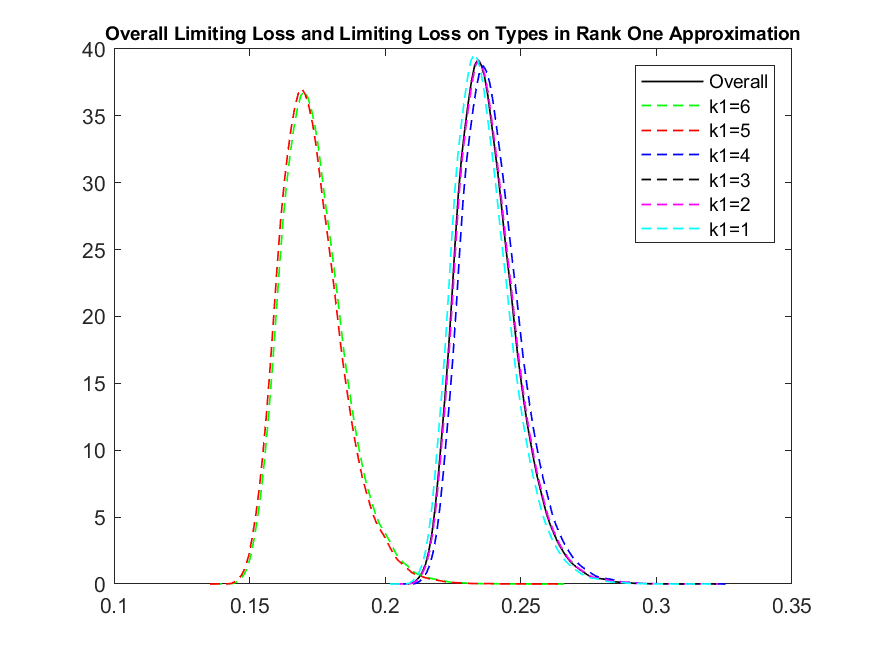}
	\small \caption{Density for overall limiting loss $D_{1\text{approx},T}$ and limiting loss for types $D_{1\text{approx},T}(k_1)$ at $T=1$.}\label{F:LimitingLossCorePeriphery1rank_2lambdabar}
\end{figure}

\begin{figure}[h!]
	\centering
	\includegraphics[scale=0.6]{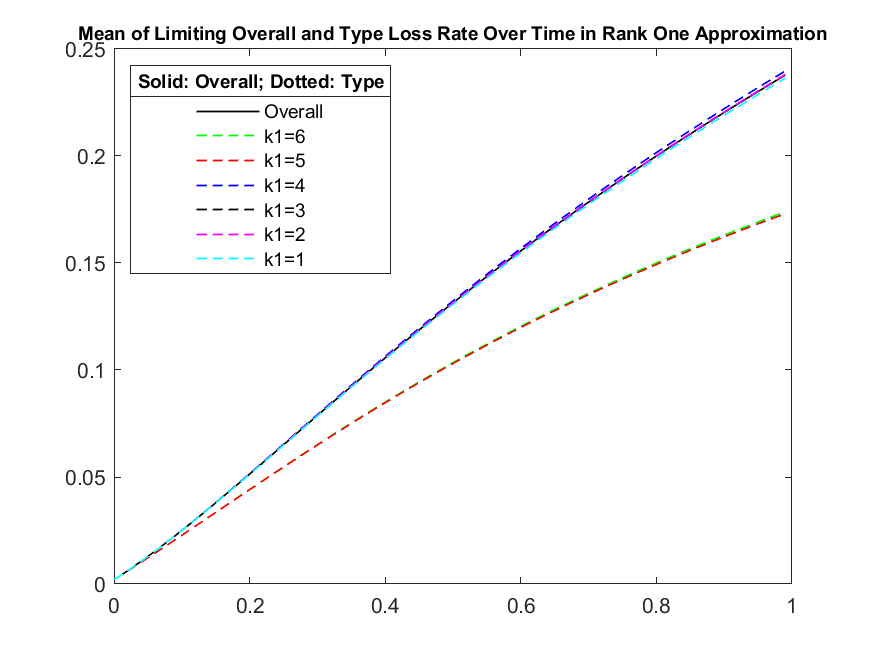}
	\small \caption{Empirical mean of overall limiting loss $D_{1\text{approx},T}$ and empirical mean of limiting loss for types $D_{1\text{approx},T}(k_1)$  up to time $T=1$}\label{F:LimitingLossOverTimeCorePeriphery1rank_2lambdabar}
\end{figure}

In Figure \ref{F:MeanImpactCorePeri1rankApprox_2lambdabar}, we plot the mean impact on a name from system wide defaults up to time $t$. As we discussed before, there are totally 6 different categories indexed by $k_1$ the choice of $\tilde{\beta}_1^C$.
\begin{figure}[h!]
	\centering
	\includegraphics[scale=0.6]{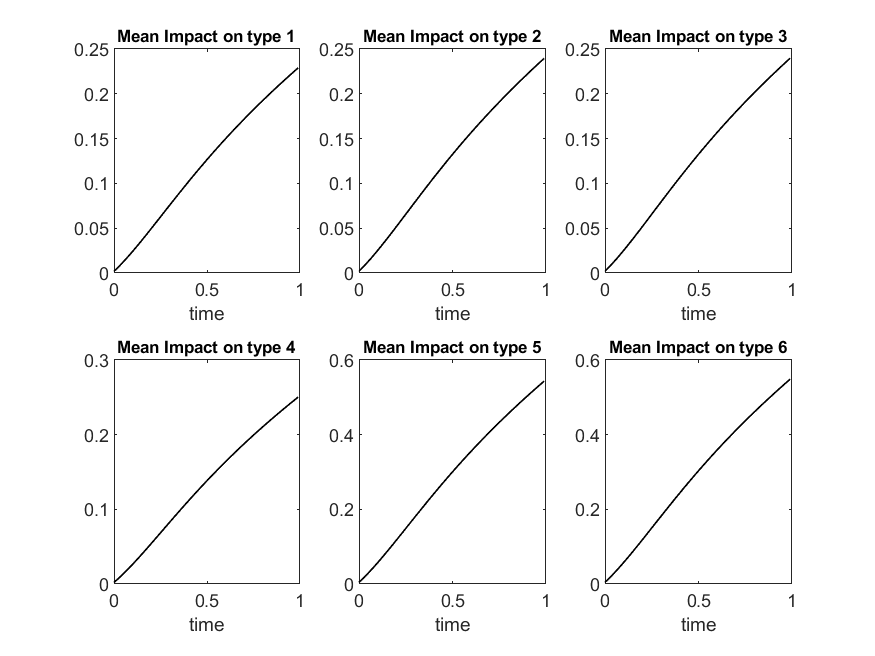}
	\small \caption{Mean impact on names of different types from system wide default by time t for the Core-Periphery case approximated by the first rank.}\label{F:MeanImpactCorePeri1rankApprox_2lambdabar}
\end{figure}

\section{Tightness and Characterization of the limit}\label{S:LimitCharacterization}
Let us now discuss relative compactness of the family $\{\mu^{N}\}_{N\in\mathbb{N}}$ and characterize its limit as $N\rightarrow\infty$.

\begin{lemma}\label{L:MuMeasureTight} The family $\{\mu^N\}_{N\in \mathbb{N}}$ is relatively compact as a $D_{E}[0,\infty)-$valued random variable.\end{lemma}
\begin{proof}
Due to Lemma \ref{lambda} proven in Appendix \ref{A:Appendix}, the proof of the lemma is as in Section 6 of \cite{GSS}. Hence, the details are omitted.
\end{proof}

Next, we want to use the martingale problem to identify the limit of $\mu^N$'s as $N$ grows. Let $\mathcal{S}$ be the collection of elements $\Phi$ in $B(\mathbb{R}\times\mathscr{P}(\hat{\mathcal{P}}))$ of the form
	$$\Phi(x,\mu)=\varphi_1(x)\varphi_2(\langle f_1,\mu \rangle_E,\ \langle f_2,\mu \rangle_E, \ldots, \langle f_M,\mu \rangle_E) $$
	for some $M \in \mathbb{N}$, some $\varphi_1 \in C^\infty(\mathbb{R})$, $\varphi_2 \in C^\infty(\mathbb{R}^M)$ and some $\{f_m\}_{m=1}^M$ in $C^\infty(\hat{\mathcal{P}})$. Then $\mathcal{S}$ separates the probability measure space $\mathcal{P}(\hat{\mathcal{P}})$. Then it is enough to consider the martingale convergence problem on $\mathcal{S}$.

Let's fix $f\in C^\infty(\hat{\mathcal{P}})$ and understand what happens to $\la f,\mu^N\ra_E$ when one of the firms defaults.  Suppose that the $n$-th firm defaults at time $t$ and that none of the other names defaults at time $t$ (defaults occur simultaneously
with probability zero).  We have that
\begin{align*} \la f,\mu^N_t\ra_E &= \frac{1}{N} \sum_{\substack{1\leqslant n^\prime \leqslant N \nonumber\\ n^\prime \neq n}} f\big( p^{N,n^\prime},\ \lambda_{t-}^{N,n^\prime} +  \frac{1}{N} \sum_{j=1}^r \xi_j^2 \ u_{n^\prime, j} \ l_{n,j} \big)  \ M_{t}^{N,n^\prime},\nonumber
	\end{align*}
	where we used the fact that the jump size in $\lambda^{N,n^\prime}$ at time t when there is a default in the $n-$th firm is $\frac{1}{N}\sum_{j=1}^r \xi_j^2 \ u_{n^\prime, j}  \ l_{n,j}$. In addition, noting that $M_t^{N,n}=0$ (since $n$-th firm defaults at time t means $\int_{0}^t \lambda_s^{N,n} ds =  \mathfrak{e}_n$), gives
	\begin{align}
	\langle f,\mu^{N}_{t-} \rangle_E
	&=\frac{1}{N} \sum_{\substack{1\leqslant n^\prime \leqslant N \\ n^\prime \neq n}} f(p^{N,n^\prime},\ \lambda_{t-}^{N,n^\prime})  M_{t}^{N,n^\prime} + \frac{1}{N} f(p^{N,n},\ \lambda_{t-}^{N,n}).\nonumber
	\end{align}
	
Therefore, we have that	
	$$\langle f,\mu^{N}_{t} \rangle_E - \langle f,\mu^{N}_{t-} \rangle_E= \mathcal{J}_{N,n}^f(t)$$
	where
	\begin{eqnarray}
	\mathcal{J}_{N,n}^f(t)&=&\frac{1}{N} \sum_{1\leqslant n^\prime \leqslant N } \Big( f\big( p^{N,n^\prime},\ \lambda_{t-}^{N,n^\prime} + \frac{1}{N} \sum_{j=1}^r \xi_j^2 \ u_{n^\prime, j}  \ell_{n,j} \big)- f(p^{N,n^\prime},\ \lambda_{t-}^{N,n^\prime}) \Big)  \ M_{t}^{N,n^\prime}\nonumber \\
	&-& \frac{1}{N} f(p^{N,n},\ \lambda_{t-}^{N,n}). \nonumber
	\end{eqnarray}

For $f\in C^2(\mathbb{R})$ define the operator
$$\mathcal{G}(f)(x)=b_0(x)\frac{\partial f}{\partial x} (x) +\frac{1}{2} \sigma_0^2(x) \frac{\partial^2 f}{\partial x^2}(x).$$

In addition, define the operators
	\begin{align*}
		&(\mathcal{A}\Phi)(x,\mu) = \mathcal{G}(\varphi_1)(x) \   \varphi_2\big(\langle f_1,\mu \rangle_E,\ \langle f_2,\mu \rangle_E, \ldots, \langle f_M,\mu \rangle_E \big)\\
		&\quad + \sum_{m=1}^M \varphi_1(x) \frac{\partial \varphi_2}{\partial x_m}\big(\langle f_1,\mu \rangle_E,\ \langle f_2,\mu \rangle_E, \ldots, \langle f_M,\mu \rangle_E \big)\\
		& \qquad \times \big[ \big\langle \mathcal{L}_1 f_m, \mu  \big\rangle_E + \big\langle \mathcal{L}_2^x f_m, \mu  \big\rangle_E + \big\langle \iota \nu, \mu  \big\rangle_E \cdot \big\langle \mathcal{L}_4^x f_m, \mu  \big\rangle_E   \big]\\
		& \quad + \frac{1}{2}\  \sum_{m=1}^M \frac{\partial \varphi_1}{\partial x}(x)\ \frac{\partial \varphi_2}{\partial x_m}\big(\langle f_1,\mu \rangle_E,\ \langle f_2,\mu \rangle_E, \ldots, \langle f_M,\mu \rangle_E \big)  \langle \sigma_0(x) \mathcal{L}_3^{x} f_m,\mu \rangle_E\\
		& \quad +\frac{1}{2} \sum_{p,q=1}^M \varphi_1(x) \frac{\partial^2 \varphi}{\partial x_p \partial x_q} \big(\langle f_1,\mu \rangle_E,\ \langle f_2,\mu \rangle_E, \ldots, \langle f_M,\mu \rangle_E \big)\\
		&\qquad \times \big( \langle \mathcal{L}_3 f_p,\mu \rangle_E \langle \mathcal{L}_3 f_q,\mu \rangle_E  \big).\\
	\end{align*}
and	
	\begin{align*}
		(\mathcal{B}\Phi)(x,\mu) &= \sigma_0(x)\frac{\partial \varphi_1}{\partial x} (x) \varphi_2\big(\langle f_1,\mu \rangle_E,\ \langle f_2,\mu \rangle_E, \ldots, \langle f_M,\mu \rangle_E \big)\\
		& + \varphi_1(x) \sum_{m=1}^M \frac{\partial \varphi_2}{\partial x_m} \big(\langle f_1,\mu \rangle_E,\ \langle f_2,\mu \rangle_E, \ldots, \langle f_M,\mu \rangle_E \big) \langle \mathcal{L}_3^{x} f_m,\mu \rangle_E.\\
	\end{align*}

Then, Theorem \ref{mm} characterizes the possible limit points.
\begin{thm}\label{mm}
We have that
\begin{align}
\lim_{N \to \infty}&\mathbb{E} \left[ \left\{ \Phi(X_{t_2},\mu_{t_2}^N)- \Phi(X_{t_1},\mu_{t_1}^N)- \int_{t_1}^{t_2} (\mathcal{A}\Phi)(X_s,\mu_s^N)ds- \right.\right.\nonumber\\
&\quad\left.\left. \int_{t_1}^{t_2} (\mathcal{B}\Phi)(X_s,\mu_s^N) dV_s \right\} \prod_{j=1}^J \psi_j(x_{r_j},\mu_{r_j}^N) \right]=0\nonumber
\end{align}
	for any $\Phi \in \mathcal{S}$ and $0\leq r_1 \leq r_2 \leq \cdots \leq r_J=t_1 < t_2 <T$ and $\{ \psi_j \}_{j=1}^J \in B(\mathbb{R}\times E)$.	
\end{thm}

\begin{proof}
First, we notice that,
\[
\mathcal{M}_t^{N,n} \ = \ 1-M_t^{N,n} - \int_{0}^t \lambda_s^{N,n} M_s^{N,n} ds
\]
is a martingale. This means that we can write
	$$d(1-M_t^{N,n}) \ = \ d\mathcal{M}_t^{N,n} + \lambda_t^{N,n} M_t^{N,n} dt$$

By It\^{o}'s formula we obtain
	\begin{align}
	d \langle f,\mu^{N}_{t} \rangle_E &= \frac{1}{N} \sum_{n=1}^N \Big[ \mathcal{L}_1 f(\hat{p}_t^{N,n}) +\lambda^{\NN}_{t}f(\hat{p}_t^{N,n}) \Big]  \ M_{t}^{N,n} \ dt+\frac{1}{N} \sum_{n=1}^N \Big[ \mathcal{L}_2^{X_t} f(\hat{p}_t^{N,n}) \Big]  \ M_{t}^{N,n} \ dt\nonumber\\
	&\quad+\frac{1}{N} \sum_{n=1}^N \ \sigma_{n} \ (\lambda_t^{N,n})^{\rho} \ \frac{\partial f}{\partial \lambda} (\hat{p}_t^{N,n})  \ M_{t}^{N,n} \ d W_t^{n}+\frac{1}{N} \sum_{n=1}^N  \  \mathcal{L}_3^{X_t} f(\hat{p}_t^{N,n}) \ M_{t}^{N,n} \ dV_t\nonumber\\
	&\quad +\sum_{n=1}^N \mathcal{J}_{N,n}^f(t) \  d \ [1-M_t^{N,n}]\nonumber\\
&=  \left[\langle \mathcal{L}_1 f,\mu^{N}_{t} \rangle_E+  \big\langle \iota f, \mu^{N}_t  \big\rangle \right] dt + \langle \mathcal{L}_2^{X_t} f,\mu^{N}_{t} \rangle_E\ dt\nonumber\\
	&\quad +\frac{1}{N} \sum_{n=1}^N \ \sigma_{n} \ (\lambda_t^{N,n})^{\rho} \ \frac{\partial f}{\partial \lambda} (\hat{p}_t^{N,n})  \ M_{t}^{N,n} \ d W_t^{n}\nonumber\\
	&\quad + \langle \mathcal{L}_3^{X_t} f,\mu^{N}_{t} \rangle_E \ dV_t+\sum_{n=1}^N \mathcal{J}_{N,n}^f(t) \  d \ [1-M_t^{N,n}].\nonumber\\\label{itomut}
	\end{align}
	
	Again, by It\^{o}'s formula for $\Phi(X_t, \mu_t^N)$ we subsequently obtain that	
\begin{align*}
		\Phi(X_t, \mu_t^N) &= \Phi(X_0, \mu_0^N) + \int_{0}^{t} \mathcal{G}(\varphi_1)(X_s) \   \varphi_2\big(\langle f_1,\mu_s^{N} \rangle_E,\ \langle f_2,\mu_s^{N} \rangle_E, \ldots, \langle f_M,\mu_s^{N} \rangle_E \big)\ ds\\
		& \quad + \int_{0}^t \varphi_1(X_s) \sum_{m=1}^M \frac{\partial \varphi_2}{\partial x_m} \big(\langle f_1,\mu_s^{N} \rangle_E,\ \langle f_2,\mu_s^{N} \rangle_E, \ldots, \langle f_M,\mu_s^{N} \rangle_E \big)\\
		& \qquad  \qquad \times \Big\{  \langle \mathcal{L}_1 f_m,\mu^{N}_{s} \rangle_E\ + \big\langle \iota f_m, \mu^{N}_t  \big\rangle + \langle \mathcal{L}_2^{X_s} f_m,\mu^{N}_{s} \rangle_E\  \Big\}\ ds\\
		& \quad + \frac{1}{2}\ \int_{0}^t \sum_{m=1}^M \frac{\partial \varphi_1}{\partial x}(X_s)\ \frac{\partial \varphi_2}{\partial x_m}\big(\langle f_1,\mu_s^{N} \rangle_E,\ \langle f_2,\mu_s^{N} \rangle_E, \ldots, \langle f_M,\mu_s^{N} \rangle_E \big)\\
		& \qquad \qquad \times \sigma_0(X_s) \langle \mathcal{L}_3^{X_s} f_m,\mu^{N}_{s} \rangle_E\ ds\\
		& \quad + \int_{0}^t \varphi_1(X_s) \sum_{n=1}^N \lambda_s^{N,n} \ \Big\{ \varphi_2 \big( \langle f_1,\mu_s^{N} \rangle_E \ +\mathcal{J}_{N,n}^{f_1}(s),\langle f_2,\mu_s^{N} \rangle_E \ +\mathcal{J}_{N,n}^{f_2}(s),\nonumber\\
			& \qquad \qquad \ldots, \langle f_M,\mu_s^{N} \rangle_E \ +\mathcal{J}_{N,n}^{f_M}(s) \big) - \varphi_2 \big(\langle f_1,\mu_s^{N} \rangle_E,\ \langle f_2,\mu_s^{N} \rangle_E, \ldots, \langle f_M,\mu_s^{N} \rangle_E \big) \Big\}\ M_s^{N,n}\  ds\\
		& \quad + \int_{0}^t  \sigma_0(X_s)\frac{\partial \varphi_1}{\partial x} (X_s) \varphi_2\big(\langle f_1,\mu_s^{N} \rangle_E,\ \langle f_2,\mu_s^{N} \rangle_E, \ldots, \langle f_M,\mu_s^{N} \rangle_E \big) \ dV_s\\
		& \quad + \int_{0}^t \varphi_1(X_s) \sum_{m=1}^M \frac{\partial \varphi_2}{\partial x_m} \big(\langle f_1,\mu_s^{N} \rangle_E,\ \langle f_2,\mu_s^{N} \rangle_E, \ldots, \langle f_M,\mu_s^{N} \rangle_E \big) \langle \mathcal{L}_3^{X_s} f_m,\mu^{N}_{s} \rangle_E\ dV_s\\
		& \quad + \frac{1}{2N^2} \sum_{p,q=1}^M \int_{0}^t \varphi_1(X_s)\ \frac{\partial^2 \varphi}{\partial x_p \partial x_q} \big(\langle f_1,\mu_s^{N} \rangle_E,\ \langle f_2,\mu_s^{N} \rangle_E, \ldots, \langle f_M,\mu_s^{N} \rangle_E \big)\\
		&\qquad \qquad \times \Big( \sum_{n=1}^N \sigma_n^2 \ (\lambda_s^{N,n})^{2\rho} \ \frac{\partial f_p}{\partial \lambda} (\hat{p}_s^{N,n}) \  \frac{\partial f_q}{\partial \lambda} (\hat{p}_s^{N,n})  \ M_s^{N,n} \Big)\ ds
\end{align*}	
	\begin{align*}
		& \quad +\frac{1}{2} \sum_{p,q=1}^M  \int_{0}^t \varphi_1(X_s)\ \frac{\partial^2 \varphi}{\partial x_p \partial x_q} \big(\langle f_1,\mu_s^{N} \rangle_E,\ \langle f_2,\mu_s^{N} \rangle_E, \ldots, \langle f_M,\mu_s^{N} \rangle_E \big)\\
		&\qquad \qquad \times \big( \langle \mathcal{L}_3^{X_s} f_p,\mu^{N}_{s} \rangle_E \langle \mathcal{L}_3^{X_s} f_q,\mu^{N}_{s} \rangle_E  \big) \ ds\\
		&\quad + \frac{1}{N} \sum_{n=1}^N \sum_{m=1}^M \int_{0}^t \varphi_1(X_s) \frac{\partial \varphi_2}{\partial x_m} \big(\langle f_1,\mu_s^{N} \rangle_E,\ \langle f_2,\mu_s^{N} \rangle_E, \ldots, \langle f_M,\mu_s^{N} \rangle_E \big)\\
		& \qquad \qquad \times \big(  \ \sigma_{n} \ (\lambda_s^{N,n})^{\rho} \ \frac{\partial f_m}{\partial \lambda} (\hat{p}_s^{N,n})  \ M_{s}^{N,n}\big) \ d W_s^{n}\\
		&\quad + \quad \int_{0}^t \varphi_1(X_s) \sum_{n=1}^N  \ \Big\{ \varphi_2 \big( \langle f_1,\mu_s^{N} \rangle_E \ +\mathcal{J}_{N,n}^{f_1}(s),\langle f_2,\mu_s^{N} \rangle_E \ +\mathcal{J}_{N,n}^{f_2}(s),\\
		& \qquad \qquad \ldots, \langle f_M,\mu_s^{N} \rangle_E \ +\mathcal{J}_{N,n}^{f_M}(s) \big) - \varphi_2 \big(\langle f_1,\mu_s^{N} \rangle_E,\ \langle f_2,\mu_s^{N} \rangle_E, \ldots, \langle f_M,\mu_s^{N} \rangle_E \big) \Big\} \  d \mathcal{M}_s^{N,n}\\
		&\qquad= \sum_{i=1}^{11}J_i^N,
	\end{align*}
where, for $i=1,\cdots, 11$, $J_i^N$ represents the $i^{\textrm{th}}$ term in the right hand side of the last display.  Notice that,
	
$$J_6^N+J_7^N=\int_0^t (\mathcal{B}\Phi) (X_s,\mu_s^N) dV_s, $$
	and		
$$J_1^N+J_2^N+J_3^N+J_4^N+J_9^N=\Phi(X_0, \mu_0^N) + \int_0^t (\mathcal{A} \Phi) (X_s,\mu_s^N) - \tilde{A}_s^N ds$$

where the $\tilde{A}_t^N$ is defined as
\begin{align*}
	\tilde{A}_t^N = &\sum_{m=1}^M \varphi_1(X_t) \frac{\partial \varphi_2}{\partial x_m}\big(\langle f_1,\mu_t^{N} \rangle_E,\ \langle f_2,\mu_t^{N} \rangle_E, \ldots, \langle f_M,\mu_t^{N} \rangle_E \big)\\
	& \qquad \times \frac{1}{N} \sum_{n=1}^N \lambda_t^{N,n} \ \tilde{\mathcal{J}}_{N,n}^{f_m} (t) M_t^{N,n}.
\end{align*}

and $\tilde{\mathcal{J}}_{N,n}^f(t)$ is defined as
$$\tilde{\mathcal{J}}_{N,n}^f(t)=\frac{1}{N} \sum_{1\leqslant n^\prime \leqslant N } \Big( \sum_{j=1}^r \xi_j^2 \ u_{n^\prime, j} \ell_{n,j} \frac{\partial f}{\partial \lambda} (\hat{p}_t^{N,n^\prime}) \big)  \ M_{t}^{N,n^\prime} - f(\hat{p}_t^{N,n}). $$

Notice that we have
$$\sum_{j=1}^r \xi_j^2 \ u_{n^\prime, j} \ell_{n,j} = \beta_{N,n^\prime}^C \cdot l^n,$$
where $\beta_{N,n^\prime}^C=(\xi_1^2 u_{n^\prime,1}, \xi_2^2 u_{n^\prime,2},\ldots, \xi_r^2 u_{n^\prime,r})$ and $l^n=(l_{n,1},l_{n,2},\ldots,l_{n,r})$.

Recalling that
$$\mathcal{L}_4 f = \beta^C \frac{\partial f}{\partial \lambda} (\hat{p}).$$

where $\beta^C=(\xi_1^2 u_1, \xi_2^2 u_2, \ldots \xi_r^2 u_r)$, we get that

$$\tilde{\mathcal{J}}_{N,n}^f(t) = l^n \cdot \big\langle \mathcal{L}_4 f, \mu_t^{N}  \big\rangle_E \ - \ f(\hat{p}_t^{N,n}). $$

Therefore we obtain that
\begin{align*}
	\tilde{A}_t^N = &\sum_{m=1}^M \varphi_1(X_t) \frac{\partial \varphi_2}{\partial x_m}\big(\langle f_1,\mu_t^{N} \rangle_E,\ \langle f_2,\mu_t^{N} \rangle_E, \ldots, \langle f_M,\mu_t^{N} \rangle_E \big)\\
	& \qquad \times \frac{1}{N} \sum_{n=1}^N \lambda_t^{N,n} \Big[ l^n \cdot \big\langle \mathcal{L}_4 f_m, \mu_t^{N}  \big\rangle_E \ - \ f_m(\hat{p}_t^{N,n}) \Big] M_t^{N,n}\\
	&=\sum_{m=1}^M \varphi_1(X_t) \frac{\partial \varphi_2}{\partial x_m}\big(\langle f_1,\mu_t^{N} \rangle_E,\ \langle f_2,\mu_t^{N} \rangle_E, \ldots, \langle f_M,\mu_t^{N} \rangle_E \big)\\
	& \qquad \times \big[ \frac{1}{N} \sum_{n=1}^N \lambda_t^{N,n} l^n \cdot  \big\langle \mathcal{L}_4 f_m, \mu_t^{N}  \big\rangle_E M_t^{N,n} \ - \ \frac{1}{N} \sum_{n=1}^N \lambda_t^{N,n} f_m(\hat{p}_t^{N,n}) M_t^{N,n} \big]\\
	&=\sum_{m=1}^M \varphi_1(X_t) \frac{\partial \varphi_2}{\partial x_m}\big(\langle f_1,\mu_t^{N} \rangle_E,\ \langle f_2,\mu_t^{N} \rangle_E, \ldots, \langle f_M,\mu_t^{N} \rangle_E \big)\\
	& \qquad  \times \big[ \big\langle \iota \nu, \mu_t^{N}  \big\rangle_E \cdot \big\langle \mathcal{L}_4 f_m, \mu_t^{N}  \big\rangle_E \ - \ \big\langle \iota f_m, \mu_t^{N}  \big\rangle_E \big].
\end{align*}

Now we prove that  $\left|J_5^N-\int_0^t \tilde{A}_s^N ds \right| \to 0$ as $N \to \infty$. Denote the operator
$$\mathcal{L}_5 f(\hat{p})=\sigma \lambda^{\rho} \frac{\partial f}{\partial \lambda}$$

Denote the jump term $J_{5}^N$ in the expression $\Phi(X_t,\mu_t^{N})$ as
$\int_{0}^t A_s^N ds$.
Now we look at the limit of this term as $N \to \infty$.

Hence there exists a constant $K$ which depends on the uppper bound of the coefficients such that
$$\Big| \mathcal{J}_{N,n}^f(t) \ -\ \frac{1}{N} \tilde{\mathcal{J}}_{N,n}^f(t) \Big| \leqslant \frac{K^2}{N^2} \lVert \frac{\partial^2 f}{\partial \lambda^2}\rVert.$$

Hence, we get that
	$$\lim_{N \to \infty} \mathbb{E} \big[ \int_{0}^{t} |A_s^N-\tilde{A}_s^N| \ ds \big] = 0.$$

Let us next show that $J_8^N \to 0$. The term $J_{8}^N$ above can be written as,
\begin{align*}
	& \frac{1}{2N} \sum_{p,q=1}^M \int_{0}^t \varphi_1(X_s) \frac{\partial^2 \varphi}{\partial x_p \partial x_q} \big(\langle f_1,\mu_s^{N} \rangle_E,\ \langle f_2,\mu_s^{N} \rangle_E, \ldots, \langle f_M,\mu_s^{N} \rangle_E \big)\\
	& \qquad \times \Big\{ \frac{1}{N} \sum_{n=1}^N (\mathcal{L}_5 f_p) \ (\hat{p}_s^{N,n}) (\mathcal{L}_5 f_q) \ (\hat{p}_s^{N,n})  \ M_s^{N,n} \Big\} ds\\
	&= \frac{1}{2N} \sum_{p,q=1}^M \int_{0}^t \varphi_1(X_s) \frac{\partial^2 \varphi}{\partial x_p \partial x_q} \big(\langle f_1,\mu_s^{N} \rangle_E,\ \langle f_2,\mu_s^{N} \rangle_E, \ldots, \langle f_M,\mu_s^{N} \rangle_E \big)\\
	& \qquad \times \big\langle (\mathcal{L}_5 f_p \ \mathcal{L}_5 f_q),\mu_s^{N} \big\rangle_E \  ds.
\end{align*}

This term goes to zero as $N$ goes to infinity. Indeed, for the given $M$ and $\{f_m\}_{m=1}^M$ and $t$ there exists a constant $C$ depending on $\max_{\{p,q=1,\ldots,M\}}\lVert \frac{\partial^2 \varphi}{\partial x_p \partial x_q} \rVert$ and $\max_{\{m=1\ldots M\}}\lVert f_m \rVert$ and the upper bound of the coefficients such that,
\begin{align*}
	&\Big|\frac{1}{2N} \sum_{p,q=1}^M \int_{0}^t \varphi_1(X_s) \frac{\partial^2 \varphi}{\partial x_p \partial x_q} \big(\langle f_1,\mu_s^{N} \rangle_E,\ \langle f_2,\mu_s^{N} \rangle_E, \ldots, \langle f_M,\mu_s^{N} \rangle_E \big)\\
	& \qquad \times \big\langle (\mathcal{L}_5 f_p \ \mathcal{L}_5 f_q),\mu_s^{N} \big\rangle_E \  ds \Big| \leqslant \frac{C}{N} \longrightarrow 0.
\end{align*}

Lastly, we treat the terms $J_{10}^N$ and $J_{11}^N$. Notice that the second to the last term $J_{10}^N$ is a Brownian martingale and the term $J_{11}^N$ is also a martingale. Denote their sum as a martingale $\mathcal{M}_t^N$. Calculations similar to the ones done above yield that
\[
\lim_{N\rightarrow\infty}\sup_{t\in[0,T]}\mathbb{E}|\mathcal{M}_t^N|^{2}=0,
\]
	and the proof of the theorem is complete.

\end{proof}

\section{Identification of the unique limit point}\label{S:Uniqueness}

The uniqueness of the solution to the limiting martingale problem implied by Theorem \ref{mm} is analogous to the  duality argument of Lemma 7.1 of \cite{GSS}  and the proof will not be repeated here.

Let us now identify this unique solution in the following two lemmas. Lemma \ref{uniqueQ} will give us the existence of a unique solution to a certain stochastic differential equation  which will then be used in identifying the unique limiting  solution in Lemma \ref{uni}.

   \begin{lem}
	Let $W^{*}$ be a reference Brownian motion and $T<\infty$. For each $\hat{p} \in \hat{\mathcal{P}}$, with $\hat{p}=(p,\lambda_{0})$, each $t \leq T$ there is a unique pair of $(Q_i(t),\lambda_t^{*}(\hat{p}),i=1, \ldots,r)$
	$$Q_i(t)=\int_{\hat{p} \in \hat{\mathcal{P}}} l_i \mathbb{E}_{\mathcal{V}_t} \left\{\lambda_t^*(\hat{p})\exp\left[-\int_{0}^t \lambda_s^*(\hat{p})ds\right]\right\}\pi(dp)\Lambda_0(d\lambda_0).$$
	$$\lambda_t^*(\hat{p})=\lambda_0+\int_{0}^t b(\lambda_s^*(\hat{p}),a)ds + \sigma\cdot(\lambda_s^*(\hat{p}))^{\rho} dW^*_s + \int_{0}^t \sum_{i=1}^r \beta_i^C Q_i(s) ds + \beta^S \int_{0}^t \lambda_s^*(\hat{p}) dX_s.$$ \label{uniqueQ}
\end{lem}

Lemma \ref{uniqueQ} is proven in the Appendix.
\begin{lem}
	Let $(Q_i(t),\lambda_t^{*}(\hat{p}),i=1, \ldots,r)$, with $\hat{p}=(p,\lambda_{0})$, be the unique pair from Lemma \ref{uniqueQ} with $\mathcal{V}_t$ the filtration generated by the limiting $X$. For any $A \in \mathfrak{B}(\mathcal{P})$ and $B \in \mathfrak{B}({\mathbb{R}_{+}})$, $\bar{\mu}$ is given by
	$$\bar{\mu_t}(A \times B) = \int_{\hat{p}\in \hat{\mathcal{P}}} \chi_A(p) \mathbb{E}_{\mathcal{V}_t} \left[ \chi _B(\lambda_t^*(\hat{p}))\exp\left[-\int_{0}^t\lambda_s^*(\hat{p})ds \right] \right] \pi(dp) \Lambda_0(d\lambda_0).$$    \label{uni}
\end{lem}

\begin{proof}
	For any $f \in C^{\infty}(\mathcal{\hat{P}})$, define a $\mathcal{V}_t-$adapted random element $\bar{\mu}$ of $D_{E}$ by the action
	$$\langle f,\bar{\mu}_t \rangle_E=\int_{\hat{p}\in \hat{\mathcal{P}}} \mathbb{E}_{\mathcal{V}_t} \left[f(p,\lambda_t^*(\hat{p})) \exp\left[ -\int_0^t \lambda_s^*(\hat{p})ds\right] \right] \pi(dp) \Lambda_0(d\lambda_0).$$

By It\^{o}'s formula, we obtain, using Lemmas B.1 and B.2 in \cite{GSSS}, that
	\begin{align}
	&d  \langle f,\bar{\mu}_t\rangle_E =\nonumber\\
&= \left\{\int_{\hat{p}\in \hat{\mathcal{P}}} \mathbb{E}_{\mathcal{V}_t} \left[ \left[ (\mathcal{L}_1 f)(p,\lambda_t^*(\hat{p}))+ (\mathcal{L}_2^{X_t} f)(p,\lambda_t^*(\hat{p}))\right] \exp\left[ -\int_0^t \lambda_s^*(\hat{p})ds \right]\right]\right.\nonumber\\
&\hspace{8cm}\left. \pi(dp) \Lambda_0(d\lambda_0) \right\} dt\nonumber\\
	&\quad +\left\{\int_{\hat{p}\in \hat{\mathcal{P}}} \mathbb{E}_{\mathcal{V}_t} \left[ ( \mathcal{L}_3^{X_t} f)(p,\lambda_t^*(\hat{p})) \exp\left[ -\int_0^t \lambda_s^*(\hat{p})ds \right]\right]\pi(dp)\Lambda_0(d\lambda_0)\right\}dV_t \nonumber\\
	&\quad + \left\{ \int_{\hat{p}\in \hat{\mathcal{P}}} \mathbb{E}_{\mathcal{V}_t} \left[  \sum_{i=1}^r Q_i(t) (\mathcal{L}_{4} f)_i(p,\lambda_t^*(\hat{p})) \exp \left[ -\int_0^t \lambda_s^*(\hat{p})ds \right]\right] \pi(dp)\Lambda_0(d\lambda_0) \right\} dt \nonumber\\
	&\quad=\left\{ \langle\mathcal{L}_1 f, \bar{\mu_t}\rangle_E + \langle\mathcal{L}_2^{X_t} f, \bar{\mu_t}\rangle_E + \sum_{i=1}^r Q_i(t) \langle(\mathcal{L}_{4} f)_i, \bar{\mu_t}\rangle_E\right\}dt+ \left\{\langle\mathcal{L}_3^{X_t} f, \bar{\mu_t}\rangle_E\right\} dV_t. \nonumber
	\end{align}
	where $\iota(\lambda,p)=\lambda$. Define now
\[
G_i(t)= \int_{\hat{p}\in \hat{\mathcal{P}}}  \mathbb{E}_{\mathcal{V}_t} l_i \left[  \exp\left[-\int_{0}^t\lambda_s^*(\hat{p})ds \right] \right] \pi(dp) \Lambda_0(d\lambda_0).
\]
	
Then, we have that
\[
G'_i(t)= - \int_{\hat{p}\in \hat{\mathcal{P}}}  \mathbb{E}_{\mathcal{V}_t} l_i \left[ \lambda_t^*(\hat{p}) \exp\left[-\int_{0}^t\lambda_s^*(\hat{p})ds \right] \right] \pi(dp) \Lambda_0(d\lambda_0)=-\langle l_i\iota,\bar{\mu_t} \rangle_E.
\]
	
On the other hand by Lemma \ref{uniqueQ}, we have $G'_i(t)=-Q_i(t)$, concluding the proof of the lemma due to uniqueness.	
\end{proof}

\section{Conclusions and further research work}\label{S:Conclusions}

We consider a general point process model of correlated default timing in a pool of components (e.g. firms or names) interacting via a weighted directed graph which determines the impact of default among the different components. The model is empirically motivated and incorporates contagion effects, common systematic risk factors as well as  idiosyncratic effects.

We prove a law of large numbers for the empirical survival distribution. This is then used to study the behavior of dynamic quantities of interest, such as mean loss rate in the pool or mean impact on given names from system wide defaults. The presence of the network structure  enlarges the set of interesting questions that we can ask and at the same time allows via singular value decomposition arguments to reduce the computational burden via low rank approximations.

One of the interesting questions that we did not address here is that of the effect of choices such as bistability in the idiosyncratic component of the intensity-to-default process. Questions motivated by such choices, as well as others including the study of most likely paths to default, are more suitable for large deviations analysis in the spirit of \cite{SpiliopoulosSowers2013}, which will be done in a follow up work.  In the present work we focus on establishing mathematical well-posedness of such models and on numerically exploring the effects of the network structure and low rank approximations on the typical behavior of quantities of interest.

Another potential interesting question is what happens when one wants to allow the rank of the low-rank approximation to $\Delta$ to increase with $N$, say $r=r(N)\rightarrow\infty$. In such a case, we expect that the term $Q^{N,n}_{t}=\beta^{C}_{n}\cdot L^{N}_{t}$ in equation (\ref{Eq:MainModel}) should be scaled by $\frac{1}{r(N)}$ and thus be replaced  by $\frac{1}{r(N)} \beta^{C}_{n}\cdot L^{N}_{t}$. We do not study this question in this paper, but we believe that the techniques developed in this paper will be useful in order to address this question.

\appendix

\section{Appendix}\label{A:Appendix}
In this appendix we prove lemmas used throughout the paper. We remark here that most of the technical difficulties arising from dropping the affine structure in the idiosyncratic part of the intensity process are encountered in the proofs of the results in this Appendix.

Let $\xi$ be a vector of processes having $r$ components, predictable, bounded, right continuous, monotone with $\xi_{0}=0$. Define the process
$$Z_t=\lambda_0+\beta^C \cdot \int_0^t e^{\Gamma_s} d\xi_s.$$

\begin{lem}\label{Zt}
	Let  $p\geq1$ be such that Assumptions \ref{s0} and \ref{Xt} hold. Then we have that
\begin{align}
&\mathbb{E}[Z_t^{2p}]^{1/(2p)} \leq \lambda_0 + ||\beta^C||_1\mathbb{E}[e^{2p\Gamma_t}]^{1/(2p)}\nonumber\\
	&+t^{1-1/{2p}} \left[\left(\int_0^t\mathbb{E} \left[e^{4p\Gamma_s}\right]ds\right)\right]^{1/{4p}} \left\{\left(\int_0^t ||\beta^C||_1^{1/{4p}} (\beta^S)^{1/{4p}} \mathbb{E}\left(\left[b_0(X_s)\right]^{4p}\right)ds\right)\right\}^{1/{4p}}.\nonumber
\end{align}
In particular, we have that there is a finite constant $0<K<\infty$ such that $\mathbb{E}[Z_t^{2p}]\leq K$.
\end{lem}

\begin{proof}[Proof of Lemma \ref{Zt}]
Notice that $Z_t$ can be written as
	\begin{eqnarray}
	Z_t&=&\lambda_0+\beta^C \cdot \{e^{\Gamma_t} \xi_t+\int_0^t e^{\Gamma_s}\xi_s\beta^Sb_0(X_s) ds\}\nonumber\\
	&=&\lambda_0+\int_0^t e^{\Gamma_s} \beta^C \cdot \xi_s \beta^S b_0(X_s) ds+ \beta^C \cdot \xi_t e^{\Gamma_t}\nonumber
	\end{eqnarray}
	Next given that  $\beta^C \cdot \xi_t \leq \sum_{j=1}^r |\beta_j^C|=||\beta^C||_1$, we obtain
	$$\{\mathbb{E}[Z_t^{2p}]\}^{1/(2p)} \leq \lambda_0 + ||\beta^C||_1\mathbb{E}[e^{2p\Gamma_t}]^{1/(2p)}+\Bigg\{\mathbb{E}\left[\left(\int_0^t \beta^C\cdot \xi_s e^{\Gamma_s}\beta^Sb_0(X_s) ds\right)^{2p}\right]\Bigg\}^{1/(2p)}.$$

	By Cauchy-Schwartz inequality and H\"{o}lder inequality, we have
	\begin{eqnarray}
	& &\left\{\mathbb{E}\left[\left(\int_0^t \beta^C\cdot \xi_s e^{\Gamma_s}\beta^Sb_0(X_s) ds\right)^{2p}\right]\right\}^{1/(2p)}\nonumber\\
	&\leq & \left\{\mathbb{E}\left[\left(\int_0^t e^{2\Gamma_s} ds\right)^p \left(\int_0^t \left[\beta^C \cdot \xi_s \beta^S b_0(X_s)\right]^2ds\right)^p\right]\right\}^{1/(2p)}\nonumber\\
	&\leq& \left[\mathbb{E}\left(\int_0^t e^{2\Gamma_s} ds\right)^{2p}\right]^{1/4p} \left[\mathbb{E}\left(\int_0^t [\beta^C \cdot \xi_s \beta^S b_0(X_s)]^2ds\right)^{2p}\right]^{1/4p}.\nonumber
	\end{eqnarray}
	By Holder inequality,
	$$\int_0^t e^{2\Gamma_s}ds \leq \left[\int_0^t\left(e^{2\Gamma_s}\right)^{2p}\right]^{1/(2p)} \left[\int_0^t 1 ds\right]^{1-1/{2p}}=t^{1-1/{2p}}\left(\int_0^t e^{4p\Gamma_s}ds\right)^{1/{2p}}.$$
	So, we have that
	$$\left[\mathbb{E}\left(\int_0^t e^{2\Gamma_s} ds\right)^{2p}\right]^{1/4p}\leq\left[t^{2p-1}\mathbb{E}(\int_0^t e^{4p\Gamma_s}ds)\right]^{1/{4p}}.$$
	Similarly, we get
	$$\left[\mathbb{E}\left[\left(\int_0^t \left[\beta^C \cdot \xi_s \beta^S b_0(X_s)\right]^2ds\right)^{2p}\right]\right]^{1/4p}\leq \left[t^{2p-1}\mathbb{E}\left(\int_0^t \left[||\beta^C||_1 \beta^S b_0(X_s)\right]^{4p}ds\right) \right]^{1/4p}.$$
	Therefore, we have
	\begin{eqnarray}
	& &\left\{\mathbb{E}\left[\left(\int_0^t \beta^C\cdot \xi_s e^{\Gamma_s}\beta^Sb_0(X_s) ds\right)^{2p}\right]\right\}^{1/(2p)}\nonumber\\
	&\leq& t^{1-1/{2p}} \left[\left(\int_0^t\mathbb{E} \left[e^{4p\Gamma_s}\right]ds\right)\right]^{1/{4p}} \left\{\mathbb{E}\left(\int_0^t \left[||\beta^C||_1 \beta^S b_0(X_s)\right]^{4p}ds\right)\right\}^{1/{4p}}\nonumber\\
	& = & t^{1-1/{2p}} \left[\left(\int_0^t\mathbb{E} \left[e^{4p\Gamma_s}\right]ds\right)\right]^{1/{4p}} \left\{\left(\int_0^t ||\beta^C||_1^{4p} (\beta^S)^{4p} \mathbb{E}\left(\left[b_0(X_s)\right]^{4p}\right)ds\right)\right\}^{1/{4p}},\nonumber
	\end{eqnarray}
concluding the proof of the lemma.
\end{proof}

\begin{proof}[Proof of Lemma \ref{sol}]
	The proof of this lemma will be given in several steps. Let us first discuss existence and uniqueness of the equation for $\lambda_{t}$ assuming that $b(\lambda,\alpha)$ is uniformly bounded.

	The existence and uniqueness of the solution $\lambda_t$ follows along similar lines as in chapter V.11 in \cite{RW}.  However, due to the peculiarities of the model considered here, the derivation of the bounds for the necessary norms are more complicated. Below we mention the adjustments needed for the proof of uniqueness as the adjustments needed for the proof of existence are basically the same.

For any $M>0$, let us set
\[
 b_M(\lambda,a)=b(\lambda,a), \text{ for all } |\lambda|\leq M.
\]

Let $Y^{M}$ satisfy the equation
\begin{align}
    Y_{t}^{M} &=  \int_0^{t\wedge\tau_M} e^{\Gamma_s}[ b_M(e^{-\Gamma_s}((Y_s^{M}+Z_s)\vee 0),a)]ds +\sigma \int_0^{t\wedge\tau_M} e^{\Gamma_s(1-\rho)} ((Y_s^{M}+Z_s)\vee 0)^{\rho} dWs\nonumber\\
    &+  \beta^S \int_0^{t\wedge\tau_M} \sigma_0(X_s) ((Y_s^{M}+Z_s)\vee 0) dV_s,\nonumber
 \end{align}
 where $\tau_{M}$ is the random time defined via
\begin{equation}
\tau_M=\inf\left\{t\geq 0:|e^{-\Gamma_t}((Y_{t}^{M}+Z_t)\vee 0)|>M\right\} \wedge M.\label{Eq:randomTrancationTime}
\end{equation}

It is clear that  up to time $\tau_M$, the process $Y_{t}^{M}$ will be the same as the process $Y_t$, which has $b$ in place of $b_{M}$ as its corresponding drift coefficient.

 Now, we  assume that the equation for $Y_t^{M}$ has one more solution, potentially different than $Y_t^{M}$, denoted by $Y'^{M}_{t}$, and we denote by $\tau'_M$ the corresponding random time.

Let us consider $0<\eta\ll 1$ and define the function	
\begin{equation}\label{psi}
	\psi_\eta(x)=\frac{2}{\ln \eta^{-1}} \int_0^{|x|} \left\{\int_0^y \frac{1}{z} \chi_{[\eta,\eta^{1/2}]}(z)dz\right\}dy.
\end{equation}

Notice that $\psi_\eta$ is an even function. In addition, its first and second derivatives satisfy
\begin{equation*}  \psi'_\eta(x) = \frac{2}{\ln \eta^{-1}} \int_{z=0}^x \frac{1}{z}\chi_{[\eta,\eta^{1/2}]}(z) dz \qquad \text{and}\qquad  \psi''_\eta(x) = \frac{2}{\ln \eta^{-1}} \frac{1}{x}\chi_{[\eta,\eta^{1/2}]}(x) \end{equation*}
for all $x>0$.  Monotonicity arguments then show that for all $x\in \mathbb{R}$ and $\eta> 0$, $| \psi'_\eta(x)|\leq 1$, and
\begin{equation*} |x|\le \psi_\eta(x)+\sqrt{\eta}. \end{equation*}

Additionally, we note that
\begin{equation}\left| \psi''_\eta(x)\right|\le \frac{2}{\ln \eta^{-1}}\frac{1}{|x|}\chi_{[\eta,\sqrt{\eta})}(|x|)\le \frac{2}{\ln \eta^{-1}}\min\left\{ \frac{1}{|x|},\frac{1}{\eta}\right\},\label{Eq:SecondDerivativeBound}
 \end{equation}
and that $x {\psi'}_{\eta}(x)\geq 0$ for  all $x\in \mathbb{R}$.

We have
\[
|Y_t^{M}-Y_t^{'M}| \leq \psi_{\eta}(Y_t^{M}-Y_t^{'M})+\sqrt{\eta} \leq D_t^{1,M}+\sigma^2 D_t^{2,M} + (\beta^S)^2 D_t^{3,M} + \mathcal{M}_{t}+\sqrt{\eta}
\]
where $\mathcal{M}_{t}$ is a martingale, and
\begin{align}
	D_t^{1,M}&=\int_0^{t\wedge \tau_M\wedge\tau'_M} \psi_{\eta}^{'}(Y_s^{M}-Y_s^{'M})e^{\Gamma_s}\nonumber\\
	& \left|b_M(e^{-\Gamma_s}((Y_s^{M}+Z_s)\vee 0),a) -b_M(e^{-\Gamma_s}((Y_s^{'M}+Z_s)\vee 0),a) \right|ds\nonumber\\
	&\leq \int_0^{t\wedge \tau_M\wedge\tau'_M} \psi_{\eta}^{'}(Y_s^{M}-Y_s^{'M}) C_{M,1} |Y_s^{M}-Y_s^{'M}|ds,\nonumber
\end{align}
where $C_{M,1}$ is the Lipschitz constant for  the truncated function $b_M(\cdot,a)$. Also,

\begin{align}
D_t^{2,M}&= 1/2 \int_0^{t\wedge \tau_M\wedge\tau'_M} {\psi''}_{\eta}(Y_s^{M}-Y_s^{'M})e^{2\Gamma_s(1-\rho)}\nonumber\\
    & \quad\times\left[((Y_s^{M}+Z_s)\vee 0)^{\rho}-((Y_s^{'M}+Z_s)\vee 0)^{\rho}\right]^2ds\nonumber\\
    & \leq 1/2 \int_0^{t\wedge \tau_M\wedge\tau'_M} {\psi''}_{\eta}(Y_s^{M}-Y_s^{'M})e^{2\Gamma_s(1-\rho)}\nonumber\\
	& \quad\times\left[((Y_s^{M}+Z_s)\vee 0)^{2\rho}-((Y_s^{'M}+Z_s)\vee 0)^{2\rho}\right]ds\nonumber\\
	&\leq  1/2 \int_0^{t\wedge \tau_M\wedge\tau'_M} {\psi''}_{\eta}(Y_s^{M}-Y_s^{'M})e^{2\Gamma_s}\nonumber\\
	& \quad\times \left[\left(e^{-\Gamma_s}((Y_s^{M}+Z_s)\vee 0)\right)^{2\rho}-\left(e^{-\Gamma_s}((Y_s^{'M}+Z_s)\vee 0)\right)^{2\rho}\right]ds\nonumber\\
	&\leq  1/2 \int_0^{t\wedge \tau_M\wedge\tau'_M} {\psi''}_{\eta}(Y_s^{M}-Y_s^{'M})e^{\Gamma_s} C_{M,2} |Y_s^{M}-Y_s^{'M}|ds\nonumber\\
	&\leq  \frac{C_{M,2}}{\ln \eta^{-1}}\int_0^{t\wedge \tau_M\wedge\tau'_M} e^{\Gamma_s} ds \nonumber
	\end{align}
for some constant $K_2$, where (\ref{Eq:SecondDerivativeBound}) was used. Here $C_{M,2}$ is the Lipschitz coefficient of the locally Lipschitz function $f(x)=x^{2\rho}$ for $|x|\leq M$. Similarly, using (\ref{Eq:SecondDerivativeBound}) and Assumption \ref{s0}  we can show
	\begin{align}
	D_t^{3,M}& =  1/2 \int_0^{t\wedge \tau_M\wedge\tau'_M} {\psi''}_{\eta}(Y_s^{M}-Y_s^{'M}) \sigma^{2}_{0}(X_{s})\left[((Y_s^{M}+Z_s)\vee 0)-((Y_s^{'M}+Z_s)\vee 0)\right]^{2}ds\nonumber\\
&\leq 1/2 \int_0^{t\wedge \tau_M\wedge\tau'_M} {\psi''}_{\eta}(Y_s^{M}-Y_s^{'M}) \sigma^{2}_{0}(X_{s})\left[((Y_s^{M}+Z_s)\vee 0)^{2}-((Y_s^{'M}+Z_s)\vee 0)^{2}\right]ds\nonumber\\
	&\leq  K_{3} \int_0^{t\wedge \tau_M\wedge\tau'_M} {\psi''}_{\eta}(Y_s^{M}-Y_s^{'M}) \left|Y_s^{M}-Y_s^{'M}\right|e^{\Gamma_{s}}\left|e^{-\Gamma_{s}}(Y_{s}^{M}+Z_{s})+e^{-\Gamma_{s}}(Y_{s}^{'M}+Z_{s})\right|ds\nonumber\\
	&\leq  K_{3} C_{M,3} \int_0^{t\wedge \tau_M\wedge\tau'_M} {\psi''}_{\eta}(Y_s^{M}-Y_s^{'M}) \left|Y_s^{M}-Y_s^{'M}\right|e^{\Gamma_{s}}ds\nonumber\\
&\leq  \frac{K_3 C_{M,3}}{\ln \eta^{-1}} \int_0^{t\wedge \tau_M\wedge\tau'_M} e^{\Gamma_{s}}ds\nonumber
	\end{align}

Therefore, we get that
\begin{align}
\sup_{t\leq T\wedge \tau_M\wedge\tau'_M}\mathbb{E} |Y_t^{M}-Y_t^{'M}|&\leq \sqrt{\eta} +(\sigma^2+\beta^S)^2 T\frac{K_2 C_{M,2}+K_3 C_{M,3}}{\ln \eta^{-1}}\nonumber\\
&\qquad + C_{M,1}\int_0^T \sup_{s\leq t\wedge \tau_M\wedge\tau'_M}\mathbb{E}|Y_s^{M}-Y_s^{'M}| dt.
\end{align}

By Gronwall's lemma, we obtain that	
\[
\sup_{t\leq T\wedge \tau_M\wedge\tau'_M}\mathbb{E} |Y_t^{M}-Y_t^{'M}|\leq \left(\sqrt{\eta} +(\sigma^2+\beta^S)^2 T \frac{K_2 C_{M,2}+K_3 C_{M,3}}{\ln \eta^{-1}} \right) \exp\{ C_{M,1} T\}.
\]

Let $\eta \downarrow 0$, we have for any $T>0$.
\[
\sup_{t\leq T\wedge \tau_M\wedge\tau'_M}\mathbb{E}  |Y_t^{M}-Y_t^{'M}| = 0.
\]
	
That is $Y_t^{M}=Y_t^{'M}$ for any $M \in \mathbb{N}$ and $t\leq T\wedge \tau_M\wedge\tau'_M$. Then let $M \to \infty$ and together with the observation that $\tau_M, \tau^{'}_M$ increase to infinity almost surely, which follows by Lemma \ref{tauM}, we obtain uniqueness of the solution $Y_t$ to the following SDE
	\begin{align}
	Y_t &= \int_0^{t} e^{\Gamma_s}[b(e^{-\Gamma_s}((Y_s+Z_s)\vee 0),a)]ds+ \sigma \int_0^{t} e^{\Gamma_s(1-\rho)} ((Y_s+Z_s)\vee 0)^{\rho} dWs\nonumber\\
	&+  \beta^S \int_0^{t} \sigma_0(X_s) ((Y_s+Z_s)\vee 0) dV_s\nonumber
	\end{align}

Let us set now $\bar{Y}_t = Y_t+Z_t$. Then, $\bar{Y}_t$ satisfies
\begin{align}
	\bar{Y}_t &=  Z_t + \int_0^{t} e^{\Gamma_s}[b(e^{-\Gamma_s}(\bar{Y}_s\vee 0),a)+ \sigma \int_0^{t} e^{\Gamma_s(1-\rho)} (\bar{Y}_s\vee 0)^{\rho} dWs\nonumber\\
	&+  \beta^S \int_0^{t} \sigma_0(X_s) (\bar{Y}_s\vee 0) dV_s.\nonumber
	\end{align}

It is easy to see now that $\lambda_t = e^{-\Gamma_t}\bar{Y}_t$ is the unique solution defined in the lemma \ref{sol}. Next we show that $\lambda_t\geq 0$. 	First, we notice that $\bar{Y}_0=Z_0=\lambda_0>0$.

 By It\^{o}'s formula for the function $\psi_\eta(\cdot)$
	\begin{eqnarray}
	\psi_{\eta}(\bar{Y}_t) \chi_{\mathbb{R}_{-}}(\bar{Y}_t) & = & \psi_{\eta}(\bar{Y}_0) \chi_{\mathbb{R}_{-}}(\bar{Y}_0)\nonumber\\
	& + & \int_0^{t} {\psi'}_{\eta}(\bar{Y}_s) \chi_{\mathbb{R}_{-}}(\bar{Y}_s) b(e^{-\Gamma_s}(\bar{Y}_s\vee 0),a) ds \nonumber\\
	&+& (1/2) \sigma^2 \int_0^{t} {\psi''}_{\eta}(\bar{Y}_s) \chi_{\mathbb{R}_{-}}(\bar{Y}_s) e^{2\Gamma_s(1-\rho)} (\bar{Y}_s\vee 0)^{2\rho} ds\nonumber\\
	&+& 1/2 (\beta^S)^2 \int_0^{t} {\psi'}_{\eta}(\bar{Y}_s) \chi_{\mathbb{R}_{-}}(\bar{Y}_s) \sigma_0(X_s) (\bar{Y}_s\vee 0)^{2} ds + \mathcal{M}_t\nonumber
	\end{eqnarray}
	
	where $\mathcal{M}_t$ is a martingale. Notice that for $s>0$ at least one of $\chi_{\mathbb{R}_{-}}(\bar{Y}_s)$ and $(\bar{Y}_s\vee 0)$ have to be zero, then taking expectation for both sides:
	$$\mathbb{E}[\psi_{\eta}(\bar{Y}_t) \chi_{\mathbb{R}_{-}}(\bar{Y}_t)]=\mathbb{E}[\int_0^{t} \psi'_{\eta}(\bar{Y}_s) \chi_{\mathbb{R}_{-}}(\bar{Y}_s) b(e^{-\Gamma_s}(\bar{Y}_s\vee 0),a)] ds$$
	
	Notice that $\chi_{\mathbb{R}_{-}}(\bar{Y}_s) b(e^{-\Gamma_s}(\bar{Y}_s\vee 0),a)$ can only take the nonzero value $b(0,a)>0$ when $\bar{Y}_s \leq 0$. Also notice $\psi'_{\eta}(x)$ takes non-positive values when $x \leq 0$  and is 0 when $|x|<\eta$. Thus, if we let $\eta \to 0$ the right hand side of the above equation is no greater than zero. On the left hand side, recall that as $\eta \to 0$, $\psi_{\eta}(x)$ goes to $|x|$. Therefore, letting $\eta \to 0$, we have	
\[
\mathbb{E}[\bar{Y}_t^{-}]=\mathbb{E}[|\bar{Y}_t| \chi_{\mathbb{R}_{-}}(\bar{Y}_t)] \leq 0
\]
	
Hence, we get that
\[
\mathbb{E}[\bar{Y}_t^{-}]=0,
\]
i.e., $\bar{Y}_t$ is nonnegative and as a consequence $\lambda_{t}=e^{-\Gamma_t}\bar{Y}_t$  is also nonnenative. This concludes the proof of the lemma.
\end{proof}

\begin{proof}[Proof of Lemma \ref{lambda}]
	For each $N \in \mathbb{N}$ and $n \in {1,2,\ldots,N}$ define
	$$\Gamma_t^{N,n}=-\beta^S_{N,n} \int_0^t b_0(X_s) ds$$
	$$Z_t^{N,n}=\lambda_{0,N,n}+\beta^C_{N,n} \cdot \int_0^t e^{\Gamma_s^{N,n}} dL^N_s$$
	\begin{align}
	Y_t^{N,n} &=  \int_0^t e^{\Gamma_s^{N,n}}[b(e^{-\Gamma_s^{N,n}}(	Y_s^{N,n}+Z_s^{N,n}),a_n)] ds +  \sigma^{N,n} \int_0^t e^{\Gamma_s^{N,n}(1-\rho)} (	Y_s^{N,n}+Z_s^{N,n})^{\rho} dW^n_s \nonumber\\
	&+ \beta^S_{N,n} \int_0^t \sigma_0(X_s) (	Y_s^{N,n}+Z_s^{N,n}) dV_s.\nonumber
	\end{align}
	Then $\lambda_t^{N,n}=e^{-\Gamma_t^{N,n}}(	Y_s^{N,n}+Z_t^{N,n})$. 	So, we have
	$$|\lambda_t^{N,n}|^p \leq \frac{1}{2} \left[e^{-2p\Gamma_t^{N,n}}+(Y_t^{N,n}+Z_t^{N,n})^{2p}\right]. $$
	
Hence, due to Assumption \ref{Xt}, it is enough to show that $\sup_{t\leq T}\mathbb{E}|Y_t^{N,n}+Z_t^{N,n}|^{2p}\leq K$ for some appropriate finite constant $K$.
	
Apply It\^{o}'s formula to $|Y_t^{N,n}+Z_t^{N,n}|^{2p}$. We claim that without loss of generality the martingale terms that appear in the It\^{o} formula can be considered to be true martingales and thus have zero expectation. With this in mind,  $ 1-M_t^{N,n} - \int_{0}^t \lambda_s^{N,n} M_s^{N,n} ds$ is a martingale and we write $\lambda_t^{N,n}=e^{-\Gamma_t^{N,n}}(	Y_s^{N,n}+Z_t^{N,n})$.

Then, we can write down
	\begin{align}
	&\mathbb{E}|Y_t^{N,n}+Z_t^{N,n}|^{2p} \label{a1}\\
	&= \mathbb{E}\int_0^t 2p |Y_s^{N,n}+Z_s^{N,n}|^{2p-1} e^{\Gamma^{N,n}_s}[b(e^{-\Gamma^{N,n}_s}((Y_s^{N,n}+Z_s^{N,n})\vee 0),a_n)]ds\nonumber\\
	&+ \mathbb{E}\frac{{\sigma^{N,n}}^2}{2} \int_0^t 2p(2p-1) |Y_s^N + Z_s|^{2p-2} e^{2\Gamma^{N,n}_s(1-\rho)} ((Y_s^{N,n}+Z_s^{N,n})\vee 0)^{2\rho} ds\nonumber\\
	&+ \mathbb{E}\frac{(\beta^S_{\NN})^2}{2} \int_0^t 2p(2p-1) |Y_s^{N,n}+Z_s^{N,n}|^{2p-2} (\sigma_0(X_s))^2(Y_s^{N,n}+Z_s^{N,n})\vee 0)^2 ds\nonumber\\
	&+ \mathbb{E}\int_0^t 2p |Y_s^{N,n}+Z_s^{N,n}|^{2p-1} e^{\Gamma_s^{N,n}} \beta^{C}_{N,n} \cdot \frac{1}{N}\sum_{i=1}^{N}\ l^{i}\left(e^{-\Gamma_s^{N,i}}(	Y_s^{N,i}+Z_s^{N,i}) M_s^{N,i}\right) d s,\nonumber	
\end{align}

 By	Assumption \ref{assumptionb}, we have that there is some $K>0$ such that $\lambda b(\lambda,a)\leq-\gamma(a)|\lambda|^{d}$ for $|\lambda|\geq K$. Without loss of generality, we can assume that the dissipativity condition holds everywhere (if not we just consider separately the cases $|\lambda|<K$ and $|\lambda|\geq K$). Then, we have the estimate
	\begin{eqnarray}
	&& \mathbb{E}\int_0^t 2p |Y_s^{N,n}+Z_s^{N,n}|^{2p-1} e^{\Gamma^{N,n}_s}[b(e^{-\Gamma^{N,n}_s}((Y_s^{N,n}+Z_s^{N,n})\vee 0),a_n)]ds\label{a2}\\
	&\leq& - \mathbb{E}\int_0^t 2p |Y_s^{N,n}+Z_s^{N,n}|^{2p-2} e^{2\Gamma^{N,n}_s}  \gamma (a_n) |e^{-\Gamma^{N,n}_s}(Y_s^{N,n}+Z_s^{N,n})|^d ds.\nonumber\\
	&\leq& 0\nonumber
	\end{eqnarray}

		For the second term, we have	
	\begin{eqnarray}
	&& \mathbb{E}\int_0^t e^{2\Gamma^{N,n}_s(1-\rho)}|Y_s^{N,n}+Z_s^{N,n}|^{2p-2}  ((Y_s^{N,n}+Z_s^{N,n})\vee 0)^{2\rho} ds\label{a3}\\
	&\leq& \mathbb{E}\int_0^t e^{2\Gamma^{N,n}_s(1-\rho)}|Y_s^{N,n}+Z_s^{N,n}|^{2p-2}  |Y_s^{N,n}+Z_s^{N,n}|^{2\rho} ds\nonumber\\
	&\leq& 2^{2(p-1+\rho)-1}\mathbb{E}\int_0^t e^{2\Gamma^{N,n}_s(1-\rho)} \left(|Y_s^{N,n} |^{2p-2+2\rho} +|Z_s^{N,n}|^{2p-2+2\rho} \right)  ds\nonumber\\
	&\leq& 2^{2(p-1+\rho)-1}\mathbb{E}\int_0^t \left[ \frac{p-1+\rho}{p}|Y_s^{N,n}|^{2p} + \frac{1-\rho}{p}e^{2p\Gamma^{N,n}_s}\right] ds\nonumber\\
	&& + 2^{2(p-1+\rho)-1}\mathbb{E}\int_0^t \left[\frac{1-\rho}{p} e^{2p\Gamma^{N,n}_s} + \frac{p-1+\rho}{p} |Z_s^{N,n}|^{2p} \right]ds.\nonumber
	\end{eqnarray}
	
	The third term is similar with the second term with the help of Assumption \ref{s0} on the bound for $\sigma_0$.
	\begin{eqnarray}
	&& \mathbb{E}\int_0^t |Y_s^{N,n}+Z_s^{N,n}|^{2p-2} \left(\sigma(X_s)\right)^2 ((Y_s^{N,n}+Z_s^{N,n})\vee 0)^2 ds\label{a4}\\
	&\leq& \mathbb{E}\int_0^t \left(\sigma(X_s)\right)^2 |Y_s^{N,n}+Z_s^{N,n}|^{2p} ds\nonumber\\
	&\leq& 2^{2p-1} K_{\ref{s0}}^2 \mathbb{E}\int_0^t   \left(|Y_s^{N,n}|^{2p}+|Z_s^{N,n}|^{2p}\right) ds\nonumber
	\end{eqnarray}
	
For the fourth term we apply subsequently Young's inequality, use Assumption \ref{Xt} and we get
	\begin{align}
  &\mathbb{E}\int_0^t 2p |Y_s^{N,n}+Z_s^{N,n}|^{2p-1} e^{\Gamma_s^{N,n}} \beta^{C}_{N,n} \cdot \frac{1}{N}\sum_{i=1}^{N}\ l^{i}\left(e^{-\Gamma_s^{N,i}}(	Y_s^{N,i}+Z_s^{N,i}) M_s^{N,i}\right) d s\nonumber\\
  &\leq C_{0}K_{\ref{bdd}} \mathbb{E} \frac{1}{N}\sum_{i=1}^{N}\int_0^t  |Y_s^{N,n}+Z_s^{N,n}|^{2p-1} e^{\Gamma_s^{N,n}}  \left(e^{-\Gamma_s^{N,i}}|	Y_s^{N,i}+Z_s^{N,i}| \right) d s\nonumber\\
&\leq C_{1} \left(1+ \frac{1}{N}\sum_{n=1}^{N}\mathbb{E} \int_0^t  \left[ |Y_s^{N,n}|^{2p} +  |Z_s^{N,n}|^{2p} \right]ds \right)\label{a4_2}
   \end{align}
  for appropriate constants $C_{0},C_{1}<\infty$.

	Notice now that
	\begin{eqnarray}
	|Y_t^{N,n}|^{2p} &=&  |Y_t^{N,n}+Z_t^{N,n}-Z_t^{N,n}|^{2p}\label{aa}\\
	&\leq& 2^{2p-1} |Y_t^{N,n}+Z_t^{N,n}|^{2p} + 2^{2p-1} |Z_t^{N,n}|^{2p}.\nonumber
	\end{eqnarray}

	Next step is to bound (\ref{aa}) using  (\ref{a2}), (\ref{a3}), (\ref{a4}), (\ref{a4_2}). First we average equations (\ref{a2}), (\ref{a3}), (\ref{a4}), (\ref{a4_2}) over $n\in\{1,\cdots,N\}$ and together with Assumption \ref{bdd}, Assumption \ref{s0}, Assumption \ref{Xt}, Lemma \ref{Zt}, we have that there is a constant $K$ such that
	
\[
 \frac{1}{N}\sum_{n=1}^{N}\mathbb{E}[|Y_t^{N,n}|^{2p}] \leq K + K \int_0^t  \frac{1}{N}\sum_{n=1}^{N} \mathbb{E}[|Y_s^{N,n}|^{2p}] ds.
\]
	
	By Gronwall lemma, we obtain that
	\begin{equation}
	\sup_{0\leq t \leq T}\frac{1}{N}\sum_{n=1}^{N}\mathbb{E}[|Y_t^{N,n}|^{2p}] \leq K e^{KT}.  \label{prop0}
	\end{equation}

In addition, notice that using 	(\ref{prop0}) now, (\ref{aa}) together with   (\ref{a2}), (\ref{a3}), (\ref{a4}), (\ref{a4_2}), also gives that for any $n\in\{1,\cdots, N\}$
\begin{equation}
	\sup_{0\leq t \leq T}\mathbb{E}[|Y_t^{N,n}|^{2p}] \leq K,  \label{prop1}
	\end{equation}
for an appropriate constant $K<\infty$ with the upper bounds being independent of $N$.

Together with Assumption \ref{Xt} and Lemma \ref{Zt} we can finally get from (\ref{prop0}) the bound advertised in the lemma.

It remains to address the claim on the martingale property of the stochastic integrals. Indeed, using the  same truncation argument as in the proof of Lemma \ref{sol} we get that for each fixed $M>0$ the terms in question are true martingales. Then, because the corresponding upper bound in (\ref{prop0}) turns out to be uniform with respect to $M>0$ and due to  Lemma \ref{tauM} the claim is proven, concluding the proof of the lemma.
\end{proof}

\begin{proof}[Proof of Lemma \ref{uniqueQ}]
	As in the proof of Lemma \ref{sol}, if we can prove that the result holds for the truncated processes which has $b_{M}$ in place of $b$, then, due to Lemma \ref{tauM}, the result will be true for the limit as $M\rightarrow \infty$ as well. Therefore, we can restrict attention to the case where $b(\lambda,\alpha)$ is replaced by $b_{M}(\lambda,\alpha)$ for an arbitrary constant $M<\infty$.

In addition, let $S(\mathbb{R}_{+})$ be the set of $\mathbb{R}_{+}$ valued,
adapted, continuous processes  $\{\lambda_t\}_{t\in[0,T]}$ such that
\begin{equation*}
 \left\Vert \lambda\right\Vert_{T,1}=\sup_{0\leq t\leq T}\mathbb{E}|\lambda_t|<\infty.
\end{equation*}
The space $S(\mathbb{R}_{+})$ endowed with the norm $\left\Vert \cdot\right\Vert_{T,1}$ is a Banach space.

Consider a nonnegative process $U_t(\hat{p}) \in S(\mathbb{R}_+)$ and set $\xi(U)_t=(\xi_1(U)_t, \ldots, \xi_r(U)_t)$
\[
\xi_i(U)_t= \int_{\hat{p}\in \mathcal{P}} \left(1-\mathbb{E}_{\mathcal{V}_t} \left\{\exp\left[-\int_{0}^t U_s(\hat{p})ds \right] \right\}\right) \pi(dp) \Lambda_0(d\lambda_0).
\]
	
For given $U_t(\hat{p}), U'_t(\hat{p}) \in S(\mathbb{R}_+)$, we consider $\xi_t=\xi(U)_t=(\xi_1(U)_t, \ldots, \xi_r(U)_t)$, $\xi'_t=\xi(U')_t=(\xi_1(U')_t, \ldots, \xi_r(U')_t)$. Define the map $\Phi:S(\mathbb{R}_{+})\mapsto S(\mathbb{R}_{+})$ by letting $\Phi(U)$ denoting the unique solution $\lambda=\Phi(U)$ to the SDE
\begin{align}
\lambda_{t}&= \lambda_{0}+\int_0^{t}b_M(\lambda_s,a) ds + \int_0^{t} \sigma \lambda_s^{\rho}  dW_s+  \beta^C\cdot\xi_t + \beta^S \int_0^{t} \lambda_s dX_s.\nonumber
 \end{align}

Similarly, we define $\lambda'=\Phi(U')$ for the solution of the equation with $\xi'$ in place of $\xi$. Then, the process $R_{t\wedge \tau_{M}\wedge \tau'_{M}}=\lambda_{t\wedge \tau_{M}\wedge \tau'_{M}}-\lambda'_{t\wedge \tau_{M}\wedge \tau'_{M}}$ satisfies
\begin{align}
R_{t\wedge \tau_{M}\wedge \tau'_{M}}&= \int_0^{t\wedge \tau_{M}\wedge \tau'_{M}}\left(b_M(\lambda_s,a)-b_M(\lambda'_s,a)\right) ds + \int_0^{t\wedge \tau_{M}\wedge \tau'_{M}} \sigma \left( \lambda_s^{\rho} - {\lambda'_s}^{\rho} \right) dW_s\nonumber\\
 &+ \sum_{i=1}^r \beta^C_i \int_0^{t\wedge \tau_{M}\wedge \tau'_{M}} (d\xi_s-d\xi'_s) + \beta^S \int_0^{t\wedge \tau_{M}\wedge \tau'_{M}} R_s dX_s.\nonumber
 \end{align}
	
Apply It\^{o}'s formula to $\psi_{\eta}(R_t)$ where $\psi$ is defined in Equation (\ref{psi}), and get	
\begin{align}
	&\psi_{\eta}(R_{t\wedge \tau_M\wedge\tau'_M})  = \int_0^{t\wedge \tau_M\wedge\tau'_M} \left(b_{M}(\lambda_s,a)-b_{M}(\lambda'_s,a)\right) \psi'_{\eta}(R_s) ds \nonumber\\
&\quad+ \sum_{i=1}^r \beta^C_i \int_0^{t\wedge \tau_M\wedge\tau'_M} \psi'_{\eta}(R_s) (d\xi_s-d\xi'_s)\nonumber\\
	&\quad + \frac{\sigma^2}{2} \int_0^{t\wedge \tau_M\wedge\tau'_M} \left( \lambda_s^{\rho} - {\lambda'_s}^{\rho} \right)^{2} \psi^{''}_{\eta}(R_s) ds + \sigma \int_0^{t\wedge \tau_M\wedge\tau'_M} \left( \lambda_s^{\rho} - {\lambda'_s}^{\rho} \right)\psi'_{\eta}(R_s) dW_s\nonumber\\
	& \quad+ \beta^S \int_0^{t\wedge \tau_M\wedge\tau'_M} b_0(X_s)R_s \psi'_{\eta}(R_s) ds + \beta^S \int_0^{t\wedge \tau_M\wedge\tau'_M} \sigma_0(X_s) R_s \psi'_{\eta}(R_s) dVs\nonumber\\
	& \quad+ \int_0^{t\wedge \tau_M\wedge\tau'_M} \frac{1}{2} \left(\beta^S \sigma_0(X_s)R_s\right)^2\psi^{''}_{\eta}(R_s)ds.\nonumber
	\end{align}
	
Taking expectation of $\psi_{\eta}(R_t)$ we get	
	\begin{align}
&	\mathbb{E}\psi_{\eta}(R_{t\wedge \tau_M\wedge\tau'_M})  = \mathbb{E}\int_0^{t\wedge \tau_M\wedge\tau'_M} \left(b_M(\lambda_s,a)-b_M(\lambda'_s,a)\right) \psi'_{\eta}(R_s) ds\nonumber\\
 &\quad + \sum_{i=1}^r \beta^C_i \mathbb{E}\int_0^{t\wedge \tau_M\wedge\tau'_M} \psi'_{\eta}(R_s) (d\xi_s-d\xi'_s)\nonumber\\
	&\quad+ \frac{\sigma^2}{2} \mathbb{E}\int_0^{t\wedge \tau_M\wedge\tau'_M} \left( \lambda_s^{\rho} - {\lambda'_s}^{\rho} \right)^{2} \psi^{''}_{\eta}(R_s) ds +\beta^S \mathbb{E}\int_0^{t\wedge \tau_M\wedge\tau'_M} b_0(X_s)R_s \psi'_{\eta}(R_s) ds \nonumber\\
	&\quad + \mathbb{E}\int_0^{t\wedge \tau_M\wedge\tau'_M} \frac{1}{2} \left(\beta^S \sigma_0(X_s)R_s\right)^2\psi^{''}_{\eta}(R_s)ds.\nonumber
	\end{align}
	
As in Lemma A.2. in \cite{GSSS}, the latter expression yields
\[
\mathbb{E} |\xi_t-\xi_t'| \leq K_{\ref{bdd}} t\int_{\hat{p}\in \mathcal{P}} ||U_.(\hat{p})-U_.'(\hat{p})||_t \pi(dp) \Lambda_0(d\lambda_0).
\]
	
Therefore, we have	
\[
\left|\sum_{i=1}^r \beta^C_i \mathbb{E}\int_0^{t\wedge \tau_M\wedge\tau'_M} \psi'_{\eta}(R_s) (d\xi_s-d\xi'_s)\right| \leq t C_0   K_{\ref{bdd}} \int_{\hat{p}\in \mathcal{P}} ||U_.(\hat{p})-U_.'(\hat{p})||_t \pi(dp) \Lambda_0(d\lambda_0).
\]

At the same time, we have

$$\left|\mathbb{E}\int_0^{t\wedge \tau_M\wedge\tau'_M} \left(b_M(\lambda_s,a)-b_M(\lambda'_s,a)\right) \psi'_{\eta}(R_s) ds \right| \leq C_{1,M} \int_0^{t\wedge \tau_M\wedge\tau'_M} \mathbb{E} |R_s|ds.$$

For the third term we obtain
\begin{align}
\left|\frac{\sigma^2}{2} \mathbb{E}\int_0^{t\wedge \tau_M\wedge\tau'_M} \left( \lambda_s^{\rho} - {\lambda'_s}^{\rho} \right)^2 \psi^{''}_{\eta}(R_s) ds \right|
& \leq \left|\frac{\sigma^2}{2} \mathbb{E}\int_0^{t\wedge \tau_M\wedge\tau'_M} \left|\lambda_s^{2\rho} - {\lambda'_s}^{2\rho} \right| \psi^{''}_{\eta}(R_s) ds \right|\nonumber\\
& \leq \left|\frac{\sigma^2}{2} \mathbb{E}\int_0^{t\wedge \tau_M\wedge\tau'_M} C_{2,M} |R_s| \psi^{''}_{\eta}(R_s) ds \right|\nonumber\\
& \leq C_{2,M} {K_{\ref{bdd}}}^2 \frac{2t}{\ln \eta^{-1}}\nonumber\\
& = C_2(\eta,t,M).\nonumber
\end{align}

Now, let us assume $b_0$ is bounded. Then we have
\[
\left| \beta^S \mathbb{E}\int_0^{t\wedge \tau_M\wedge\tau'_M} b_0(X_s)R_s \psi'_{\eta}(R_s) ds \right| \leq K_{\ref{bdd}} K \int_0^{t\wedge \tau_M\wedge\tau'_M} \mathbb{E} |R_s|ds.
\]

For the last term
\begin{align}
&\left|\mathbb{E}\int_0^{t\wedge \tau_M\wedge\tau'_M} \frac{1}{2}\left(\beta^S \sigma_0(X_s)R_s\right)^2\psi^{''}_{\eta}(R_s)ds \right|\nonumber\\
& \leq \frac{1}{2} (\beta^S)^2 \left|\mathbb{E}\int_0^{t\wedge \tau_M\wedge\tau'_M}  (\sigma_0 (X_s))^2 (\lambda_s^2-{\lambda_s^{'}}^2) \psi^{''}_{\eta}(R_s)ds \right| \nonumber\\
& \leq \frac{1}{2} (\beta^S)^2 \left|\mathbb{E}\int_0^{t\wedge \tau_M\wedge\tau'_M} (\sigma_0(X_s))^2 C_{3,M} |R_s| \psi^{''}_{\eta}(R_s)ds\right|\nonumber\\
& \leq \frac{1}{2} C_{3,M} {K_{\ref{bdd}}}^2 \frac{1}{\ln \eta^{-1}} \mathbb{E}\int_0^{t\wedge \tau_M\wedge\tau'_M} (\sigma_0(X_s))^2 ds\nonumber\\
& = C_3(\eta,t,M).\nonumber
\end{align}

For any $0<M<\infty$, we have that the both terms $C_2(\eta,t,M)$ and $C_3(\eta,t,M)$  go to zero as  $\eta \downarrow 0$ or $t \downarrow 0$.
	
Thus, for any $M<\infty$, we have
\begin{align}
		\mathbb{E}\psi_{\eta}(R_{t\wedge \tau_M\wedge\tau'_M})&\leq (C_1+K K_{\ref{bdd}}) \int_0^{t\wedge \tau_M\wedge\tau'_M} \mathbb{E} |R_s|ds\nonumber\\
		&+ t C_0  K_{\ref{bdd}} \int_{\hat{p}\in \hat{\mathcal{P}}} ||U_.(\hat{p})-U_.'(\hat{p})||_t \pi(dp) \Lambda_0(d\lambda_0) + C_2(\eta,t,M)+C_3(\eta,t,M).\nonumber
	\end{align}

Then applying $|x| \leq \psi_{\eta}(x)+\sqrt{\eta}$ and using Gronwall's Lemma we have
	\begin{align}
	\mathbb{E} |R_{t\wedge \tau_M\wedge\tau'_M}| \leq & \left[ t C_0  K_{\ref{bdd}} \int_{\hat{p}\in \hat{\mathcal{P}}} ||U_.(\hat{p})-U_.'(\hat{p})||_t \pi(dp) \Lambda_0(d\lambda_0) + C_2(\eta,t,M)\right.\nonumber\\
&\quad\left.+C_3(\eta,t,M) + \sqrt{\eta} \right]\cdot e^{ (C_1+K K_{\ref{bdd}})t }.\nonumber	
	\end{align}

Send $\eta\downarrow 0$ and notice that we can pick $t$ small enough such that  $C(t)=t C_0  K_{\ref{bdd}} e^{ (C_1+K K_{\ref{bdd}})t } <1$. Hence, we obtain
	$$\mathbb{E} |R_{t\wedge \tau_M\wedge\tau'_M}| \leq C(t) \int_{\hat{p} \in \hat{\mathcal{P}}} ||U(\hat{p})-U'(\hat{p})||_{t,1} \pi(dp) \Lambda_0(d\lambda_0),$$
where $C(t)<1$.

Hence, we have obtained that the map $\Phi$ defined by $\lambda=\Phi(U)$ with $U\in S(\mathbb{R}_{+})$ is a contraction on $S(\mathbb{R}_{+})$ equipped with the $L^{1}$ norm. Standard Picard iteration shows that there is a fixed point $\lambda^{*}$ such that $\lambda^{*}_{t}=\Phi_{t}(\lambda^{*})$ for $0\leq t\leq t_{1}\wedge \tau_M\wedge\tau'_M$ with $C(t_{1})<1$. This fixed point is unique, since
\begin{align}
\sup_{t\leq t_1\wedge \tau_M\wedge\tau'_M}\left|\lambda_t^*(\hat{p})-\lambda_t^{'*}(\hat{p})\right| & \leq C(t) \int_{\hat{p}} \sup_{t\leq t_1\wedge \tau_M\wedge\tau'_M}\left|\lambda_t^*(\hat{p})-\lambda_t^{'*}(\hat{p})\right| \pi(dp)\Lambda_0(d\lambda_0)\nonumber
\end{align}

So, we have that
$$\sup_{t\leq t_1\wedge \tau_M\wedge\tau'_M}\left|\lambda_t^*(\hat{p})-\lambda_t^{'*}(\hat{p})\right| =0$$

Thus, we have proven  uniqueness of $\lambda_t^*$ on $[0,t_1\wedge \tau_M\wedge\tau'_M]$. Then, starting from $t_1$ we obtain uniqueness on $[t_1 \wedge \tau_M\wedge\tau'_M, (2t_1)\wedge \tau_M\wedge\tau'_M]$ in the same way and we conclude by filling in the whole interval $[0,T\wedge \tau_M\wedge\tau'_M]$.

Next, letting $M \to \infty$, and using Lemma \ref{tauM} which implies that $\tau_M, \tau'_M$ converge to infinity almost surely, we have the proof of the lemma for bounded $b_0$.

For the case of general $b_0$, Assumption \ref{unbdd_b} guarantees that
\[
M_T=e^{-\int_0^T u(X_s) dV_s - 1/2 \int_0^T |u(X_s)|^2 ds},
\]
is a martingale by Novikov's condition. Assumption \ref{unbdd_b} also assumes $\mathbb{E}|M_T|^p < \infty$. Then the result follows from the proof of Lemma A.6. in \cite{GSSS}. \end{proof}

\begin{lem}\label{tauM}
		For any $T>0$ and for $\tau_{M}$ defined via (\ref{Eq:randomTrancationTime}), we have that $$\lim_{M\to \infty}\mathbb{P} [\tau_M <T]=0.$$
	\end{lem}

\begin{proof}[Proof of Lemma \ref{tauM}]
	
	For any $T>0$,
	$$\mathbb{P}[\tau_M <T] \leq \frac{1}{M} \mathbb{E} \left[\sup_{t \wedge \tau_M <T} |e^{-\Gamma_t} ((Y_t^{M} + Z_t)\vee 0)|\right].$$
	
	Due to Assumption \ref{Xt} and Lemma \ref{Zt}, it is enough to prove that
	$$\mathbb{E} \sup_{t \wedge \tau_M \leq T} |Y_t^{M}+ Z_t|^{2} \leq \tilde{K}$$
	where $\tilde{K}$ is independent of $M$.	Now $\sup_{t \leq T} |Y_t^M+ Z_t|^{2}$ can be estimated similarly as before. Indeed, applying It\^{o}'s formula to $|Y_t^M+ Z_t|^{2}$, we get
	\begin{align}
		  |Y_t^M+Z_t|^2&= \lambda_0 + \int_0^t 2 |Y_s^M+Z_s| \ \ e^{\Gamma_s} [b(e^{-\Gamma_s}((Y_s^M+Z_s)\vee0),a)] ds\nonumber\\
		& \quad + \frac{1}{2} \sigma^2 \int_0^t 2 e^{2\Gamma_s (1-\rho)} ((Y_s^M+Z_s)\vee0)^{2\rho} ds\nonumber\\
		& \quad + \frac{1}{2} (\beta^S)^2 \int_0^t 2 (\sigma_0(X_s))^2 ((Y_s^M+Z_s)\vee0)^2 ds\nonumber\\
		& \quad + \sigma \int_0^t 2 |Y_s^M+Z_s| \ \ e^{\Gamma_s (1-\rho)} ((Y_s^M+Z_s)\vee0)^{\rho} dW_s\nonumber\\
		& \quad + \beta^S \int_0^t 2 |Y_s^M+Z_s| \ \ \sigma_0(X_s) ((Y_s^M+Z_s)\vee0) dV_s\nonumber\\
		& \quad + \int_0^t 2 |Y_s^M+Z_s| \ \ e^{\Gamma_s} \beta^C \cdot d {\xi}_s\nonumber
	\end{align}
	
	The first line in the right hand side of the expression above is bounded due to Assumption \ref{assumptionb} (similarly to (\ref{a2}) we can assume without loss of generality that the dissipativity condition holds everywhere) and we have
\begin{align}
		& \lambda_0 + \int_0^t 2 |Y_s^M+Z_s| \ \ e^{\Gamma_s} [b(e^{-\Gamma_s}((Y_s^M+Z_s)\vee0),a)] ds\nonumber\\
		& \qquad \leq \lambda_0 - \int_0^t 2 \ \ e^{2\Gamma_s} \gamma(a) |e^{-\Gamma_s}(Y_s^M+Z_s)|^d ds \leq \lambda_0\nonumber
	\end{align}
	
	Therefore, if we square both sides in the It\^{o}'s formula expression, we will get
	\begin{align}
		\quad |Y_t^M+Z_t|^4
		& \leq 6 {\lambda_0}^2 + 6 \sigma^4 \left[\int_0^t e^{2\Gamma_s (1-\rho)} |Y_s^M+Z_s|^{2\rho} ds \right]^2\nonumber\\
		& \quad + 6 (\beta^S)^4 \left[\int_0^t (\sigma_0(X_s))^2 |Y_s^M+Z_s|^2 ds \right]^2\nonumber\\
		& \quad + 24 \sigma^2 \left[\int_0^t |Y_s^M+Z_s| \ \ e^{\Gamma_s (1-\rho)} \ \ |Y_s^M+Z_s|^{\rho} dW_s\right]^2\nonumber\\
		& \quad + 24 (\beta^S)^2 \left[\int_0^t |Y_s^M+Z_s| \ \ \sigma_0(X_s) \ \ |Y_s^M+Z_s| dV_s\right]^2\nonumber\\
		& \quad + 24 \left[\int_0^t |Y_s^M+Z_s| \ \ e^{\Gamma_s} \beta^C \cdot d {\xi}_s \right]^2\nonumber
	\end{align}
	
	Taking expectation of the supremum of the second term and using H\"{o}lder inequality, together with the fact that $\rho <1$ and Assumption \ref{Xt} we have
	\begin{align}
		 \mathbb{E} \sup_{t \wedge \tau_M \leq T} \left(\int_0^t e^{2\Gamma_s (1-\rho)} |(Y_s^M+Z_s)|^{2\rho} ds \right)^2
		& \leq \mathbb{E} \left[\left(\int_0^T e^{2\Gamma_s (1-\rho)} |Y_s^M+Z_s|^{2\rho} ds \right)^2\right]\nonumber\\
		& \leq \mathbb{E} \left(\int_0^T |Y_s^M+Z_s|^{4\rho} ds \int_0^T e^{4\Gamma_s (1-\rho)} ds \right)\nonumber\\
		& \leq c_{p_1} \mathbb{E} \int_0^T \sup_{u \wedge \tau_M \leq s} |Y_u^M+Z_u|^4 ds,\nonumber
	\end{align}
		for some constant $c_{p_1}>0$. Similar calculations together with Assumption \ref{s0}, gives a similar bound for the third term as well. 	Using Burkholder-Davis-Gundy inequality for the fourth term, together with Young's inequality, the fact that $\rho <1$ and Assumption \ref{Xt} we get
	\begin{align}
		& \quad \mathbb{E} \left[ \sup_{t \wedge \tau_M \leq T} \left[\int_0^t |Y_s^M+Z_s| \ \ e^{\Gamma_s (1-\rho)} \ \ |Y_s^M+Z_s|^{\rho} dW_s\right]^2 \right]\nonumber\\
		& \leq c_{p_2} \mathbb{E} \int_0^T |Y_s^M+Z_s|^{2(1+\rho)} \ \ e^{2\Gamma_s (1-\rho)} ds \leq c_{p_3} + c_{p_4} \mathbb{E} \int_0^T \sup_{u \wedge \tau_M \leq s}  |Y_u^M+Z_u|^4 ds\nonumber
	\end{align}
	
	where $c_{p_2}$, $c_{p_3}$ and $c_{p_4}$ are some positive constants. The stochastic integral term with respect to the $V-$ Brownian motion is treated analogously using Assumption \ref{s0}.	For the last term, we use Young's inequality and Assumptions \ref{bdd} and \ref{s0}. We obtain	
	\begin{align}
		& \quad \mathbb{E} \sup_{t \wedge \tau_M \leq T} \left( 24 \int_0^t |Y_s^M+Z_s| \ \ e^{\Gamma_s} \beta^C \cdot d {\xi}_s \right)^2\nonumber\\
		& \leq c_{p_5} \ \ \mathbb{E} \left[\left( \sup_{t \wedge \tau_M \leq T} \int_0^t |Y_s^M+Z_s| \ \ e^{\Gamma_s} \beta^C \cdot d {\xi}_s \right)^2\right]\nonumber\\
		& \leq c_{p_5} \ \ \mathbb{E} \left[\left( \sup_{t \wedge \tau_M \leq T} |Y_t^M+Z_t| \sup_{t \wedge \tau_M <T} \int_0^t e^{\Gamma_s} \beta^C \cdot d {\xi}_s\right)^2\right]\nonumber\\
		& \leq c_{p_5} \ \ \mathbb{E} \left[ \frac{\epsilon}{2} \sup_{t \wedge \tau_M \leq T} |Y_t^M+Z_t|^4 + \frac{1}{2\epsilon} [\sup_{t \wedge \tau_M \leq T} \int_0^t e^{\Gamma_s} \beta^C \cdot d {\xi}_s]^4\right]\nonumber\\
		& \leq c_{p_5}\epsilon \ \ \mathbb{E} \sup_{t \wedge \tau_M \leq T} |Y_t^M+Z_t|^4 + c_{p_{\epsilon}}\nonumber
	\end{align}
	
	for any $\epsilon>0$ and correspondent constant $c_{p_{\epsilon}}>0$. Therefore we can choose $\epsilon$ small enough so $c_{p_5} \epsilon <1$ and we can move this term to the left hand side.	

Thus, combining all the terms together with Assumption \ref{bdd} leads to the estimate	
	\begin{align*}
		\mathbb{E} \sup_{t \wedge \tau_M \leq T} |Y_t^M+Z_t|^{4}& \leq  K\left(1+ \int_0^T  \mathbb{E} \sup_{u \wedge \tau_M \leq s} |Y_u^M+Z_u|^{4} ds\right).
	\end{align*}
	
	Then by Gronwall lemma, the term $\mathbb{E} \sup_{t \wedge \tau_M \leq T} |Y_t^M+Z_t|^{4}$ is bounded by a constant which is independent of $M$ and we conclude by Fatou's lemma followed Jensen's inequality.	
\end{proof}

\end{document}